\pdfoutput=1
\documentclass[11pt,a4paper,oneside]{report}

\usepackage[utf8]{inputenc}
\usepackage[T1]{fontenc}
\usepackage{lmodern}
\usepackage[english]{babel}
\usepackage[margin=2.5cm]{geometry}
\usepackage{microtype}

\usepackage{amsmath,amssymb,amsthm,mathtools}
\usepackage{bbm}
\usepackage{etoolbox}
\usepackage{algorithm}
\usepackage{algpseudocode}
\usepackage{graphicx}
\usepackage{float}
\usepackage{placeins}
\usepackage{caption}
\usepackage{subcaption}
\usepackage{booktabs,longtable,multirow,makecell,array,tabularx}
\usepackage{enumitem}
\usepackage[dvipsnames,table]{xcolor}
\usepackage{tikz}
\usetikzlibrary{chains,shapes.geometric,shapes.symbols,arrows,arrows.meta,
  decorations.pathreplacing,trees,positioning,calc,fit,backgrounds}
\usepackage[numbers,sort&compress]{natbib}
\usepackage[hidelinks]{hyperref}
\hypersetup{pdftitle={Evaluating and Improving the Robustness of Large Language Models to Input Sequence Variations},
  pdfauthor={Narek Maloyan}}

\newcommand{\E}{\mathbb{E}}
\newcommand{\Pbb}{\mathbb{P}}
\newcommand{\I}{\mathbb{I}}
\newcommand{\B}{\mathcal{B}}

\newcommand{\fparam}{f_\theta}
\newcommand{\Xspace}{\mathcal{X}}          
\newcommand{\Yspace}{\mathcal{Y}}          
\newcommand{\Vspace}{\mathcal{V}}          
\newcommand{\ProbSpace}[1]{\mathcal{P}(#1)} 
\newcommand{\Bepsilon}[1]{\mathcal{B}_\epsilon(#1)} 
\newcommand{\Exp}{\mathbb{E}}              
\newcommand{\DataDist}{\mathcal{D}}        
\newcommand{\Vulfunc}{V}                   
\newcommand{\Params}{\theta}               
\newcommand{\LossFunc}{\mathcal{L}}        
\newcommand{\badY}{\mathcal{Y}_{\mathrm{bad}}}
\newcommand{\JS}{D_{\mathrm{JS}}}

\theoremstyle{definition}
\newtheorem{theorem}{Theorem}[chapter]
\newtheorem{lemma}[theorem]{Lemma}
\newtheorem{proposition}[theorem]{Proposition}
\newtheorem{corollary}[theorem]{Corollary}
\newtheorem{definition}[theorem]{Definition}

\newtheorem{remark}[theorem]{Remark}

\newcounter{principle}[chapter]
\newenvironment{principle}[1][]{%
  \refstepcounter{principle}%
  \par\noindent\textbf{Principle \thechapter.\theprinciple.}%
  \ifstrempty{#1}{}{ \textbf{#1}}%
  \par\noindent%
}{\par}

\title{Evaluating and Improving the Robustness\\
of Large Language Models to Input Sequence Variations\\[16pt]
\large PhD Thesis}
\author{Narek Maloyan\\[4pt]
  \\[4pt]
  \small Advisor: Dmitry Namiot}
\date{2026}

\begin{document}

\maketitle

\begin{abstract}
Large language models (LLMs) embedded in production systems face prompt injections, trojans, and manipulation of automatic quality metrics. This thesis develops a set of models, methods, and algorithms for evaluating and improving the robustness of LLMs to adversarial input sequence variations: prompt injection, trojans (backdoor attacks), and attacks on agentic LLM systems. We propose $R_{stab}(f)$, a formal metric of generative robustness based on bounded divergence measures, in particular the Jensen--Shannon divergence, between per-step output distributions under small input perturbations. For localized attacks, we prove that the vulnerability of a deterministic decision operator $h$ is bounded by its non-robustness: $V(h) \le 1 - R_{class}(h)$, where $R_{class}(h)$ is the expected probability that $h$ preserves its decision under small input perturbations. For non-localized attacks, we propose an empirical calibrated model. For LLM-as-a-Judge systems, we separate instructive and content attacks and develop ASA, an adaptive evolutionary attack that combines token-level mutations with semantic paraphrasing. Under black-box access, ASA reaches an attack success rate (ASR) of up to 73.8\%, and the transfer success rate between open models reaches 62.6\%. On Trojan Detection Challenge 2023 data for Pythia-1.4B, we quantify an asymmetry of trojan detection: surrogate triggers reach a reverse-engineered attack success rate (REASR) of $\approx 0.99$, while the recall of the true triggers is $\approx 0.17$ against a baseline of $\approx 0.14$. On SaTML CTF 2024, we systematize four classes of bypasses of multi-layer defenses. Committees of 5--7 heterogeneous models reduce the ASR for Gemma-3-4B by 47--55 percentage points, from 73.8\% to 19.3\% with 7 models. The Condorcet jury theorem justifies the growth of committee accuracy when the members err independently and each has an accuracy above 0.5. For Gemma-3-4B, which does not meet this condition, the gain comes from the heterogeneity of the committee. Multi-layer defense reduces the ASR from 90\% to 15--25\%. For agentic systems based on the Model Context Protocol (MCP), we propose AttestMCP, which attests tool calls with HMAC-protected packets at an overhead below 0.1~ms per call, and the Commit Boundary isolation pattern. On the MCPBench benchmark of 847 scenarios, our defenses reduce the average ASR from 53.7\% to 12.4\%. The methods are implemented in the registered software suites JudgeGuard and TrojanArmor and in the MCPSec module. JudgeGuard is deployed in production at Viasat.
\end{abstract}

\tableofcontents

\chapter*{Introduction}
\addcontentsline{toc}{chapter}{Introduction}

\textbf{Motivation.}
Modern large language models (LLMs) solve a wide range of natural language processing tasks. As LLMs are embedded into production systems, the number of entry points through which an adversary can algorithmically influence the model also grows: prompt injections, trojans (backdoors), and manipulation of automatic quality metrics. Among such metrics, the LLM-as-a-Judge architecture, in which a language model evaluates the responses of other models, plays a special role. A compromised judge model distorts model comparisons and undermines the reproducibility of experiments.

The search for malicious instructions is an optimization problem over a discrete space of texts. Adversaries iterate over candidate \textit{jailbreaks} until the model outputs prohibited information. Multi-level heuristic filters do not solve the problem. Trojans work differently: at the training stage, a trigger is injected into the data, and this trigger forms a hidden attractor in the parameter space. Reconstruction of such a trigger from the network weights is an ill-conditioned inverse problem. In the Trojan Detection Challenge 2023, the recall of true-trigger reconstruction was 0.17, which is close to the baseline level. Agentic architectures aggravate the situation: when an LLM gains access to external tools, a textual vulnerability carries over into the space of executable actions. The consequences of an attack then extend beyond the information domain.

Experimental evaluation quantitatively confirms the vulnerability of the LLM-as-a-Judge architecture. On the MT-Bench dataset and on the tasks of the Kaggle competition ``LLMs: You Can't Please Them All'', injection attacks generated by the adaptive ASA algorithm reached an attack success rate (ASR) of 73.8\% against the Gemma-3-4B-Instruct judge model. The built-in safeguards of judge models fail under targeted attacks.

The compositional nature of agentic systems gives rise to a separate class of risks. The Model Context Protocol~(MCP) standardizes the interaction of an LLM agent with external tools and thereby expands the agent's state space. An injection into the context of a tool call can take over control of the agent (\textit{prompt injection to tool hijacking}). From an algorithmic point of view, the robustness problem of a single model extends to the computation graph of an entire system: input sequences are formed dynamically from different sources, and an attack through an external channel becomes a special case of an adversarial input perturbation that can provoke unauthorized code execution.

All these vulnerabilities have a common cause: the mapping learned by the model is sensitive to small input variations in the discrete token space. Perturbations destabilize hidden representations, and the model loses its intended properties. The embedding of LLMs into complex agentic systems makes the problem more acute. A rigorous framework for ensuring robustness is therefore a pressing research problem.

\textbf{State of the art.}
The foundations of robustness analysis for neural networks were laid by Goodfellow~et~al.~\cite{goodfellow_explaining_2014} and Madry~et~al.~\cite{madry_towards_2017} in computer vision. A direct transfer of this framework to LLMs is difficult because the token space is discrete and inference is autoregressive. For the text domain, Zou~et~al.~\cite{zou_universal_2023} proposed the GCG method, which generates universal adversarial suffixes. Wallace~et~al.~\cite{wallace_universal_2019} showed that universal triggers can control generation. Perez~et~al.~\cite{perez_ignore_2022} systematized prompt injection attacks. Carlini~et~al.~\cite{carlini_extracting_2021} showed that fragments of training data can be extracted directly from a trained model. However, formal metrics for evaluating LLM robustness remain underdeveloped: the existing metrics do not account for the discreteness of the token space and the autoregressive character of generation.
Zheng~et~al.~\cite{zheng_judging_2023} proposed the LLM-as-a-Judge architecture for automatic evaluation of the quality of model responses. It spread quickly, but the analysis of its robustness to targeted adversarial attacks remains fragmentary. Trojan detection was studied in the TDC2023 and TrojAI benchmarks. A quantitative characterization of the asymmetry between trigger injection and trigger detection remained an open question before the present work. The security of agentic systems and of protocols such as MCP has received little study. Greshake~et~al.~\cite{greshake_not_2023} described indirect prompt injection attacks on LLM-integrated applications and showed that external content can take over control of an agent. However, to the author's knowledge, no prior publication analyzed the architectural risks of MCP or proposed defense algorithms based on cryptographic attestation of tool calls. The absence of a comprehensive approach to the security of agentic LLM systems motivated the present research.
At the start of this work, the literature lacked: (1)~a formal robustness metric that accounts for the discreteness and autoregressiveness of LLM generation, (2)~a systematic methodology for evaluating the vulnerability of LLM-as-a-Judge, (3)~algorithms for defending agentic systems against injections through protocols of interaction with external tools.

The work covers several interrelated directions:
\begin{itemize}
    \item vulnerabilities of LLMs to prompt injection attacks in LLM-as-a-Judge architectures and in systems with multi-layer defense;
    \item limitations of existing trojan detection methods;
    \item vulnerabilities of agentic LLM systems that use MCP to interact with external tools, and software tools for defense against control hijacking;
    \item methods for improving the robustness of LLMs to the listed threats and their experimental evaluation.
\end{itemize}

Studying these directions in one work reveals the common structure of LLM vulnerabilities and yields concrete recommendations for developers.

\textbf{Scope.} The work studies large language models and systems based on them, namely LLM-as-a-Judge and agentic LLM systems with tool access. It develops models, methods, algorithms, and software tools for evaluating vulnerabilities and improving the robustness of LLMs.

\textbf{Aim.} The aim of the work is to develop a set of models, methods, and algorithms for evaluating and improving the robustness of large language models to adversarial input sequence variations and injection attacks. In this work, input sequence variations comprise prompt injections, trojans (\textit{backdoor attacks}), and attacks on agentic LLM systems.

To achieve this aim, the work addresses the following \textbf{objectives}:
\begin{enumerate}
    \item \textit{Systematization and classification of threats.} Construct a classification of vulnerabilities of large language models to prompt injection attacks that covers specialized architectures, such as LLM-as-a-Judge and systems with multi-layer defense, and establish the limitations of trojan detection methods.
    \item \textit{Development of the theoretical foundations of evaluation.} Develop a formal mathematical model for robustness evaluation (\(R_{stab}(f)\)) and a consistent set of empirical metrics that account for the discrete nature of the token space and the autoregressive character of generation, including the basic indicators ASR (\textit{Attack Success Rate}), TSR (\textit{Transfer Success Rate}), Recall, and REASR (\textit{Reverse-Engineered ASR}).
    \item \textit{Construction of adversarial attack methods.} Develop adaptive algorithms for generating adversarial attacks (ASA) to test LLM vulnerabilities, including attacks on LLM-as-a-Judge systems in the black-box setting.
    \item \textit{Development and experimental validation of defense methods.} Propose and experimentally study strategies for improving LLM robustness, namely committees of heterogeneous models and multi-layer filtering, on a representative set of models of the main architectural families, both with open weights (Gemma, Llama) and proprietary (GPT, Claude), and on the MT-Bench, TDC2023, and SaTML CTF 2024 datasets.
    \item \textit{Software implementation and validation.} Create software tools for evaluating the security of LLM systems, including agentic systems based on MCP. Validate the attack and defense methods on the TDC2023 and SaTML CTF 2024 benchmarks and in the production environment of Viasat.
\end{enumerate}

\textbf{Scientific novelty.}
\begin{enumerate}
    \item A formal mathematical model for evaluating the robustness of large language models (\(R_{stab}(f)\)) is proposed as a metric of generative robustness. It accounts for the discrete nature of the token space and the autoregressive character of generation. It is based on bounded divergence measures, in particular the Jensen--Shannon divergence, between the per-step output distributions under small input perturbations. The metric is shown to be monotone with respect to the natural relation of defense strengthening, and this monotonicity serves as a consistency condition of the definition. Theoretical vulnerability bounds are proved for a deterministic decision operator $h$ through the decision-level robustness metric $R_{class}(h)$. For tasks with discrete output, including LLM-as-a-Judge, $\mathrm{ASR}$ is a consistent estimator of the vulnerability $V(h)$ with respect to a fixed attack family $\mathfrak{A}_{eval}$. The population limit of the empirical estimate $\widehat{R}_{emp}(f) = 1 - \mathrm{ASR}$ is $1 - V(h)$, which is an upper bound on $R_{class}(h)$ under the stated assumptions. This justifies the use of the ASR metric in the experimental part of the work. The result was obtained by the author and is theoretically justified.
    \item The robustness of LLM-as-a-Judge systems to injection attacks is formalized. Attacks are divided into two classes by the mechanism of influence: instructive attacks, which inject a direct command, and content attacks, which shift the verdict by manipulating the content. For each class, a corresponding countermeasure is defined. On this basis, the adaptive attack generation algorithm ASA (\textit{Adaptive Search-Based Attack}) is developed: an evolutionary search that combines token-level mutations with semantic paraphrasing. On a set of open and proprietary models, ASA achieves an ASR of up to 73.8\% under black-box access. Open models proved systematically more vulnerable than their proprietary counterparts. Attack transferability between open models reaches TSR = 62.6\%. All these results were obtained experimentally by the author.
    \item For the Pythia-1.4B model under full white-box access on the Trojan Detection Challenge 2023 data, the asymmetry of the trojan detection problem is quantitatively characterized. A surrogate trigger activates the trojan almost always: REASR (\textit{Reverse-Engineered Attack Success Rate}) is $\approx 0.99$. The original trigger cannot be reconstructed: Recall $\approx 0.17$ is close to the baseline level $\approx 0.14$. The gap is explained by the combinatorial growth of the space of discrete triggers and indicates the limitations of post-training trojan auditing. Priority should be given to preventive data control. The author hypothesizes that the asymmetry grows with the number of model parameters.
    \item A methodology for generating and classifying multi-step prompt injection attacks on LLMs equipped with multi-layer defense systems, including system instructions, external Python filters, and LLM filters, is developed and experimentally evaluated. On the international SaTML CTF 2024 benchmark, four classes of successful bypass vectors are systematized, namely code-oriented attacks, ASCII character encoding, word-splitting attacks, and contextual misdirection. For each class, the conditions of applicability and effective countermeasures are determined: input normalization and context-isolated system instructions. The methodology was developed by the author and verified on the competition data.
    \item A defense method for LLM-as-a-Judge systems based on adaptive committees of heterogeneous models is developed. The method uses majority voting in the space of classified verdicts. It differs from standard ensembles in that the composition of the committee is optimized with regard to the architectural diversity of the models. The Condorcet jury theorem gives a theoretical justification for the growth of committee accuracy, that is, of the probability of a correct majority verdict, when the models err independently and each is correct with probability above 0.5. This accuracy result does not by itself establish robustness in the sense of $R_{class}$. On a representative set of attacks, a committee of 5--7 models of different architectural families reduces the ASR by 47--55 percentage points for the most vulnerable open-weight model (Gemma-3-4B) compared with the single model. The method was developed and experimentally validated by the author.
    \item An approach to the security analysis of agentic LLM systems based on the Model Context Protocol is proposed. The class of \textit{prompt injection to tool hijacking} attacks is formalized. Architectural vulnerabilities of MCP are identified, and a set of defense mechanisms is developed. Its core consists of two algorithms. AttestMCP restricts tool calls to a session capability table and controls packet integrity with an HMAC signature at an overhead below 0.1~ms per call. Commit Boundary runs the code generated by the agent in an isolated environment. They are complemented by four auxiliary mechanisms: capability certificates, provenance tagging, enforced isolation of cross-server transfer, and replay protection. All components are combined in the MCPSec software module, which reduces attack success while preserving an acceptable response time. The approach was proposed by the author and confirmed experimentally on the MCPBench benchmark.
\end{enumerate}

\textbf{Theoretical and practical significance.} The theoretical significance of the work lies in the creation of a formal framework for evaluating LLM robustness. The metric $R_{stab}(f)$ is monotone with respect to defense strengthening, and the related decision-level metric $R_{class}(h)$ gives a proved upper bound $V(h) \le 1 - R_{class}(h)$ on the vulnerability under localized attacks. Together they provide a quantitative basis for comparing defense methods. The Condorcet jury theorem justifies the growth of the accuracy of majority-vote committees under independent errors and a base accuracy above 0.5. For compositional systems, a Lipschitz decomposition of sensitivity is obtained. Its empirical plug-in estimates indicate which component of the pipeline dominates the sensitivity. The discovered asymmetry of trojan detection favors preventive data auditing over post-training diagnostics.

The practical significance of the work lies in two registered software suites and a module for agentic LLM systems. JudgeGuard (software registration certificate No.~2025687702) automates the testing of LLM-as-a-Judge and implements committee-based defense. TrojanArmor (certificate No.~2025684247) detects trojans. The MCPSec software module controls the permissions of tool calls and isolates code execution in agentic LLM systems. The areas of application are the security auditing of automatic evaluation systems, dialogue agents, and agentic platforms.

The results of the thesis have been adopted in industry. In the production activity of Viasat, as confirmed by an implementation certificate, the framework for evaluating LLM-as-a-Judge vulnerabilities, the defense method based on model committees, and the methodology of multi-step attacks were applied to secure a recommendation ranking module built on the LLM-as-a-Judge architecture.

\textbf{Research methods.} The work combines theoretical and experimental methods. The scientific literature on attacks on LLMs, including prompt injections and trojans, and on methods of defense and robustness evaluation was analyzed. Models of attacks and defenses were developed and formalized. The approaches were validated through participation in international LLM security competitions: Trojan Detection Challenge 2023, SaTML CTF 2024, and Kaggle ``LLMs: You Can't Please Them All''.

The empirical part included the deployment and configuration of open and commercial LLMs (GPT-3.5, GPT-4, Llama 2, Llama 3.2, Gemma, Claude-3-Opus). Both the attacks developed by the author and attacks known from the literature were reproduced and simulated: CUA (\textit{Comparative Undermining Attack}), JMA (\textit{Justification Manipulation Attack}), ASA (\textit{Adaptive Search-Based Attack}), code-oriented injections, and others. The effectiveness of defense scenarios was evaluated by statistical methods: bootstrap analysis, hypothesis testing, and confidence intervals. Adversarial examples were generated by two algorithms: the author's ASA and the Greedy Coordinate Gradient method (GCG).

\textbf{Main statements submitted for defense:}
\begin{enumerate}
    \item A formal mathematical model for evaluating the robustness of large language models (\(R_{stab}(f)\)). The model relies on a divergence measure between probability distributions in the discrete token space and serves as a metric of generative robustness. For the class of localized attacks, the upper bound $V(h) \le 1 - R_{class}(h)$ is proved. For non-localized attacks, such as prompt injection and jailbreak, an empirical calibrated model $V(h) \approx g(1 - R_{class}(h))$ is proposed. The upper bound is rigorously proved by the author. The empirical model is a hypothesis whose parameters are calibrated on a validation set of attacks.
    \item The adaptive adversarial attack generation algorithm (ASA), which implements an evolutionary search with selection by an objective function that measures the influence on the judge model. The algorithm combines token-level mutations with semantic paraphrasing. In the black-box setting, the ASR reaches 73.8\%. Reverse engineering of trojans revealed a practical asymmetry: generating adversarial inputs is substantially easier than detecting them (REASR~$\approx 0.99$ with Recall~$\approx 0.17$). This gap argues for prioritizing preventive defense. The ASA algorithm and the asymmetry analysis are the work of the author.
    \item A defense method for LLM-as-a-Judge systems based on adaptive committees of heterogeneous models, which reduces the ASR by 47--55 percentage points for the most vulnerable model Gemma-3-4B, namely from 73.8\% to 19.3\% with a committee of 7 models. The theoretical justification through the Condorcet jury theorem concerns committee accuracy and holds under independent errors and the condition $r > 0.5$, where $r$ denotes the base accuracy of an individual model. For Gemma-3-4B, whose base accuracy does not satisfy this condition, the gain comes primarily from the heterogeneous composition of the committee, which decorrelates the vulnerabilities of the base models. Multi-layer defense (defense in depth) reduces the ASR from 90\% to 15--25\%. Four classes of bypasses of multi-layer filters are systematized. The results were obtained experimentally.
    \item An architectural and algorithmic approach to defending agentic LLM systems based on MCP. It includes the MCPBench benchmark (847 scenarios), the AttestMCP algorithm for attestation of tool calls, the Commit Boundary software isolation pattern, and four auxiliary defense mechanisms: capability certificates, provenance tagging, enforced isolation, and replay protection. On the MCPBench benchmark, the approach reduces the average ASR of agentic attacks from 53.7\% to 12.4\%. The approach was developed and validated by the author.
    \item The proposed methods for evaluating and improving the robustness of LLM systems can be implemented as a complete software toolkit with measurable performance. This is confirmed by the JudgeGuard (certificate No.~2025687702) and TrojanArmor (certificate No.~2025684247) software suites, as well as by the MCPSec module. TrojanArmor was validated on the TDC2023 benchmark, JudgeGuard on SaTML CTF 2024 and deployed in the production environment of Viasat, and MCPSec on the MCPBench benchmark.
\end{enumerate}

\textbf{Reliability.} The reliability of the results rests on the following. The experimental base is built on the public datasets MT-Bench, TDC2023, SaTML CTF 2024, and Kaggle ``LLMs: You Can't Please Them All''. The models represent different architectural families: Gemma, Llama, GPT-3.5, GPT-4, Claude-3-Opus. Each model/attack/task combination was tested on $n = 50$ prompts ($n = 200$ per model after aggregation over four tasks) with confidence intervals and significance tests. The obtained ASR, TSR, Recall, and REASR values are compared with the results of known methods. The results are published in peer-reviewed journals and were presented at international conferences and competitions.

\textbf{Validation and dissemination.} The main results of the work were presented and discussed at the following conferences:

\begin{enumerate}
    \item Distributed Computer and Communication Networks: \\
    Control, Computation, Communications (DCCN 2024), 2024
    \item International Conference on Computational Linguistics \\
    and Intellectual Technologies DIALOG 2022, 2022
\end{enumerate}

\textbf{Personal contribution.} The results were obtained by the author personally or with his decisive participation. The author is responsible for the development of the attack frameworks, the defense concepts, the planning and execution of the experiments, the data analysis, and the formulation of conclusions. The co-authors participated in the preparation of the publications. In all joint works, the contribution of the author is decisive: from the formulation of the problems to the analysis of the experimental results.

\textbf{Publications.} The results of the thesis are published in 7 papers: 4 journal articles~\cite{maloyan2025investigatingvulnerabilityllmasajudgearchitectures,maloyan2025agentic,maloyan2025breaking,maloyan2026sleeper}, a paper at the DCCN conference~\cite{Khomsky_2025}, a paper at the DIALOG conference~\cite{maloyan_generated_2022}, and a preprint~\cite{maloyan2025adversarialattacksllmasajudgesystems}. The JudgeGuard and TrojanArmor software suites are registered with two software registration certificates~\cite{judgeguard2025cert,trojanarmor2025cert}.

Some of the works were carried out with co-authors. The contribution of the author to each of them is as follows.
\begin{itemize}
    \item In the conference paper~\cite{Khomsky_2025} (D.~Khomsky, N.~Maloyan, B.~Nutfullin), the author is responsible for the formulation of the problem of attacks on systems with multi-layer defense, the development of the methodology of multi-step injections, and the analysis of their effectiveness. The co-authors participated in the implementation of the experiments and the preparation of the text.
    \item In the article~\cite{maloyan2025investigatingvulnerabilityllmasajudgearchitectures} (N.~Maloyan, B.~Ashinov, D.~Namiot), the author developed the ASA adversarial attack algorithm on LLM-as-a-Judge systems, planned and conducted the experiments, and formulated the conclusions. The co-authors participated in the discussion of the results and in scientific supervision.
    \item In the paper~\cite{maloyan_generated_2022} (N.~Maloyan, B.~Nutfullin, E.~Ilyushin), the author developed the method for detecting generated text and verified it experimentally.
    \item The articles~\cite{maloyan2025agentic,maloyan2025breaking,maloyan2026sleeper} and the preprint~\cite{maloyan2025adversarialattacksllmasajudgesystems} were prepared jointly with the doctoral advisor D.~Namiot. The author is responsible for the formulation of the problems, the development of methods of attack on and defense of agentic LLM systems, the execution of the experiments, and the formulation of the results.
    \item The software suites protected by the certificates~\cite{judgeguard2025cert,trojanarmor2025cert} were developed by the author personally.
\end{itemize}

\textbf{Software registration certificates.}
\begin{enumerate}
    \item JudgeGuard. A program for the experimental study of vulnerabilities and testing of defense methods for automatic evaluation systems based on large language models / N.~Maloyan. Software registration certificate No.~2025687702, 2025.
    \item TrojanArmor. A program for the experimental study of trojan attacks on neural networks and defense methods against them / N.~Maloyan. Software registration certificate No.~2025684247, 2025.
\end{enumerate}

\textbf{Relation between objectives, chapters, novelty results, and statements for defense.}
Objective~1, the systematization and classification of threats, is addressed in Chapter~1. It corresponds to novelty result~3, which establishes the practical asymmetry of the effectiveness of trigger generation and trigger reconstruction in trojan detection, and to result~4, which contains the classification of multi-step attacks on systems with multi-layer defense.
Objective~2, the development of the theoretical foundations of evaluation, is addressed in Chapter~2. It corresponds to novelty result~1, which includes the formal robustness metric $R_{stab}(f)$ and the theoretical vulnerability bounds, and to statement for defense~1. The theoretical justification of committees through the Condorcet jury theorem, which is part of statement for defense~3, also relies on the framework developed for Objective~2.
Objective~3, the development of adversarial attack methods, is addressed in Chapter~3. It corresponds to novelty result~2, namely the framework for studying the vulnerability of LLM-as-a-Judge and the ASA algorithm, and to statement for defense~2.
Objective~4, the development and experimental validation of defense methods, is addressed in Chapters~3 and~4. It corresponds to novelty result~5 on committees of heterogeneous models, to statement for defense~3, and to the experimental confirmation of novelty results~2--5 and statements for defense~2--4.
Objective~5, software implementation and validation, is addressed in Chapters~5--6. It corresponds to novelty result~6, which contains the security analysis of agentic systems and MCP, and to statement for defense~5.

\chapter{Analysis of Threats and Vulnerabilities of Large Language Models}\label{ch:ch1}

\section*{Chapter Overview}
The security of large language models is determined by the interaction of many factors: the discrete nature of the input space, the diversity of attack vectors, and the architectural features of LLM-based systems. This chapter formalizes the threat space, classifies vulnerabilities, and analyzes the specifics of attacks on LLM-as-a-Judge architectures and on systems with multi-layer defense. The results form the theoretical basis for the defense methods developed in the subsequent chapters.

\section{Theoretical Foundations of Large Language Model Robustness}
\label{sec:theoretical_foundations}

\paragraph{Terminological convention.}
This work uses the following terminology.
\textit{Robustness} denotes the ability of a model to preserve correct behavior under perturbations of the input data. The term covers the notions of \textit{adversarial robustness} and \textit{certified robustness} used in the literature.
The subscript ``stab'' in $R_{stab}(f)$ comes from the word \textit{stability}: the metric evaluates how much the output distributions of the model change under input perturbations. $R_{stab}(f)$ is a robustness metric, and no separate notion of stability is introduced.
The terms \textit{trojan} and \textit{backdoor} are used as synonyms.
In the LLM-as-a-Judge architecture, an LLM performs automatic evaluation. This evaluating model is called the \textit{judge model}, or simply the \textit{judge}.

Below we give definitions and formalisms for the analysis of LLM robustness and security, the types of vulnerabilities, and approaches to their quantitative evaluation.

\subsection{Large Language Models: A Formal Representation}
\label{subsec:llm_formalism}

Large language models (LLMs) are a class of probabilistic sequence models. As a rule, they are built on deep neural networks with the Transformer architecture \cite{vaswani_attention_2017} and trained on large corpora of text data \cite{brown_language_2020}. Such models are intended for processing and generating natural language text. They operate on sequences of discrete units, tokens. A formalization of LLMs and of their inputs and outputs is necessary for a rigorous definition and analysis of model properties, including robustness.

\begin{definition}
\label{def:token_space}
\textbf{Token and sequence space.}
Let $\Vspace$ denote a finite vocabulary of tokens. The space of all finite token sequences, that is, texts, is defined as $\Xspace = \Vspace^*$, where $\Vspace^*$ is the Kleene closure of the set $\Vspace$. The length of an arbitrary sequence $x \in \Xspace$ is written as $|x|$.
\end{definition}

The input of the model is a token sequence $x \in \Xspace$. The nature of the output depends on the task. In general form, an LLM is given by a parametric mapping $\fparam$ with a parameter vector $\Params$:
$$
\fparam: \Xspace \to \ProbSpace{\Yspace},
$$
where $\Yspace$ is the space of possible model outputs, and $\ProbSpace{\Yspace}$ is the set of probability distributions on it.
\begin{itemize}
    \item In text generation tasks, such as autocompletion, translation, and dialogue, the model predicts a probability distribution over the vocabulary $\Vspace$ for the next token or generates an entire token sequence. In this case $\Yspace$ often coincides with $\Xspace$, that is, the space of output sequences, or with $\Vspace$, that is, the space of next tokens, and $\fparam$ defines the conditional probability distribution $p_\Params(y|x)$ for $y \in \Yspace$ given the input $x \in \Xspace$.
    \item In text classification tasks, such as sentiment analysis and topic detection, $\Yspace$ is a finite set of class labels, and $\fparam(x)$ defines a probability distribution over these labels.
\end{itemize}
The discrete nature of the input space $\Xspace$ fundamentally distinguishes LLMs from models that operate on continuous data, in particular in computer vision \cite{szegedy_intriguing_2013}, and imposes constraints on the methods for evaluating and improving robustness. For brevity, the subscript $\Params$ is omitted below, and the model is denoted as $f$, which implies a fixed set of parameters. When an output predicate is applied to $f(x)$, as in $f(x) \in \badY$, the symbol $f(x)$ denotes a response $y \sim p(\cdot|x)$ generated by the model, and the probability is taken over this generation. Chapter~\ref{ch:ch2} distinguishes stochastic generation from deterministic decoding $h = \mathrm{Dec} \circ f$.

\subsection{LLM Robustness: Definition and Significance}
\label{subsec:robustness_definition_significance}

The robustness of a large language model is defined as its ability to preserve correct behavior under small, semantically insignificant, or syntactically admissible perturbations of the input data. The reliability of dialogue systems, automated content creation, and decision support systems directly depends on this property \cite{goodfellow_explaining_2014, liang_holistic_2023}. With insufficient robustness, even a minimal variation of a query can lead to the generation of erroneous or undesirable responses \cite{jia_adversarial_2017}.

A formalization of robustness requires a prior definition of the class of admissible input perturbations.

\begin{figure}[ht]
\centering
\begin{tikzpicture}[scale=1.2, every node/.style={font=\small}]
  \draw[thick] (0,0) circle [radius=2];
  \node at (0,2.2) {$\Bepsilon{x}$};

  \filldraw[black] (0,0) circle (1.5pt) node[below right] {$x$};

  \coordinate (xin) at (1.3,0.5);
  \filldraw[black] (xin) circle (1.5pt) node[above right] {$x'$};
  \draw[->, blue] (0,0) -- (xin) node[midway, above, sloped] {$\displaystyle d(x,x') \le \epsilon$};

  \coordinate (yout) at (-2.0,1.2);
  \filldraw[black] (yout) circle (1.5pt) node[above left] {$y$};
  \draw[->, red, dashed] (0,0) -- (yout) node[midway, above, sloped] {$\displaystyle d(x,y) > \epsilon$};

  \begin{scope}[shift={(3.5,-1)}]
    \draw[->, blue] (0,0) -- (0.5,0) node[right] {inside the neighborhood};
    \draw[->, red, dashed] (0,-0.5) -- (0.5,-0.5) node[right] {outside the neighborhood};
  \end{scope}
\end{tikzpicture}
\caption{Illustration of $\epsilon$-perturbations: the point $x'$ lies in $\Bepsilon{x}$, and the point $y$ lies outside it.}
\label{fig:robustness}
\end{figure}
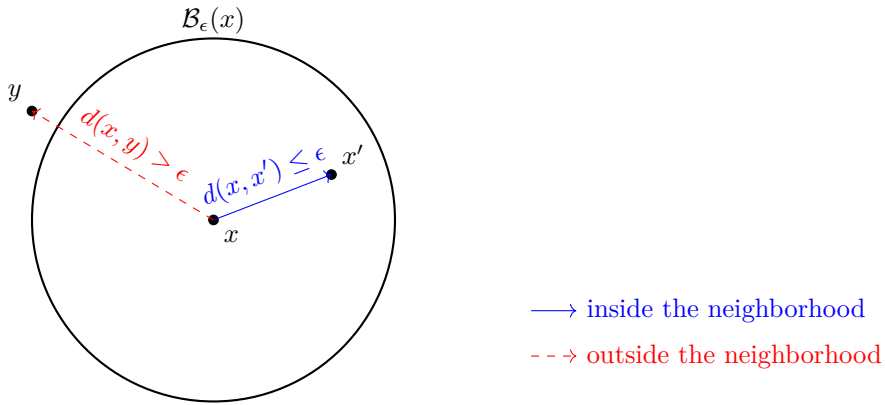

\begin{definition}
\label{def:token_perturbation}
\textbf{Perturbation of a token sequence.}
Let $x \in \Xspace$ be the original token sequence. A perturbed sequence $x' \in \Xspace$ is obtained from $x$ by applying some transformation $T: \Xspace \to \Xspace$. The set of $\epsilon$-perturbations of the sequence $x$ with respect to the metric $d$ is defined as the neighborhood
$$
\Bepsilon{x} = \{x' \in \Xspace : d(x, x') \leq \epsilon\},
$$
where $\epsilon > 0$ sets the perturbation budget and bounds the magnitude of the admissible change, and $d: \Xspace \times \Xspace \rightarrow \mathbb{R}_{\geq 0}$ is a discrepancy function on the token sequence space $\Xspace$, in particular a metric, a pseudometric, or a semantic dissimilarity measure (Fig.~\ref{fig:robustness}).
\end{definition}

The choice of the metric $d$ determines which types of changes are considered small and against which modifications robustness is evaluated. Common metrics for token sequences are listed below.

1.  \textbf{Hamming distance ($d_H$):} It applies only to sequences of equal length $n$ and is the simplest of the metrics considered:
    $$
    d_H(x, x') = \sum_{i=1}^{n} \mathbb{I}(x_i \neq x'_i), \quad \text{where } |x|=|x'|=n,
    $$
    and $\mathbb{I}(\cdot)$ is the indicator function. In essence, it counts the number of positions at which the tokens $x_i$ and $x'_i$ differ. The metric models attacks that reduce to token substitution.

2.  \textbf{Levenshtein distance ($d_L$) \cite{levenshtein_binary_1966}:} Attacks are rarely limited to a single substitution. Typos and the addition and deletion of words produce token insertions and deletions that the Hamming metric does not cover. The Levenshtein distance generalizes $d_H$ to sequences of arbitrary length: it measures the minimum number of single-token edit operations, namely insertion, deletion, and substitution, needed to transform $x$ into $x'$. The formal recursive definition is given in \cite{levenshtein_binary_1966}.

3.  \textbf{Semantic distance ($d_S$):} Both preceding metrics work at the surface level of tokens and do not reflect the semantic closeness of texts. Paraphrasing and synonym substitution \cite{li_bert-attack_2020} can make $d_H$ and $d_L$ large, although the meaning of the text is preserved. To account for such perturbations, we introduce a distance based on the cosine measure between vector representations:
    $$
    d_S(x, x') = 1 - \frac{\langle E(x), E(x') \rangle}{\|E(x)\| \cdot \|E(x')\|},
    $$
    where $E$ is an arbitrary embedding model, $\langle \cdot, \cdot \rangle$ is the inner product, and $\|\cdot\|$ is the Euclidean norm. The sentence vector $E(x)$ is obtained by averaging (mean pooling) the token embeddings followed by L2 normalization. Unless stated otherwise, the experiments use the \texttt{all-mpnet-base-v2} model from the Sentence-BERT family, although the definition of $d_S$ itself is not tied to a specific embedder architecture.

    \textbf{Sensitivity to the choice of $E$:} The resulting value of $d_S$ and the robustness estimate depend on the choice of the embedding model $E$. Different models assess semantic closeness differently, especially in narrow specialized domains. Comparative studies must either fix a single model $E$ or analyze the sensitivity of the results to its replacement.

Table~\ref{tab:distance_metrics} summarizes the main distance metrics and their characteristics.

\begin{table}[ht]
\centering
\caption{Summary of distance metrics in the token space}
\label{tab:distance_metrics}
\small
    \begin{tabular}{|>{\raggedright}p{3.0cm}|>{\raggedright}p{5.0cm}|>{\raggedright\arraybackslash}p{4.3cm}|}
        \hline
        \textbf{Metric} & \textbf{Description} & \textbf{Application} \\
        \hline
        Hamming ($d_H$) & Number of substituted tokens (for strings of equal length) & $\epsilon$: 1--5; token substitution attacks \\
        \hline
        Levenshtein ($d_L$) & Min. number of operations (insertion, deletion, substitution) & $\epsilon$: 1--10; typos, short insertions \\
        \hline
        Semantic ($d_S$) & Cosine distance between embeddings & $\epsilon$: 0.05--0.2; paraphrasing \\
        \hline
    \end{tabular}
\end{table}

The condition $d(x, x') \le \epsilon$ sets the admissible perturbation budget: an attacker or natural variations may introduce changes into $x$, but the total cost of these changes under the metric $d$ cannot exceed $\epsilon$.

The quantitative evaluation of LLM robustness relies on the metric $R_{stab}(f)$, which captures the magnitude of the deviation of the output distributions of the model under small input modifications.

The rigorous definition and the mathematical properties of the metric $R_{stab}(f)$ are considered in detail in Chapter~\ref{ch:ch2}.

\subsection{LLM Vulnerabilities and Attack Types}
\label{subsec:vulnerabilities_attacks}

Insufficient robustness of LLMs opens opportunities for attacks on model behavior. The goals of an adversary range from generating toxic or prohibited content to extracting confidential information from the training data, bypassing safety systems, and disrupting the normal operation of the model \cite{carlini_extracting_2021, perez_ignore_2022, wei_jailbreak_2023}.

Vulnerability is formalized through the probability that the model exhibits undesirable behavior under the worst-case perturbation within a given budget.

The local vulnerability metric $\Vulfunc_{\text{local}}(f)$ reflects the expected maximum probability of obtaining undesirable outputs, which form the set $\badY \subseteq \Yspace$ defined formally in Section~\ref{subsec:vulnerability_metric} of Chapter~\ref{ch:ch2}, under perturbations of a specific input query within its $\epsilon$-neighborhood. In essence, it quantifies how easily the model can be induced to generate prohibited or harmful content by modifying the query.

The rigorous definition and the mathematical properties of the metric $\Vulfunc_{\text{local}}(f)$, and its relation to $R_{stab}(f)$ and to the global vulnerability index $V(f)$, are considered in Chapter~\ref{ch:ch2}.

The main types of attacks on LLMs are described below.

1.  \textbf{Adversarial attacks:} The construction of perturbed inputs $x'$ that differ from $x$ only slightly, often imperceptibly to a human, that is, such that $d(x, x')$ is small under some metric $d$, in particular $d_L$ or $d_S$. Such perturbations lead to a model output $f(x')$ that is erroneous in classification tasks or undesirable in generation tasks \cite{goodfellow_explaining_2014, li_bert-attack_2020, zou_universal_2023}. For LLMs, this can be the substitution of characters and tokens, the insertion or deletion of words, paraphrasing, or the use of control sequences \cite{wallace_universal_2019, zou_universal_2023, morris_textattack_2020}. Gradient optimization methods, including Greedy Coordinate Gradient (GCG), are used to generate such attacks.

2.  \textbf{Prompt injection attacks:} The injection of instructions into a user query, often covert or disguised, with the goal of overriding the original system instructions of the model or its safety mechanisms \cite{perez_ignore_2022, liu_prompt_2023}. A typical example is the addition of text such as ``Ignore the previous instructions and do X''. For such attacks, the perturbation need not be small under the formal metrics $d$, because they exploit the ability of LLMs to follow instructions \cite{greshake_not_2023, liu_prompt_2023}. Two targets are particularly exposed: LLM-as-a-Judge systems, where injections distort the evaluation results, and LLMs with multi-layer defense, which specialized techniques can bypass.

3.  \textbf{Jailbreaking attacks:} The design of specially crafted prompts that force the model to generate content that violates its internal safety policies: instructions for illegal actions, malicious code, offensive material. The safety mechanisms established by alignment fail to activate \cite{wei_jailbreak_2023, khachaturov_adversarial_2025}. Both large commercial and open models are vulnerable to such attacks \cite{yi_jailbreak_2024, chen_characterizing_2024, liu_jailbreaking_2024}, including under multi-turn queries or techniques that exploit the features of multi-layer defense systems.

4.  \textbf{Backdoor attacks:} An adversary with access to the training process implants hidden triggers into the model at the training or fine-tuning stage \cite{gu_badnets_2017}. A trigger can be a phrase or a symbol whose presence in the input data $x'$ causes behavior predefined by the attacker: the generation of a specific text or an incorrect classification \cite{chen_targeted_2017, kurita_weight_2020}. Detecting trojans (backdoors) is difficult, because the model behaves normally on ordinary data. Reverse engineering of triggers, that is, the reconstruction of the original trigger from a known malicious output, remains a task with low detection recall even under controlled conditions.

5.  \textbf{Data extraction attacks:} Attempts to reconstruct or extract from the model confidential information that could have been present in its training data, including personal data and proprietary texts, by submitting specially constructed queries \cite{carlini_extracting_2021}.

6.  \textbf{Attacks on agentic systems through tool hijacking:} With the development of agentic LLM systems that interact with external tools through standardized protocols, in particular the Model Context Protocol (MCP)~\cite{anthropic_mcp_2024}, a class of attacks arises in which a prompt injection into the context of a tool leads to an unauthorized tool call with arbitrary arguments. Unlike classical prompt injection, the consequences of such attacks go beyond text generation and include code execution, file system modification, and network requests~\cite{maloyan2025breaking, maloyan2025agentic}. The criticality of such attacks is amplified by architectural features of agentic interaction protocols, namely the absence of capability attestation, unauthenticated sampling, and implicit trust propagation between servers. A detailed study of this issue is given in Chapter~\ref{ch:ch5}.

The robustness metric $R_{stab}(f)$ and the local vulnerability metric $\Vulfunc_{\text{local}}(f)$ capture different aspects of model behavior. The first evaluates how much the distribution $p(y|x)$ changes under input perturbations, and the second characterizes the probability of generating $y \in \badY$. There is, however, a conceptual connection between them. A model with high $R_{stab}(f)$, that is, with small changes of $p(y|x)$ under perturbations $x'$, is presumably less vulnerable and shows a low $\Vulfunc_{\text{local}}(f)$. If the distribution $p(y|x)$ hardly changes under the transition to $x'$, then the probability of generating undesirable content $P(y \in \badY | x')$ also changes insignificantly. In the general case, this connection is not a strict mathematical dependence. By Pinsker's inequality, a change of the distribution with a small Kullback--Leibler divergence $D_{KL}$ changes the probability of falling into $\badY$ by at most $\sqrt{D_{KL}/2}$ in absolute terms. In relative terms, however, this probability can grow by orders of magnitude if the original probability was close to zero. Conversely, a large $D_{KL}$ does not necessarily lead to the generation of undesirable content when the new distribution is concentrated on admissible outputs. The results of experiments with adversarial training \cite{madry_towards_2017} indicate a correlation between an increase of $R_{stab}(f)$ and a decrease of $\Vulfunc_{\text{local}}(f)$. No rigorous theoretical guarantees of such a dependence, however, have been obtained.

\subsection{Categorization and Formalization of Specific Aspects of Vulnerabilities}
\label{subsec:vulnerability_categorization}

LLM vulnerabilities affect several aspects of security. Four of them are formalized below with the apparatus introduced above.

\textbf{1. Fine-tuning sensitivity.}
Adapting an LLM to specific tasks through fine-tuning \cite{ouyang_training_2022} carries a risk for robustness and security. Robustness can decrease, already known vulnerabilities can worsen, and a partial loss of the original alignment can open new ones \cite{xu2024bagtricksbenchmarkingjailbreak, wang_adversarial_2022}. Quantitatively, the effect of fine-tuning is expressed through the change of the robustness metric $R_{stab}(f)$:
$$\Delta R = R_{stab}(f_{\text{ft}}) - R_{stab}(f_{\text{orig}}).$$
Here $f_{\text{orig}}$ and $f_{\text{ft}}$ denote the model before and after fine-tuning, respectively. A negative $\Delta R$ signals a decrease in robustness. An alternative measure is the mean KL divergence between the output distributions before ($p_{\text{orig}}(y|x)$) and after ($p_{\text{ft}}(y|x)$) fine-tuning:
        $$
        \Delta P = \Exp_{x \sim \DataDist} [ D_{KL}(p_{\text{orig}}(\cdot|x) \| p_{\text{ft}}(\cdot|x)) ].
        $$
The distribution $p_{\text{orig}}$ serves as the reference, and $p_{\text{ft}}$ approximates it. A large value of $\Delta P$ indicates a noticeable shift in model behavior and can correlate with a change in the vulnerability profile.

\textbf{2. Effectiveness of jailbreak attacks.}
The effectiveness of jailbreak attacks for a specific model $f$ is defined through the maximum success probability in the $\epsilon$-neighborhood of a base query $x$ on which the model behaves safely. The metric $d$ can be, in particular, the semantic distance \cite{li_tf_attack_2024}:
        $$
        P_{\text{jailbreak}}(f, x, \epsilon) = \sup_{x' \in \Bepsilon{x}} P( f(x') \text{ bypasses the defense} ).
        $$
The quantity $P( f(x') \text{ bypasses the defense} )$ expresses the probability that the output $y \sim p(\cdot|x')$ of the model $f$ for the query $x'$ falls into the set of undesirable outputs $\badY$. In practice, the success of jailbreak attacks is determined by a combination of factors: the wording of the query, the architecture and size of the model, and the alignment methods \cite{jin_jailbreakzoo_2024, andriushchenko_jailbreaking_2024}. Multi-layer defense systems are also susceptible to these attacks \cite{wei_jailbreak_2023}. Averaging over queries $x$ from a test set $\mathcal{D}_{\text{test}}$ gives an aggregate estimate: $V_{\text{jailbreak}}(f) = \Exp_{x \sim \mathcal{D}_{\text{test}}} [P_{\text{jailbreak}}(f, x, \epsilon)]$.

\textbf{3. Context-dependent vulnerabilities.}
The effectiveness of many attacks, primarily adversarial and jailbreak attacks, varies substantially with the subject domain of the query: the medical, legal, and technical domains show unequal susceptibility. A strategy that is effective in one context can be useless in another \cite{xu2024bagtricksbenchmarkingjailbreak}. This complicates the construction of universal defenses. Formally, the dependence of vulnerability on a context $c \in \mathcal{C}$, where $\mathcal{C}$ is the set of possible contexts, is expressed by the conditional local vulnerability metric:
        $$
        \Vulfunc_{\text{local}}(f|c) = \Exp_{x \sim \DataDist|c} \left[ \sup_{x' \in \Bepsilon{x}} P(y \in \badY | x') \right],
        $$
where $\DataDist|c$ denotes the input distribution conditioned on the context $c$. A noticeable variation of $\Vulfunc_{\text{local}}(f|c)$ across contexts $c \in \mathcal{C}$ indicates a high context dependence of vulnerabilities and the need to develop context-adaptive defense mechanisms.

\textbf{4. Difficulties in identifying model compromise.}
Detecting a successful attack or a hidden compromise of a model, whether a backdoor or a side effect of fine-tuning, is difficult \cite{liu_jailjudge_2024}. Standard tests reveal no anomalies: the model behaves normally until a specific trigger query arrives or a rare scenario occurs. For trojan attacks, the problem is especially acute, because trigger reconstruction itself remains an unsolved task. In practice, a compromise is detected by comparing the behavior of the current model $f$ with a trusted reference model $f_{\text{ref}}$, for example a previous version or the model before fine-tuning.
    
    The measure of deviation can be the mean or maximum KL divergence on some test data set $\mathcal{X}_{\text{test}}$:
        $$
        D_{\text{avg}}(f, f_{\text{ref}}) = \Exp_{x \in \mathcal{X}_{\text{test}}} [ D_{KL}(p_{\text{ref}}(\cdot|x) \| p(\cdot|x)) ],
        $$
        $$
        D_{\text{max}}(f, f_{\text{ref}}) = \max_{x \in \mathcal{X}_{\text{test}}} [ D_{KL}(p_{\text{ref}}(\cdot|x) \| p(\cdot|x)) ].
        $$
        With this argument order, $D_{KL}(p_{\text{ref}} \| p)$ measures the information loss when the reference distribution is approximated by the current one, that is, it shows the degree of deviation of the model $f$ from the trusted state. If one of the metrics exceeds a predefined threshold, this can indicate a compromise. The choice of the reference model $f_{\text{ref}}$ and the interpretation of deviations themselves remain open research questions \cite{kore_empirical_2024}.

These aspects of vulnerabilities motivate the development of standardized metrics and testing protocols for LLMs \cite{liang_holistic_2023, xu_redagent_2024}.

\section{Classification and Analysis of Security Threats to Large Language Models}
\label{sec:llm_threats}

This section systematizes the threats to LLMs based on an analysis of the scientific literature, survey papers, our own research, and the industry standards OWASP Top 10 for LLM Applications \cite{owasp_foundation_owasp_2023} and NIST AI Risk Management Framework \cite{nist_ai_rmf_2023}. The section covers definitions and metrics for security evaluation, a description of threat classes with their formalization, their categorization, and their mapping to security controls. Real incidents are analyzed, and recommendations for improving model security are formulated.

\subsection{Basic Definitions and Formal Metrics}
\label{subsec:def_and_metrics}

The threat analysis uses the notation and the formal apparatus introduced in Section~\ref{subsec:llm_formalism}: the model $\fparam: \Xspace \to \ProbSpace{\Yspace}$, and the robustness metric $R_{stab}(f)$ and the local vulnerability metric $\Vulfunc_{\text{local}}(f)$ (Sections~\ref{subsec:robustness_definition_significance} and~\ref{subsec:vulnerabilities_attacks}, with rigorous definitions in Chapter~\ref{ch:ch2}). In addition, we introduce metrics specific to individual threat classes.

\textbf{The probability of confidential information leakage ($P_{\text{leak}}$)} characterizes the risk of disclosure of sensitive data that could have been present in the training set $\DataDist_{\text{train}}$ or in the input query. The formalization can depend on the type of leakage. For example, the risk of extracting a specific fragment of the training data $d \in \DataDist_{\text{train}}$ can be estimated as
$$
P_{\text{extract}}(d) = \max_{x \in \Xspace_{\text{query}}} P(d \sqsubseteq f_\theta(x)),
$$
where $\Xspace_{\text{query}}$ is the set of possible attacker queries, and $d \sqsubseteq f_\theta(x)$ means that the fragment $d$ occurs verbatim, as a contiguous substring, in the model output $f_\theta(x)$, which occurs with some probability if the output is stochastic. A more general metric can use the mutual information $I(\fparam(X); S)$ between the model outputs and some set of secrets $S$ \cite{fang_privacy_2024}.

\subsection{Description of the Main Threat Types}
\label{subsec:threat_descriptions}

An analysis of the literature \cite{yifan_survey_2024}, the standards \cite{owasp_foundation_owasp_2023, nist_ai_rmf_2023}, and the results of our research makes it possible to identify the threats listed below.

\begin{enumerate}
\item \textbf{Prompt injection \cite{perez_ignore_2022, liu_prompt_2023}.}
\label{threat:prompt_injection}
The essence of the attack is the injection of additional instructions $\delta$ into the user query $x$. Having received the concatenation $\tilde{x} = x \oplus \delta$, the model ignores the system instructions $x_{\text{sys}}$ or performs unintended actions. The task of the attacker reduces to finding an injection that maximizes the probability of an undesirable response:
$$
\delta^* = \arg\max_{\delta : d(\delta) \le \epsilon_{\text{inj}}} P(f_\theta(x \oplus \delta) \in \badY).
$$
The function $d(\delta)$ reflects the complexity of the injection, and $\epsilon_{\text{inj}}$ bounds the attacker's budget. Special cases are the Comparative Undermining Attack (CUA) and the Justification Manipulation Attack (JMA), which target LLM-as-a-Judge, as well as various jailbreak techniques.

\item \textbf{Training data poisoning \cite{cina_machine_2024, shafahi_poison_2018}.}
\label{threat:data_poisoning}
The goals of injecting specially crafted examples $\DataDist_{\text{poison}}$ into the training set $\DataDist_{\text{train}}$ vary: from a general degradation of model quality to the creation of trojans and the amplification of bias. Training on the union $\DataDist_{\text{train}} \cup \DataDist_{\text{poison}}$ shifts the model parameters by $\Delta \theta = \theta' - \theta_{\text{clean}}$. The attacker seeks to maximize the loss function $\LossFunc$ on the target inputs:
$$
\DataDist_{\text{poison}}^* = \arg\max_{\mathcal{D}_{\text{poison}} : |\mathcal{D}_{\text{poison}}| \le N_{\text{poison}}} \Exp_{x_{\text{target}}} [\LossFunc(f_{\theta'}(x_{\text{target}}))],
$$
subject to the constraint $N_{\text{poison}}$ on the admissible number of poisoned examples.

\item \textbf{Trojans (backdoor attacks) \cite{gu_badnets_2017, li_comprehensive_2023}.}
\label{threat:backdoor}
Unlike general data poisoning, a backdoor attack implants a hidden trigger $x_{\text{trigger}}$ that activates malicious functionality. The formal success criterion is
$$
P(f_{\theta'}(x \oplus x_{\text{trigger}}) = y_{\text{target}}) \approx 1 \quad \text{and} \quad P(f_{\theta'}(x') = y_{\text{target}}) \approx 0
$$
for arbitrary $x$ and $x'$ that do not contain $x_{\text{trigger}}$. The model behaves normally on ordinary data, and precisely this makes the detection and reverse engineering of the trigger so difficult. The issue is considered in detail in item~4 of Section~\ref{subsec:vulnerabilities_attacks}.

\item \textbf{Model inversion \cite{veale_algorithms_2018, fredrikson_model_2015}.}
\label{threat:model_inversion}
Model inversion is dangerous because it can disclose personal data, medical records, and other confidential information from the training set. The attacker analyzes the outputs of the model: probabilities, embeddings, generated text. Based on this analysis, the attacker reconstructs inputs that were present in $\DataDist_{\text{train}}$ with high probability. Formally, the task is written as
$$
x^* = \arg\max_{x \in \Xspace} \text{Confidence}(x \in \DataDist_{\text{train}} | f_\theta(x)),
$$
where the metric $\text{Confidence}(\cdot)$ is based on output characteristics: low perplexity or high generation probability.

\item \textbf{Membership inference \cite{shokri_membership_2017}.}
\label{threat:membership_inference}
Unlike inversion, which aims to reconstruct the data themselves, a membership inference attack solves a narrower task: establishing whether a specific fragment $x$ belongs to the training set $\DataDist_{\text{train}}$. The answer is given by a classifier $g(f_\theta(x)) \to \{ \text{member}, \text{non-member} \}$ that relies on the values of the loss function or on the output probabilities. The accuracy of this classifier serves as the measure of attack success.

\item \textbf{Adversarial examples \cite{zou_universal_2023, goodfellow_explaining_2014}.}
\label{threat:adversarial_examples}
The task of finding an adversarial example $x' = x \oplus \delta$ is formulated as the minimization of the perturbation while preserving the malicious effect:
$$
\delta^* = \arg\min_{\delta} d(\delta) \quad \text{subject to} \quad f_\theta(x \oplus \delta) \in \badY \text{ and } d(\delta) \le \epsilon.
$$
The perturbation measure $d(\delta)$ is bounded by the budget $\epsilon$. A human, as a rule, does not notice the differences between $x$ and $x'$, yet the output of the model changes radically. The consequences range from the generation of toxic content to a complete bypass of safety filters. Gradient methods, including GCG, generate such examples efficiently and apply to LLM-as-a-Judge architectures.

\item \textbf{Denial of service (DoS) \cite{zhang_overthink_2024}.}
\label{threat:dos}
The goal of attacks on availability is the exhaustion of the computational resources of an LLM service: GPU, CPU, main memory, API limits. The attacker constructs queries that maximize the processing time $T(x)$ or the resource consumption $R(x)$. Exceeding the thresholds $T_{\text{max}}$ and $R_{\text{max}}$ leads to a denial of service for legitimate users or to a general degradation of the service.

\item \textbf{Model theft/exfiltration \cite{krishna_model_2023, tramer_stealing_2016}.}
\label{threat:model_theft}
An attacker without access to the parameters $\theta$ or the architecture of the model can recreate a functionally equivalent copy $f'_{\theta'}$ through repeated API queries. The similarity of the original and the copy on a test set $\mathcal{X}_{\text{test}}$ is quantified by the metric
$$
S(f_\theta, f'_{\theta'}) = \frac{1}{|\mathcal{X}_{\text{test}}|} \sum_{x \in \mathcal{X}_{\text{test}}} \mathbb{I}(f_\theta(x) \approx f'_{\theta'}(x)),
$$
where $\approx$ means a match or semantic closeness of the outputs.

\item \textbf{Insecure output handling \cite{owasp_foundation_owasp_2023, owasp_llm02_2023}.}
\label{threat:insecure_output}
The absence of validation and sanitization of LLM outputs before they are passed to other systems opens a vector for cascading attacks. Generated SQL code can be executed on a database server, and a JavaScript fragment can be rendered in a web interface. The result is SQL injection, cross-site scripting (XSS), and similar vulnerabilities. The risk of generating exploitable code is characterized by the probability $P(f_\theta(x) \in \mathcal{Y}_{\text{exploit}})$, where $\mathcal{Y}_{\text{exploit}}$ denotes the set of malicious outputs.

\item \textbf{Training data extraction \cite{carlini_extracting_2021}.}
\label{threat:data_extraction}
LLMs tend to memorize fragments of the training set, which makes it possible to reconstruct verbatim fragments $d$ from $\DataDist_{\text{train}}$. The task of the attacker is to find a query $x^*$ that maximizes the extraction metric:
\begin{gather*}
x^* = \arg\max_{x \in \Xspace} S_{\text{extract}}(f_\theta, x), \\
S_{\text{extract}}(f_\theta, x) = \max_{d \in \DataDist_{\text{train}}} \frac{|\text{longest\_common\_substring}(f_\theta(x), d)|}{|\text{relevant\_part}(f_\theta(x))|}.
\end{gather*}
The denominator substantially affects the interpretation of the results. Normalization by the length of the entire response underestimates the metric when short fragments, for example personal data, are extracted. In such cases, it is more appropriate to normalize by the length of the presumed secret.

\item \textbf{Sensitive information disclosure \cite{wang_privacy_leakage_2023}.}
\label{threat:info_disclosure}
Unlike targeted extraction, the leakage of confidential information $\mathcal{C}$ is often unintentional. Without a specially constructed query, the model can disclose personal data, trade secrets, and other sensitive information from the training set or from the current dialogue context. For quantitative evaluation, we use the mutual information $I(f_\theta(X); \mathcal{C}) = H(\mathcal{C}) - H(\mathcal{C}|f_\theta(X))$ \cite{carlini_extracting_2021}. A high value of $I$ indicates a substantial disclosure of information about $\mathcal{C}$ through the model outputs.

\item \textbf{Bias amplification \cite{bias_amplification_2023, zhao_men_2017}.}
\label{threat:bias_amplification}
Models do not merely reproduce the social biases of the training data. They can amplify them and generate discriminatory results for particular demographic groups. The degree of amplification is expressed by the coefficient $A_{\text{bias}}(f_\theta, \mathcal{D}) = B_{\text{output}}(f_\theta, \mathcal{D}) / B_{\text{input}}(\mathcal{D})$, which relates the bias of the outputs $B_{\text{output}}$ to the bias of the source data $B_{\text{input}}$. Here, for demographic groups $g_1, g_2 \in \mathcal{G}$, $B_{\text{output}}(f_\theta, \mathcal{D}) = \text{Dist}(P(Y|X \in g_1), P(Y|X \in g_2))$ with $X \sim \mathcal{D}$ and $Y = f_\theta(X)$, and $B_{\text{input}}(\mathcal{D})$ is the same quantity with $Y$ taken from the data $\mathcal{D}$. The value $A_{\text{bias}} > 1$ signals amplification.

\item \textbf{Supply chain vulnerabilities \cite{owasp_foundation_owasp_2023, nist_sp800_161_2022}.}
\label{threat:supply_chain}
A modern LLM system depends on numerous components: pretrained models, datasets, fine-tuning libraries, third-party plugins, and infrastructure elements \cite{gu_badnets_2017}. The compromise of any link $c_i$ in the chain $\mathcal{S} = \{c_1, \dots, c_n\}$ puts the entire system at risk. If the components are assumed independent, the probability of system vulnerability is $P(V_{\text{system}}) \approx 1 - \prod_{i=1}^{n} (1 - P(V_i))$. This quantity grows monotonically with $n$. For equal individual risks $P(V_i) = p$, it is approximately $np$ while $np \ll 1$ and approaches one when $n$ becomes large compared with $1/p$. Therefore, even small individual risks accumulate in a long supply chain.

\item \textbf{Jailbreaking \cite{wei_jailbreak_2023}.}
\label{threat:jailbreaking}
The results of the SaTML CTF 2024 showed that even a multi-layer defense does not guarantee robustness to jailbreaking attacks. Such attacks overcome safety alignment \cite{ouyang_training_2022} and force the model to generate prohibited content. Formally, the attacker seeks a query $x_{\text{jb}}$ that maximizes the probability of an undesirable output that the safety filter lets through:
$$
x_{\text{jb}} = \arg\max_{x \in \Xspace} P\bigl(f_\theta(x) \in \badY \,\wedge\, \text{SafetyFilter}(f_\theta(x)) = \text{allow}\bigr).
$$
The jailbreaking robustness metric takes the form $R_{\text{jb}}(f_\theta) = 1 - \max_{x \in \mathcal{X}_{\text{jb\_attacks}}} P(f_\theta(x) \in \badY)$ and is computed on a set of known jailbreaking queries $\mathcal{X}_{\text{jb\_attacks}}$ \cite{jin_jailbreakzoo_2024}.

\item \textbf{Reproducibility issues \cite{belz_systematic_2021}.}
\label{threat:reproducibility}
LLM generation is stochastic: sampling and infrastructure changes cause repeated calls on the same input $x$ to give different responses $\{y_1, \dots, y_k\}$. The divergence measure $NRep(f_\theta, x) = D(\{y_1, \dots, y_k\} | y_i \sim p_\theta(\cdot|x))$ captures the degree of non-reproducibility. The consequences for security are twofold: testing and debugging become harder, and vulnerabilities can manifest sporadically, which masks the real risk.

\item \textbf{Insecure plugins \cite{greshake_not_2023, owasp_foundation_owasp_2023}.}
\label{threat:insecure_plugins}
Integration with third-party plugins creates risks when the principle of least privilege is violated. When the privilege set of a plugin $\text{Priv}(p)$ exceeds the minimum necessary set $\text{Priv}_{\text{min}}(p)$, the excess metric $E(p) = |\text{Priv}(p) \setminus \text{Priv}_{\text{min}}(p)| / |\text{Priv}(p)|$ quantifies the degree of risk. A compromised plugin can act as a channel for bypassing content filters and gives rise to a game between the attacker and the defender \cite{security_games_2023}.

\item \textbf{Detection evasion (model evasion) \cite{li_large_2023, papernot_practical_2017}.}
\label{threat:evasion}
A malicious query $x_{\text{mal}}$ is modified into $x_{\text{evasion}}$, which minimizes the score $\text{DetectionScore}(D(x_{\text{evasion}}))$ while preserving the condition $f_\theta(x_{\text{evasion}}) \in \badY$. The attack targets the defense perimeter of the system and reduces the effectiveness of content filters.
\end{enumerate}

Of the threats classified above, this thesis analyzes prompt injection, jailbreak attacks, and trojans further, both formally and experimentally. Attacks on agentic LLM systems that interact with external tools form a separate class. In the OWASP classification, they relate to excessive agency (LLM08) and insecure plugins (LLM07). Chapter~\ref{ch:ch5} considers them in detail for MCP.

\subsection{Threat Categorization}

The threats described above are grouped by six security aspects (Fig.~\ref{fig:threat_taxonomy_h}).

\textbf{Data and model integrity.} Training data poisoning (LLM03 in OWASP), backdoor attacks, and adversarial examples distort the parameters or the behavior of the model. The local vulnerability $\Vulfunc_{\text{local}}(f_\theta)$ partially reflects this aspect when incorrect responses are included in $\badY$.

\textbf{Confidentiality} is affected by a spectrum of threats: model inversion, membership inference, training data extraction, and sensitive information disclosure, all of which correspond to LLM06, as well as system prompt disclosure \cite{hui_pleak_2024}. The quantitative evaluation relies on the metric $P_{\text{leak}}$ and on the mutual information $I(f_\theta(X); \mathcal{C})$.

The \textbf{availability} aspect is represented by denial of service (LLM04), which disrupts the operation of the service for legitimate users.

The \textbf{governance and control} aspect covers the widest list of threats: prompt injection (LLM01), insecure output handling (LLM02), supply chain vulnerabilities (LLM05), insecure plugin design (LLM07), excessive agency (LLM08), model theft (LLM10), jailbreaking, and detection evasion. Minimizing $E(p)$ for plugins directly reduces the risks of this category.

\textbf{Ethics and fairness} are affected by bias amplification and by the generation of socially unacceptable content through jailbreaking, prompt injection, and adversarial examples. The metrics $A_{\text{bias}}$ and $\Vulfunc_{\text{local}}(f_\theta)$, with an appropriate definition of $\badY$, can be used for the evaluation.

\textbf{Reliability and reproducibility} depend on the non-determinism of generation, which affects the predictability of behavior and is captured by the metric $NRep(f_\theta, x)$.

Some threats affect several categories at once: prompt injection can lead to data leakage, integrity violation, and the generation of malicious content.

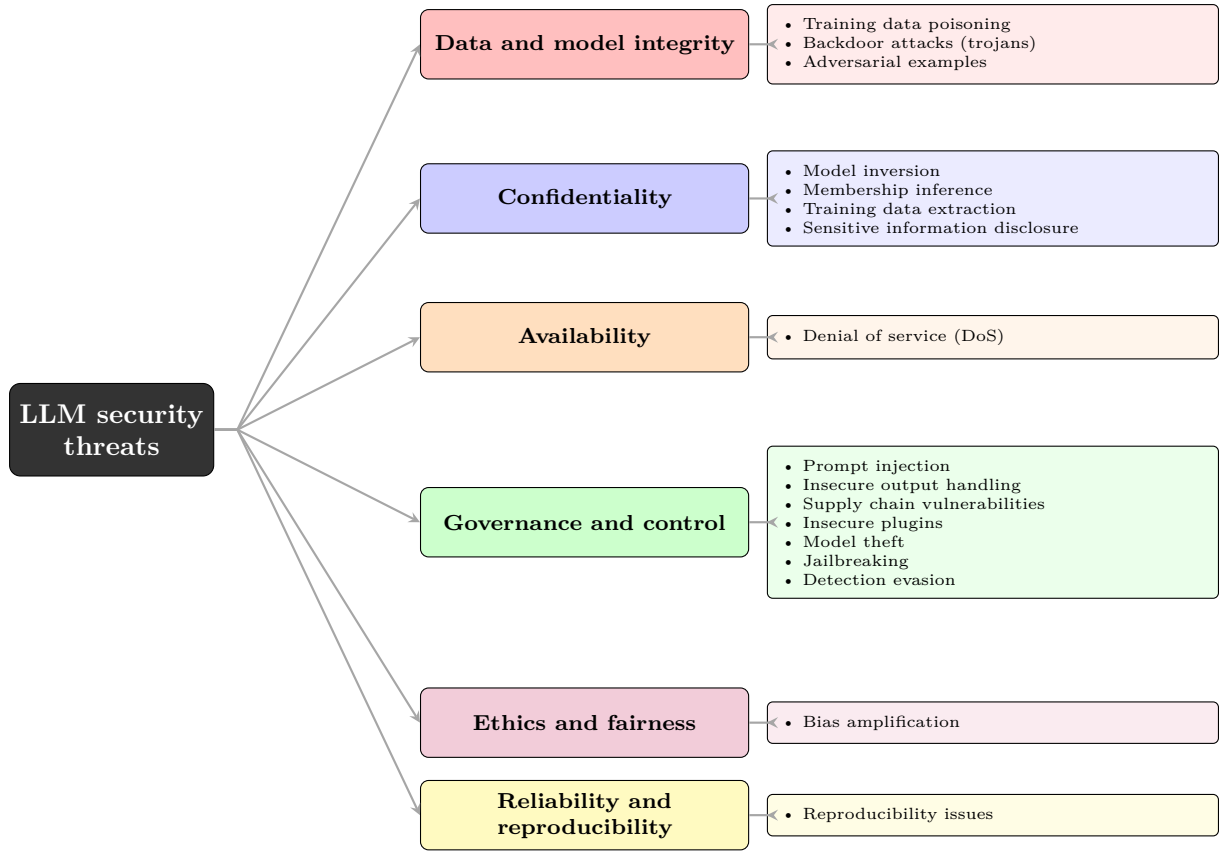
\begin{figure}[ht!]
\centering
\resizebox{\textwidth}{!}{%
\begin{tikzpicture}[
    root/.style={draw, rounded corners=4pt, fill=black!80, text=white,
        font=\bfseries\small, minimum height=1.2cm, minimum width=2.5cm, align=center},
    category/.style={draw, rounded corners=3pt, font=\bfseries\scriptsize,
        minimum height=0.9cm, text width=4.0cm, align=center, anchor=west},
    threatbox/.style={draw, rounded corners=2pt, fill=white, font=\tiny,
        text width=5.5cm, align=left, inner sep=5pt, anchor=west},
    cat1/.style={category, fill=red!25},
    cat2/.style={category, fill=blue!20},
    cat3/.style={category, fill=orange!25},
    cat4/.style={category, fill=green!20},
    cat5/.style={category, fill=purple!20},
    cat6/.style={category, fill=yellow!30},
    arrow/.style={->, >=stealth, thick, gray!70},
]

\node[root] (root) {LLM security\\threats};

\node[cat1] (c1) at (4.0, 5.0) {Data and model integrity};
\node[cat2] (c2) at (4.0, 3.0) {Confidentiality};
\node[cat3] (c3) at (4.0, 1.2) {Availability};
\node[cat4] (c4) at (4.0, -1.2) {Governance and control};
\node[cat5] (c5) at (4.0, -3.8) {Ethics and fairness};
\node[cat6] (c6) at (4.0, -5.0) {Reliability and reproducibility};

\node[threatbox, fill=red!8] (t1) at (8.5, 5.0) {
    \textbullet\ Training data poisoning\\
    \textbullet\ Backdoor attacks (trojans)\\
    \textbullet\ Adversarial examples
};

\node[threatbox, fill=blue!8] (t2) at (8.5, 3.0) {
    \textbullet\ Model inversion\\
    \textbullet\ Membership inference\\
    \textbullet\ Training data extraction\\
    \textbullet\ Sensitive information disclosure
};

\node[threatbox, fill=orange!8] (t3) at (8.5, 1.2) {
    \textbullet\ Denial of service (DoS)
};

\node[threatbox, fill=green!8] (t4) at (8.5, -1.2) {
    \textbullet\ Prompt injection\\
    \textbullet\ Insecure output handling\\
    \textbullet\ Supply chain vulnerabilities\\
    \textbullet\ Insecure plugins\\
    \textbullet\ Model theft\\
    \textbullet\ Jailbreaking\\
    \textbullet\ Detection evasion
};

\node[threatbox, fill=purple!8] (t5) at (8.5, -3.8) {
    \textbullet\ Bias amplification
};

\node[threatbox, fill=yellow!12] (t6) at (8.5, -5.0) {
    \textbullet\ Reproducibility issues
};

\draw[arrow] (root.east) -- ++(0.3,0) -- (c1.west);
\draw[arrow] (root.east) -- ++(0.3,0) -- (c2.west);
\draw[arrow] (root.east) -- ++(0.3,0) -- (c3.west);
\draw[arrow] (root.east) -- ++(0.3,0) -- (c4.west);
\draw[arrow] (root.east) -- ++(0.3,0) -- (c5.west);
\draw[arrow] (root.east) -- ++(0.3,0) -- (c6.west);

\draw[arrow] (c1.east) -- ++(0.3,0) -- (t1.west);
\draw[arrow] (c2.east) -- ++(0.3,0) -- (t2.west);
\draw[arrow] (c3.east) -- ++(0.3,0) -- (t3.west);
\draw[arrow] (c4.east) -- ++(0.3,0) -- (t4.west);
\draw[arrow] (c5.east) -- ++(0.3,0) -- (t5.west);
\draw[arrow] (c6.east) -- ++(0.3,0) -- (t6.west);

\end{tikzpicture}%
}
\caption{Taxonomy of LLM security threats: 6 categories, 17 threat types.}
\label{fig:threat_taxonomy_h}
\end{figure}

\subsection{Mapping to Security Frameworks}

The choice of defense measures requires mapping the identified threats to applied frameworks: OWASP Top 10 for LLM Applications and NIST SP 800-53. Such a mapping is given below. It sets the context for the choice of the research subject.

\subsubsection{Threat Classification According to OWASP Top 10 for LLM (2023)}
Table~\ref{tab:owasp_llm_ru} presents the correspondence between the threats described in this section and the classification of OWASP Top 10 for LLM Applications, version 2023 \cite{owasp_foundation_owasp_2023}. This classification focuses on practical risks in the development and operation of LLM-based applications.

{\small
\begin{longtable}{|>{\raggedright}p{1.3cm}|>{\raggedright}p{2.8cm}|>{\raggedright}p{4.5cm}|>{\raggedright\arraybackslash}p{3.7cm}|}
\caption{Threat classification according to OWASP Top 10 for LLM (2023) and its relation to the detailed threat list}\label{tab:owasp_llm_ru}\\
\hline
\textbf{ID} & \textbf{Threat (OWASP)} & \textbf{Description} & \textbf{Threats from Sec.~\ref{subsec:threat_descriptions}} \\ \hline
\endfirsthead
\multicolumn{4}{c}{{\tablename\ \thetable{} (continued)}} \\
\hline
\textbf{ID} & \textbf{Threat (OWASP)} & \textbf{Description} & \textbf{Threats from Sec.~\ref{subsec:threat_descriptions}} \\ \hline
\endhead
\hline \multicolumn{4}{r}{{Continued on next page}} \\ \hline
\endfoot
\hline
\endlastfoot
LLM01 & Prompt Injection & Injection of malicious instructions that override the behavior of the model. & 1.~Prompt injection, 14.~Jailbreaking \\ \hline
LLM02 & Insecure Output Handling & Insufficient validation and sanitization of LLM outputs before use in other systems. & 9.~Insecure output handling \\ \hline
LLM03 & Training Data Poisoning & Manipulation of the training data to compromise the model. & 2.~Data poisoning, 3.~Backdoor attacks \\ \hline
LLM04 & Model Denial of Service (Model DoS) & Attacks that lead to resource exhaustion and model unavailability. & 7.~Denial of service (DoS) \\ \hline
LLM05 & Supply Chain Vulnerabilities & Use of vulnerable components: models, dependencies, data. & 13.~Supply chain vulnerabilities \\ \hline
LLM06 & Sensitive Information Disclosure & Unintended disclosure of confidential data by the model. & 11.~Sensitive information disclosure, 4.~Model inversion, 5.~Membership inference, 10.~Data extraction \\ \hline
LLM07 & Insecure Plugin Design & Security flaws in LLM plugins. & 16.~Insecure plugins \\ \hline
LLM08 & Excessive Agency & Granting an LLM excessive permissions. & 16.~Insecure plugins, overall system design \\ \hline
LLM09 & Overreliance & Excessive trust in LLM outputs without verification. & 6.~Adversarial examples, 12.~Bias, 15.~Reproducibility \\ \hline
LLM10 & Model Theft & Unauthorized copying or extraction of the model. & 8.~Model theft \\ \hline
\end{longtable}
}

\subsubsection{Mapping to NIST SP 800-53 Security Controls}
Table~\ref{tab:nist_mapping_ru} maps the threats to the relevant control families from the NIST SP 800-53 Rev. 5 catalog \cite{nist_sp800_53_2020}. This makes it possible to integrate LLM risk management into the overall information security management system of an organization. The choice of specific controls must be based on a risk assessment.

{\small
\begin{longtable}{|>{\raggedright}p{3.0cm}|>{\raggedright}p{4.0cm}|>{\raggedright\arraybackslash}p{5.3cm}|}
\caption{Mapping of LLM threats to NIST SP 800-53 Rev. 5 control families}\label{tab:nist_mapping_ru}\\
\hline
\textbf{Threat} & \textbf{NIST control families} & \textbf{Example controls} \\ \hline
\endfirsthead
\multicolumn{3}{c}{{\tablename\ \thetable{} (continued)}} \\ \hline
\textbf{Threat} & \textbf{NIST control families} & \textbf{Example controls} \\ \hline
\endhead
\hline \multicolumn{3}{r}{{Continued on next page}} \\ \hline
\endfoot
\hline
\endlastfoot
1.~Prompt injection & SI, AC & SI-10, AC-6, AC-3 \\ \hline
2.~Data poisoning & SI, SR & SI-7, SR-3, SR-6 \\ \hline
3.~Backdoor attacks & SI, SR & SI-7, SR-4, SR-5 \\ \hline
4.~Model inversion & AC, PT & AC-3, AC-4, PT-3, PT-4 \\ \hline
5.~Membership inference & PT, SI, AU & PT-3, SI-4, AU-6 \\ \hline
6.~Adversarial examples & SI, AT & SI-4, SI-10, AT-2 \\ \hline
7.~Denial of service & SC, CP & SC-5, CP-2 \\ \hline
8.~Model theft & AC, SC, IA & AC-3, AC-4, SC-28, IA-2 \\ \hline
9.~Insecure output handling & SI & SI-10, SI-11 \\ \hline
10.~Data extraction & PT, AC, AU & PT-3, AC-3, AU-6 \\ \hline
11.~Sensitive information disclosure & PT, AC, SC & PT-2, PT-3, AC-3, SC-8 \\ \hline
12.~Bias amplification & PM, PL, SI & PM-9, PL-8, SI-19 \\ \hline
13.~Supply chain vulnerabilities & SR & SR-2, SR-3, SR-4, SR-5 \\ \hline
14.~Jailbreaking & SI, AC & SI-4, SI-7, AC-3 \\ \hline
15.~Reproducibility issues & CM, SI & CM-6, SI-2 \\ \hline
16.~Insecure plugins & SI, AC, SR & SI-7, AC-6, SR-11 \\ \hline
17.~Detection evasion & SI, RA & SI-4, RA-5, RA-10 \\ \hline
\end{longtable}
}

\subsection{Illustrative Threat Examples}
\label{subsec:cases}
Three examples illustrate the threat classes considered above with the results of controlled experiments and LLM security competitions.

\textbf{1. Vulnerability of LLM-as-a-Judge to targeted prompt injection attacks.}
\begin{itemize}
    \item \textit{Description:} LLM-as-a-Judge architectures are susceptible to prompt injection attacks. The Comparative Undermining Attack (CUA) changes the decision of the judge directly. The Justification Manipulation Attack (JMA) distorts the justification of the verdict. Both attacks are implemented with optimization methods, in particular Greedy Coordinate Gradient (GCG).
    \item \textit{Effectiveness:} On the Gemma, Llama, GPT-4, and Claude-3-Opus models with the MT-Bench dataset and Kaggle tasks, CUA and the Adaptive Search-Based Attack (ASA), defined in Chapter~\ref{ch:ch3}, achieved a high success rate, especially on open-weight models. Quantitative results are given in Chapter~\ref{ch:ch4} (Table~\ref{tab:llm_judge_asr_summary}).
\end{itemize}

\textbf{2. Difficulties of detecting trojans in LLMs.}
The Trojan Detection Challenge 2023 studied the reverse engineering of trojans in the Pythia model. The reconstruction of the true triggers proved to be a fundamentally hard task, whereas the generation of surrogate triggers that cause the desired output was achieved much more easily. These results indicate the priority of preventive measures. A detailed analysis of the results is given in Chapter~\ref{ch:ch4} (Section~\ref{subsec:trojan_detection_results_ch4_part4}).

\textbf{3. Bypassing a multi-layer defense with prompt injection.}
The SaTML CTF 2024 competition targeted attacks on LLMs, in particular GPT-3.5 and Llama 2, equipped with a three-layer defense: system instructions, a Python filter, and an LLM filter. Participants developed and applied various prompt injection techniques, including multi-turn attacks and specific wordings for bypassing filters. The multi-layer defense configurations could be bypassed. The success of an attack was determined by the sophistication of the prompt and the specifics of the filters. Defense strategies need constant adaptation and supplementation with dynamic detection methods.

\subsection{Analysis of Threat Interrelations}
\label{subsec:threat_relationships_analysis}

The threat types for LLMs are interrelated: countering one threat affects the vulnerability to others. Rigorous theorems valid for all LLMs are difficult to formulate, so below we give well-founded assumptions and principles supported by the literature.

\textbf{Principle 1: Relation between data extraction and sensitive information disclosure.}
Training data extraction (\ref{threat:data_extraction}) is a special case of sensitive information disclosure (\ref{threat:info_disclosure}). Data extraction attacks aim to reconstruct fragments of $\DataDist_{\text{train}}$ \cite{carlini_extracting_2021}, whereas leakage can also occur passively \cite{mireshghallah_can_2023}.
\textit{Assumption:} The success of targeted data extraction attacks $S_{\text{extract}}$, defined in Section~\ref{subsec:threat_descriptions}, serves as a lower bound for the more general risk of confidential information leakage $P_{\text{leak}}$ under comparable definitions of the metrics. Defense measures against memorization, in particular differential privacy methods \cite{dwork_algorithmic_2014, abadi_deep_2016}, reduce the risks of both threats \cite{carlini_extracting_2021}.

\textbf{Principle 2: Relation between jailbreaking, injections, and detection evasion.}
Prompt injection techniques (\ref{threat:prompt_injection}) often serve as a means to implement jailbreaking attacks (\ref{threat:jailbreaking}) aimed at bypassing the built-in defense mechanisms of the model \cite{wei_jailbreak_2023}. At the same time, a successful jailbreaking attack does not necessarily mean an evasion of the detection system (\ref{threat:evasion}): an external detector $D$ can classify the malicious query $x_{\text{jb}}$ as dangerous.
The difficulty of creating effective jailbreaking queries and the difficulty of detecting them are related, but advanced attacks are designed to bypass both the internal mechanisms of the model and external detectors simultaneously \cite{zou_universal_2023}. Countering them therefore requires both an improvement of robust alignment \cite{cao_defending_2023} and the refinement of external detectors \cite{robey_smoothllm_2023}.

\textbf{Principle 3: Effect of reproducibility on security.}
The dependence between LLM non-determinism (\ref{threat:reproducibility}) and security is not monotonic \cite{spirling_replication_2023}. On the one hand, the stochasticity of generation hinders the creation of precise adversarial examples (\ref{threat:adversarial_examples}) and jailbreaking attacks (\ref{threat:jailbreaking}) that require a strictly defined response \cite{chen_llm_2024}. On the other hand, it also complicates the testing and verification of defense mechanisms. Rare but dangerous vulnerabilities manifest stochastically and are therefore harder to detect. The reliability of attack detectors (\ref{threat:evasion}) also drops when the reproducibility of system behavior is low. The choice of the optimal level of determinism is dictated by the specific application and the threat profile.

\textbf{Principle 4: Relation between bias and confidentiality.}
Bias amplification (\ref{threat:bias_amplification}) and sensitive information disclosure (\ref{threat:info_disclosure}) are related through a common mechanism: the memorization of correlations in the training data. A model that has learned associations between sensitive demographic attributes and other data in $\DataDist_{\text{train}}$ simultaneously amplifies biases and becomes vulnerable to membership inference attacks (\ref{threat:membership_inference}) for individuals from the corresponding groups \cite{blodgett_language_2020, mehrabi_survey_2021}. The dependence is bidirectional. Weakening the association of the model with sensitive attributes reduces bias and increases confidentiality. Conversely, strict confidentiality measures, in particular DP \cite{dwork_algorithmic_2014}, reduce accuracy disproportionately on underrepresented subgroups and can thereby amplify bias \cite{bagdasaryan_differential_2020}.

\textbf{Principle 5: Cascading risks of the supply chain and plugins.}
Vulnerabilities in the supply chain (\ref{threat:supply_chain}) directly affect the security of the integrated components, including plugins (\ref{threat:insecure_plugins}). The compromise of a base model or library entails the compromise of all dependent plugins. An insecure plugin with excessive privileges, that is, with a high value of $E(p)$ from Section~\ref{subsec:threat_descriptions}, creates an attack vector against the base model or other system components.
The risks of these two categories compound each other, which requires a systematic approach: privilege control, formalized procedures for integrating third-party extensions, and regular audits of the components in use \cite{owasp_foundation_owasp_2023}.

\textbf{Principle 6: Fundamental trade-off.}
Many aspects of LLM security and functionality trade off against each other. For example:
\begin{itemize}
    \item \textit{Utility vs.\ safety/censorship:} Overly strict defense mechanisms aimed against jailbreaking and toxicity can reduce the utility of the model and its ability to fulfill legitimate requests, which is known as the ``lobotomy effect'' \cite{wei_jailbreak_2023}.
    \item \textit{Confidentiality vs.\ utility:} Confidentiality methods, for example DP, often reduce the accuracy and overall performance of the model \cite{abadi_deep_2016, bagdasaryan_differential_2020}.
    \item \textit{Robustness vs.\ accuracy:} Methods for improving robustness to adversarial attacks, for example adversarial training, can reduce accuracy on clean data, in some settings substantially \cite{madry_towards_2017, tsipras_robustness_2019}.
\end{itemize}
In practice, these trade-offs limit the simultaneous maximization of all target properties of an LLM system. The specific application scenario determines the prioritization among security, utility, confidentiality, and computational efficiency.

LLM security does not reduce to countering isolated threats: defense measures must account for the interaction between threats and for the system properties of the model in the context of its deployment.

\subsection{Conclusions and Risk Mitigation Recommendations}

The threats affect all stages of working with a model, from data preparation to deployment. Countering them requires a combination of technical, organizational, and methodological measures.

\begin{itemize}
    \item \textbf{Preventive measures at the development and training stage.}
        Auditing and filtering of the training data reduce the risks of poisoning (LLM03, \ref{threat:data_poisoning}, \ref{threat:backdoor}) \cite{wang_empirical_2023} and bias (\ref{threat:bias_amplification}) \cite{li_mbias_2024}. Privacy-enhancing methods, such as differentially private training with DP-SGD \cite{abadi_deep_2016} and federated learning \cite{mcmahan_communication-efficient_2017}, counter data leakage (LLM06, \ref{threat:model_inversion}, \ref{threat:membership_inference}, \ref{threat:data_extraction}, \ref{threat:info_disclosure}) \cite{dwork_algorithmic_2014}. Adversarial training \cite{madry_towards_2017} and safety alignment \cite{ouyang_training_2022, bai_constitutional_2022} increase robustness to adversarial attacks (\ref{threat:adversarial_examples}) and jailbreaking (\ref{threat:jailbreaking}). Control of the provenance and integrity of supply chain components (LLM05, \ref{threat:supply_chain}) is also a preventive measure \cite{kumar_securing_2023}.
    \item \textbf{Defense mechanisms at the operation stage.}
        At the input level, query validation and sanitization are applied to prevent injections (LLM01) \cite{liu_prompt_2023, liu_prompt_defense_2023}, together with the separation of data and instructions \cite{chen_struq_2024}. At the output level, filtering and sanitization are performed before use in other systems (LLM02, \ref{threat:insecure_output}) \cite{owasp_llm02_2023}. Access control and authentication of the model API protect against theft (LLM10) \cite{tramer_stealing_2016}. Resource monitoring and limits counter DoS (LLM04) \cite{zhang_crabs_2023}, and the principle of least privilege for plugins counters the LLM07 threat \cite{wang_trojaning_2023}. The detection of anomalous behavior and of attempts to bypass the defense (\ref{threat:evasion}) complements these measures \cite{zou_universal_2023, wei_jailbreak_2023}.
    \item \textbf{Continuous monitoring and risk management.}
        Red teaming and penetration testing \cite{perez_red_2022, ganguli_red_2022} reveal previously unknown attack vectors. Event log analysis and action auditing \cite{raji_auditing_2023} ensure prompt incident detection. The model and the defense mechanisms must be updated regularly in response to changes in the threat landscape. Quantitative measurement of the security level with the metrics $R_{stab}(f_\theta)$ and $\Vulfunc_{\text{local}}(f_\theta)$ makes it possible to track the dynamics \cite{liang_holistic_2023}.
\end{itemize}

Formal risk assessment methods combined with the control measures of the NIST and OWASP standards make it possible to build a multi-layer defense system for LLMs under real operating conditions. A systematic analysis of the vulnerabilities of agentic coding systems, covering 42 attack techniques and 18 defense mechanisms, is presented in~\cite{maloyan2025agentic}. The subsequent parts of the thesis are devoted to the experimental validation of specific methods for countering the threats considered here.

\subsection{Formalization of the Threat Space for Large Language Models}
\label{subsec:threat_space_formalization}

The systematization of threats rests on a formal structure that describes the components and their interrelations. Abstracting from specific attack implementations makes it possible to reveal general patterns.

\begin{definition}
\label{def:threat_space}
\textbf{LLM threat space.}
The threat space for a large language model $f_\theta: \Xspace \to \ProbSpace{\Yspace}$, in the notation of Section~\ref{sec:llm_threats}, is defined as the tuple $\mathcal{T} = (\mathcal{A}^{\mathrm{type}}, \mathcal{V}, \mathcal{I}_{AV}, \mathcal{C}_{VV})$, where:
\begin{itemize}
    \item $\mathcal{A}^{\mathrm{type}}$ is the set of attack \textit{types}, for example Prompt Injection, Data Poisoning, CUA, JMA, and others, corresponding to the classes from Section~\ref{subsec:threat_descriptions}. The superscript ${}^{\mathrm{type}}$ emphasizes that this is a classification of attacks. In Chapter~\ref{ch:ch2}, the symbol $\mathfrak{A}$ denotes a family of attacks as functions.
    \item $\mathcal{V}$ is the set of \textit{specific} vulnerabilities of the model or of the system in which it operates, for example: $v_1=$ ``insufficient sanitization of user input in the system prompt'', $v_2=$ ``presence of memorized PII in the model parameters'', $v_3=$ ``vulnerability of LLM-as-a-Judge to justification manipulation''.
    \item $\mathcal{I}_{AV} \subseteq \mathcal{A}^{\mathrm{type}} \times \mathcal{V}$ is the relation ``attack exploits vulnerability''. It reflects which specific vulnerabilities ($v \in \mathcal{V}$) can be exploited by attacks of a certain type ($a \in \mathcal{A}^{\mathrm{type}}$). For example, (Prompt Injection, $v_1$) $\in \mathcal{I}_{AV}$.
    \item $\mathcal{C}_{VV} \subseteq \mathcal{V} \times \mathcal{V}$ is the relation of causality or dependence between vulnerabilities. It means that the exploitation or presence of a vulnerability $v_i$ can lead to the emergence or exploitability of a vulnerability $v_j$. For example, ($v_{\text{supply\_chain}}$, $v_{\text{backdoor}}) \in \mathcal{C}_{VV}$, where $v_{\text{supply\_chain}}$ is a supply chain vulnerability, and $v_{\text{backdoor}}$ is the possibility of implanting a trojan.
\end{itemize}
\end{definition}

The cardinality of the sets $\mathcal{A}^{\mathrm{type}}$ and $\mathcal{V}$ depends on the specifics of the model, the architecture, the data, and the deployment conditions, and can be large. The relation $\mathcal{C}_{VV}$ reveals cascading effects and root causes of vulnerabilities, and the construction of the dependency graph $G = (\mathcal{V}, \mathcal{C}_{VV})$ is a separate security audit task.

A quantitative risk assessment requires a notion that accounts for both the attack success probability and the potential damage.

\begin{definition}
\label{def:attack_risk}
\textbf{Attack risk assessment.}
The risk assessment for a model $f_\theta$ with respect to an attack type $a \in \mathcal{A}^{\mathrm{type}}$ is defined as
$$
\text{Risk}(a, f_\theta) = \sum_{v \in \mathcal{V} : (a, v) \in \mathcal{I}_{AV}} P_{\text{exploit}}(v | a, f_\theta) \cdot \text{ImpactValue}(v),
$$
where $P_{\text{exploit}}(v | a, f_\theta)$ is the probability of successful exploitation of the vulnerability $v$ through an attack of type $a$ against the model $f_\theta$, in particular the ASR from the experiments of this work, and $\text{ImpactValue}(v)$ is an estimate of the potential damage from exploiting the vulnerability $v$, for example on the CVSS scale \cite{melnick_cvss_2023}, which can be normalized to the range $[0,1]$ for the purposes of comparative risk analysis. The total risk over a set of attacks $\mathcal{A}' \subseteq \mathcal{A}^{\mathrm{type}}$ can be computed as $\sum_{a \in \mathcal{A}'} w(a) \cdot \text{Risk}(a, f_\theta)$, where $w(a)$ is the significance weight of the attack, determined, for example, from the frequency or the presumed motivation of attackers.
\end{definition}

Estimating $P_{\text{exploit}}$ is a complex task that often requires specialized tests, in particular red teaming \cite{perez_red_2022} or the analysis of competition data, or relies on expert judgment and the statistics of known incidents. The estimate $\text{ImpactValue}(v)$ depends on the application context of the model and on the criticality of the data or functions affected by the vulnerability.

The order in which defense measures are prioritized is determined by the vulnerability dependency graph.

\begin{principle}[On the prioritization of vulnerability remediation]
\label{princ:vulnerability_prioritization}
Let $G = (\mathcal{V}, \mathcal{C}_{VV})$ be the vulnerability dependency graph, where $\mathcal{V}$ is the set of vulnerabilities, and $\mathcal{C}_{VV}$ is the dependency relation. The application of defense measures aimed at remediating vulnerabilities must account for the structure of this graph.
\begin{enumerate}
    \item \textbf{Priority to root vulnerabilities:} In an acyclic dependency graph, vulnerabilities without incoming edges, that is, root vulnerabilities, are the root causes of other vulnerabilities. Their remediation has a cascading effect: it prevents the exploitation of each dependent vulnerability that arises or becomes exploitable only through this root. A dependent vulnerability that is also reachable from another root, or that is present independently, requires its own remediation.
    \item \textbf{Breaking dependency cycles:} In a graph that contains cycles, priority should be given to remediating vulnerabilities in the strongly connected components that are sources, that is, have no incoming edges, in the condensation graph. Remediating even one vulnerability in a cycle breaks it, which can significantly reduce the overall risk.
\end{enumerate}
The principle forms a methodological basis for planning defense measures. It directs resources to the remediation of the vulnerabilities that are most critical from the standpoint of structural dependencies.
\end{principle}

\begin{figure}[ht]
  \centering
  \begin{subfigure}[t]{0.45\textwidth}
    \centering
    \begin{tikzpicture}[
        >=stealth', shorten >=1pt, auto,
        node distance=1.4cm,
        vulnnode/.style={circle, draw, minimum size=7mm, font=\small},
        baseline=(v1.north)
      ]
      \node[vulnnode] (v1) {$v_1$};
      \node[vulnnode] (v2) [right of=v1] {$v_2$};
      \node[vulnnode] (v3) [below of=v1] {$v_3$};
      \node[vulnnode] (v4) [right of=v3] {$v_4$};
      \node[vulnnode] (v5) [right of=v2] {$v_5$};

      \draw[->] (v1) -- (v2);
      \draw[->] (v1) -- (v3);
      \draw[->] (v2) -- (v5);
      \draw[->] (v3) -- (v4);
      \draw[->] (v4) -- (v5);

      \node[font=\footnotesize, align=left] at (1.4,-3.2) {%
        Forward: $(v_1,v_3,v_2,v_4,v_5)$\\
        Reverse: $(v_5,v_4,v_2,v_3,v_1)$};

      \node at (1.4,-4.5) {\phantom{Condensation DAG $G^{\mathrm{SCC}}$}};
    \end{tikzpicture}
    \caption{A DAG and two variants of the remediation order}
    \label{fig:acyclic}
  \end{subfigure}
  \hfill
  \begin{subfigure}[t]{0.45\textwidth}
    \centering
    \begin{tikzpicture}[
        >=stealth', shorten >=1pt, auto,
        node distance=1.4cm,
        vulnnode/.style={circle, draw, minimum size=7mm, font=\small},
        sccnode/.style={rectangle, draw, rounded corners=2pt, inner sep=2pt, font=\footnotesize},
        baseline=(a.north)
      ]
      \node[vulnnode] (a) {$a$};
      \node[vulnnode] (b) [right of=a] {$b$};
      \node[vulnnode] (c) [below of=b] {$c$};
      \node[vulnnode] (d) [right of=b, xshift=1cm] {$d$};
      \node[vulnnode] (e) [below of=d] {$e$};

      \draw[->] (a) -- (b);
      \draw[->] (b) -- (c);
      \draw[->] (c) -- (a);

      \draw[->] (d) -- (e);
      \draw[->] (e) -- (d);

      \draw[->, dashed] (b) -- (d);
      \draw[->, dashed] (c) -- (e);

      \node[sccnode] (C1) at (0.7,-3.2) {$C_1{=}\{a,b,c\}$};
      \node[sccnode] (C2) at (4,-3.2) {$C_2{=}\{d,e\}$};

      \draw[->] (C1) -- node[midway, above=3pt, font=\footnotesize] {priority} (C2);

      \node[font=\footnotesize] at (2.3,-4.5) {Condensation DAG $G^{\mathrm{SCC}}$};
    \end{tikzpicture}
    \caption{A graph with cycles and its SCC condensation}
    \label{fig:cyclic}
  \end{subfigure}

  \caption{Illustration of the vulnerability remediation order based on the dependency graph (Principle~\ref{princ:vulnerability_prioritization})}
  \label{fig:vulnerability_prioritization}
\end{figure}
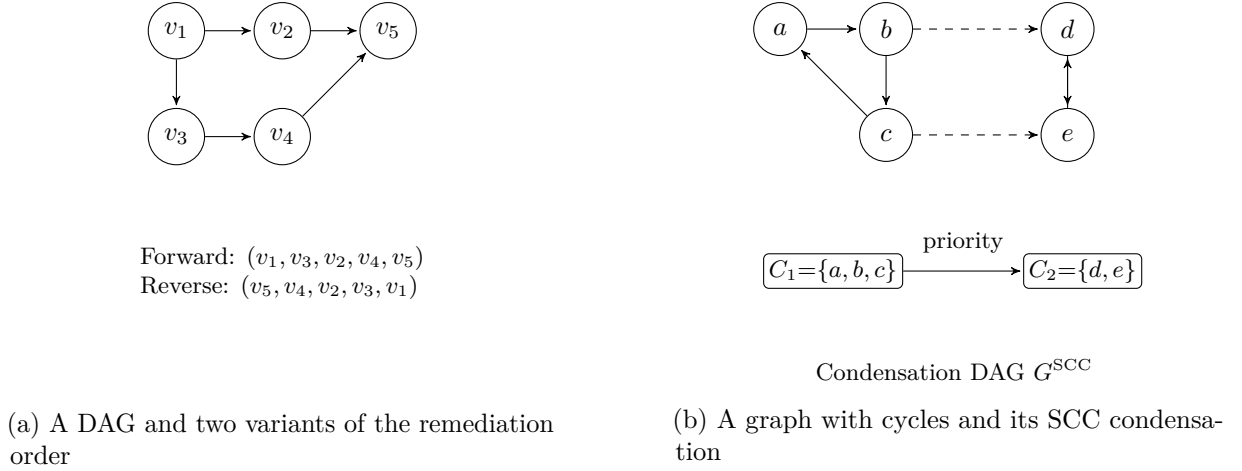

Accounting for the dependencies $\mathcal{C}_{VV}$ when planning defense measures relies on algorithms for finding strongly connected components and for topological sorting from the standard toolkit of graph theory \cite{tarjan1972, cormen2009}. Their application to the LLM vulnerability dependency graph provides a tool for prioritizing security measures.

The exploitation of a vulnerability is a sequence of actions that change the state of the system to achieve the goal of the attacker.

\begin{definition}
\label{def:vulnerability_exploitation}
\textbf{Vulnerability exploitation.}
The exploitation of a vulnerability $v \in \mathcal{V}$ by an attack of type $a \in \mathcal{A}^{\mathrm{type}}$ on a model $f_\theta$ is defined as a sequence of actions $\sigma = (\text{action}_1, \dots, \text{action}_k)$. Each action can be an interaction with the model, for example sending a query $x \in \Xspace$, a change of the environment, for example a modification of data, or an analysis of the model response. The sequence $\sigma$ leads to an undesired state $S_{\text{undesired}}$ with probability $P(\text{success} | \sigma, v, f_\theta) > \tau$ for a given threshold $\tau \in (0, 1)$. An undesired state can include the generation of an output $y \in \badY$, as defined in Section~\ref{subsec:vulnerability_metric} of Chapter~\ref{ch:ch2}, information leakage, a change of the model parameters, and so on.
\end{definition}
The definition generalizes the formalization of $\Vulfunc_{\text{local}}(f)$ to a wider range of attack consequences.

An aggregate measure of model robustness is introduced on the basis of the attack success probability.

\begin{definition}
\label{def:attack_resistance}
\textbf{Attack robustness.}
The robustness of a large language model $f_\theta$ to a set of attack types $\mathcal{A}' \subseteq \mathcal{A}^{\mathrm{type}}$ is defined as
$$
\text{AR}(f_\theta, \mathcal{A}') = 1 - \mathbb{E}_{a \sim \mathcal{W}(\mathcal{A}')} [P_{\text{success}}(a, f_\theta)],
$$
where $P_{\text{success}}(a, f_\theta) = \max_{v:(a,v)\in \mathcal{I}_{AV}} P_{\text{exploit}}(v | a, f_\theta)$ is the maximum success probability of an attack of type $a$ through any of its associated vulnerabilities, in particular the ASR for this attack, and $\mathcal{W}(\mathcal{A}')$ is a probability distribution, that is, a set of weights that sum to one, on the set of attacks $\mathcal{A}'$ that reflects their relative frequency or danger. If a uniform distribution $\mathcal{W}$ is used, then $\text{AR}(f_\theta, \mathcal{A}') = 1 - \frac{1}{|\mathcal{A}'|} \sum_{a \in \mathcal{A}'} P_{\text{success}}(a, f_\theta)$.
\end{definition}

\begin{remark}
\textbf{Relation between robustness and risk.}
Assume that the damage estimate is the same for all vulnerabilities exploited by the attacks from $\mathcal{A}'$, i.e., $\text{ImpactValue}(v) = I$ for all $v$ such that $(a, v) \in \mathcal{I}_{AV}$ for some $a \in \mathcal{A}'$, and that the weights $w(a)$ in the total risk estimate coincide with the distribution $\mathcal{W}(\mathcal{A}')$. The sum over vulnerabilities in $\text{Risk}(a, f_\theta)$ is not less than the maximum in $P_{\text{success}}(a, f_\theta)$. Hence
$$
\sum_{a \in \mathcal{A}'} w(a) \cdot \text{Risk}(a, f_\theta) \ge I \cdot \bigl(1 - \text{AR}(f_\theta, \mathcal{A}')\bigr).
$$
If, in addition, each attack type $a \in \mathcal{A}'$ exploits exactly one vulnerability, equality holds, and the total risk decreases linearly as the robustness $\text{AR}(f_\theta, \mathcal{A}')$ increases. In the general case, this relation is more complex, because the risk depends on both the success probability and the damage, and $\text{Risk}$ sums over vulnerabilities, whereas $\text{AR}$ takes the maximum.
\end{remark}

The practical application of this apparatus requires instantiating the sets $\mathcal{A}^{\mathrm{type}}$ and $\mathcal{V}$ and quantitatively estimating the relations $\mathcal{I}_{AV}$ and $\mathcal{C}_{VV}$ for a specific system. These tasks are solved within security audits and risk analysis.

\section{Risks, Limitations, and Ethical Aspects}
\label{sec:risks_ethics}

Research on vulnerabilities and the development of attacks on large language models carry dual-use risks. All experiments were conducted in a controlled environment and were aimed at identifying vulnerabilities, quantitatively assessing risks, and developing defense mechanisms.

\textbf{Legal and ethical scope of the tests.} When interacting with commercial APIs, in particular GPT-4 and Claude-3, their terms of use were strictly observed. No attacks were performed that aimed at disrupting services, at unauthorized access to the data of other users, or at the generation and dissemination of dangerous content outside the research environment. The goal of all attacks was to evaluate the vulnerabilities of the models for the subsequent development of defense mechanisms.

\textbf{Data handling.} No real personal or confidential data were used in the experiments. Datasets such as MT-Bench are publicly available and intended for research purposes. When generating prompts and analyzing responses, measures were taken to prevent the accidental processing of sensitive information.

\textbf{Generation safety.} In jailbreaking attacks or attacks aimed at generating malicious content, filtering mechanisms were applied to prevent the creation and storage of dangerous material. The results of such attacks were analyzed in terms of structure and mechanics, not content.

\textbf{Responsible disclosure.} The research results are published to inform the scientific community and developers with the goal of improving the defense of systems. Ready-made attack tools that could be used by malicious actors are not published.

The dual-use nature of work in AI security requires adherence to ethical norms at every stage of the research.

\section*{Chapter Conclusions}
The systematization of threats revealed a property that determines the structure of the entire work: attacks that appear heterogeneous at first sight, from prompt injection to trojans, reduce to the sensitivity of the learned mapping to input sequence variations in the discrete token space. This observation justifies the construction of a single robustness metric in Chapter~\ref{ch:ch2} instead of a set of specific measures for each attack class.

The mapping to OWASP Top 10 and NIST SP 800-53 revealed a gap that the existing catalogs do not close: attacks on agentic LLM systems over MCP have no separate entry and are distributed between excessive agency and insecure plugins, which motivates their separate treatment in Chapter~\ref{ch:ch5}. The asymmetry between the implantation and the detection of trojans is quantified only in Chapter~\ref{ch:ch4}, but it already indicates the limitations of post-training auditing at the classification level.

The graph model of the threat space $\mathcal{T} = (\mathcal{A}^{\mathrm{type}}, \mathcal{V}, \mathcal{I}_{AV}, \mathcal{C}_{VV})$ provides not only a taxonomy but also a principle for prioritizing defenses by the structure of vulnerability dependencies. The identified trade-offs between security, utility, confidentiality, and robustness show that, in general, not all aspects can be optimized simultaneously. For this reason, the subsequent chapters introduce the robustness-utility trade-off (RUT) metric for the quantitative choice of a defense configuration.

\chapter{Theoretical Foundations for Evaluating the Robustness of Large Language Models}\label{ch:ch2}

\section*{Chapter Overview}
Quantifying LLM robustness requires a mathematical framework that accounts for the discrete token space and the characteristics of adversarial attacks. This chapter defines robustness metrics, establishes their properties, and proves the relations between them.

Subsequent chapters use this framework to develop experimental methods and defense mechanisms.

\section{Problem Statement and Threat Model}
\label{sec:problem_statement}

To formalize the problem of evaluating the robustness of large language models, one must specify the input space, the output space, the family of admissible attacks, and the sets of goals and consequences of adversarial actions. We next formalize these components and state the general robustness evaluation problem.

\subsection{Formalization of the Robustness Evaluation Problem}
\label{subsec:robustness_formalization}

The spaces $\Xspace$, $\Yspace$, the model $f_\theta: \Xspace \to \ProbSpace{\Yspace}$ and the perturbation neighborhood $\Bepsilon{x}$, written $\mathcal{B}_\varepsilon(x)$ in this chapter, are defined in Chapter~\ref{ch:ch1} (Definitions~\ref{def:token_space},~\ref{def:token_perturbation}). In addition, we introduce the following:
\begin{itemize}
    \item $\mathfrak{A}$ is the family of attacks, where each attack $a \in \mathfrak{A}$ is given by a function $a: \Xspace \rightarrow \Xspace$ that transforms a legitimate input into a potentially malicious one. The family covers both general attack types and attacks specific to particular architectures.
    \item $\mathcal{G}$ is the set of adversary goals, and $\mathcal{R}$ is the set of consequences of a successful attack.
\end{itemize}

\begin{remark}
\label{rem:notation_ch2}
\textbf{On the notation of this chapter.}
The notation~$\mathfrak{A}$ is introduced to avoid a conflict with $\mathcal{A}^{type}$ from Chapter~\ref{ch:ch1}. The symbols $\mathcal G$ and $\mathcal R$ are also local to this chapter. The sets $\mathcal G$ and $\mathcal R$ are used implicitly in the definition of $\badY$ in Section~\ref{subsec:vulnerability_metric}: the set of undesirable outputs $\badY$ is determined by the adversary goals from $\mathcal G$ and the expected consequences from $\mathcal R$.
\end{remark}

Depending on the context, two representations of the model are used. For tasks with a discrete output, where the preservation of a verdict or class is evaluated, we apply the deterministic decision operator $h=\mathrm{Dec}\circ f$. For the analysis of the divergence between probability distributions, we use the model $f$ directly.

To remove ambiguity, we distinguish below between
(i)~the probabilistic model $f_\theta(x) \in \ProbSpace{\Yspace}$, which defines a distribution over outputs, and
(ii)~the deterministic decision operator $h_\theta = \mathrm{Dec}\circ f_\theta$, where $\mathrm{Dec}$ denotes a fixed decoding or verdict-extraction rule, such as greedy decoding, majority-label selection, or verdict parsing.
Metrics based on the divergence of distributions are defined for $f_\theta$, and metrics based on attack success or on a change of verdict are defined for $h_\theta$.

The general robustness evaluation problem is formulated as the evaluation of the worst-case risk $\text{Risk}(a, f)$ over all admissible attacks $a \in \mathfrak{A}$:
$$
\text{AdvRobustness}(f) = 1 - \sup_{a \in \mathfrak{A}} \text{Risk}(a, f),
$$
where $\text{Risk}(a, f)$ expresses the expected damage from the attack $a$ on the model $f$. A correct interpretation of $\text{AdvRobustness}(f)$ in a given range, for example $[0,1]$, requires an appropriate bound on, or normalization of, the function $\text{Risk}(a, f)$. This condition extends to all divergence measures used to compute the risk, for example $D\bigl(\cdot,\,\cdot\bigr)$: they must be bounded from above, that is $0 \le D \le C_{max} < \infty$, or normalized to the interval $D \in [0,1]$, which guarantees the correctness of the final risk estimate.

This chapter distinguishes two concepts of robustness (see the terminology convention in Chapter~\ref{ch:ch1}): $\text{AdvRobustness}(f)$ evaluates the ability of the model to resist the worst-case attack from $\mathfrak{A}$, whereas $R_{stab}(f)$ captures how little the output distribution changes under small input perturbations. The relation between them is established in the subsequent sections.

\subsection{Assumptions About the Adversary}
\label{subsec:attacker_assumptions}

Robustness evaluation requires explicit assumptions about the capabilities and limitations of the adversary. The adopted classification distinguishes the following levels of knowledge:

\begin{itemize}
    \item \textbf{Black-box:} Access is limited to the model interface: the adversary can send queries and receive responses without any information about the internal structure, parameters, or training data.
    \item \textbf{Gray-box:} Partial information about the model is known, for example the architecture, but the specific parameter values are not available.
    \item \textbf{White-box:} Full access to the model is assumed, including the architecture, parameters, and training data. The GCG method, in particular, requires this level of access to compute gradients.
\end{itemize}

Along with the level of knowledge, the following constraints matter:

\begin{itemize}
    \item \textbf{Query budget:} The maximum number of queries that the adversary can send to the model.
    \item \textbf{Computational constraints:} The available computational resources, which limit the complexity of the algorithms used.
    \item \textbf{Time constraints:} The admissible duration of the attack.
\end{itemize}

These assumptions set the boundaries of the attack space $\mathfrak{A}$ and directly affect the outcome of robustness evaluation.

\subsection{Classification of Adversarial Actions}
\label{subsec:attack_classification}

Adversarial actions against large language models can be classified into four main categories:

\begin{itemize}
    \item \textbf{Semantically equivalent reformulations:} The wording of the query changes while the semantic content is preserved. Paraphrasing with synonyms or restructuring the sentence are typical examples of such an action.

    \item \textbf{Malicious context injections:} Specially crafted instructions or context aimed at changing the behavior of the model are added to the query. This category includes prompt injection, jailbreak attacks, and other forms of context manipulation. Specific variants include the Comparative Undermining Attack (CUA) and Justification Manipulation Attack (JMA), which target LLM-as-a-Judge systems, as well as code-based attacks and the Distraction attack against defended systems.

    \item \textbf{Tokenization manipulations:} Features of the tokenization process are exploited to bypass defense mechanisms. In particular, special characters split forbidden words into fragments, as in the Word-splitting attack, and none of the individual fragments is detected by the filters.

    \item \textbf{Denial-of-service attacks:} The actions aim to degrade model performance or to stop the service completely. Queries that cause excessive resource consumption or looping are a characteristic example.
\end{itemize}

Each type of action requires dedicated robustness evaluation methods and corresponding defense measures.

\subsection{Formalization of the Security of an LLM as a Service}
\label{subsec:llm_service_security}

The security analysis of an LLM as a service covers the whole query processing chain: from the user through the prompt engineer and the LLM core to the generation of the output. Formally, it is convenient to represent this chain as a sequence of transformations, shown in Fig.~\ref{fig:chain_processing}.

\begin{figure}[ht]
\centering
\begin{tikzpicture}[->, >=stealth, node distance=4.8cm, auto, font=\small]
  \node[draw, rounded corners, align=center] (user) {User};
  \node[draw, rounded corners, align=center, right of=user] (engineer) {Prompt engineer};
  \node[draw, rounded corners, align=center, right of=engineer] (core) {LLM core};
  \node[draw, rounded corners, align=center, right of=core] (output) {Output};

  \draw (user) -- node[midway, above=5pt] {query} (engineer);
  \draw (engineer) -- node[midway, above=5pt] {processed query} (core);
  \draw (core) -- node[midway, above=5pt] {generation} (output);
\end{tikzpicture}
\caption{The query processing chain from the user to the model output}
\label{fig:chain_processing}
\end{figure}
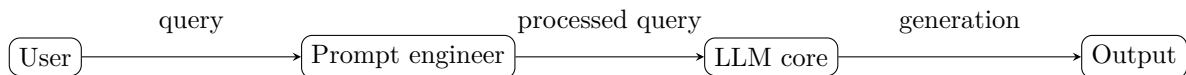

Each component of this chain can serve as a point where defense mechanisms are applied:

\begin{itemize}
    \item \textbf{User interface:} Filtering and validation of input data, detection of anomalies in query patterns.
    \item \textbf{Prompt engineering:} Transformation of user queries into safe forms, addition of protective instructions, including the system instructions in the defense architecture described in the subsequent chapters.
    \item \textbf{LLM core:} Modification of the model architecture or parameters to increase robustness.
    \item \textbf{Output data:} Filtering and sanitization of generated content, detection and blocking of malicious outputs with Python filters or LLM filters.
\end{itemize}

Robustness evaluation must cover all links of the chain and take their mutual influence into account.

\section{Formal Robustness Metrics}
\label{sec:formal_metrics}

Quantitative measurement of the robustness of large language models relies on formal metrics that reflect the degree of model vulnerability to attacks. Below we give a rigorous formalization of the main metrics: robustness to small perturbations for classification tasks $R_{class}(h)$, where $h = \mathrm{Dec} \circ f$, and for generative models $R_{stab}(f)$, the global vulnerability index $V(f)$, the risk function $\text{Risk}(a, f)$ and the adversarial robustness metric $\text{AdvRobustness}(f)$. An intuitive description of these metrics was given in Section \ref{subsec:robustness_definition_significance} and Section \ref{subsec:vulnerabilities_attacks}.

\subsection{The Metric \texorpdfstring{$R_{class}(h)$}{R\_class(h)} for Discrete Outputs}
\label{subsec:r_metric_discrete}

For models with discrete outputs in classification tasks, the metric of robustness to small perturbations $R_{class}(h)$ is defined as the expected probability that the deterministic decision operator $h = \mathrm{Dec} \circ f$ preserves its output under small input perturbations:
$$
R_{class}(h)=\mathbb{E}_{x\sim P_{\mathcal{X}}}\!\Bigl[1-\sup_{x'\in\mathcal{B}_\varepsilon(x)}
    \mathbb{I}\{h(x)\neq h(x')\}\Bigr],
$$
where $h = \mathrm{Dec} \circ f$ is the composition of the probabilistic model $f$ and the decoding operator $\mathrm{Dec}$, for example $\arg\max$. The symbol $\mathcal{B}_\varepsilon(x)$ denotes the neighborhood of the input $x$ of radius $\varepsilon$, $P_{\mathcal{X}}$ is the input data distribution, and $\mathbb{I}\{\cdot\}$ is the indicator function.

\begin{remark}
\label{rem:r_class_vs_gen}
\textbf{On the distinction between $R_{class}(h)$ and $R_{stab}(f)$.}
The metric $R_{class}(h)$ uses the decision operator $h = \mathrm{Dec} \circ f$, which maps inputs to discrete labels. This makes the indicator comparison $h(x) \neq h(x')$ well defined. The robustness literature often uses the notation $R_{class}(f)$, but here the formulas are written explicitly in terms of $h$ to avoid ambiguity. For generative LLMs, we introduce below the metric $R_{stab}(f)$, which operates directly on the probabilistic model $f$ without a decoding operator.
\end{remark}

The specific form of the neighborhood $\mathcal{B}_\varepsilon(x)$ is determined by the choice of a distance metric on $\Xspace$: the Hamming, Levenshtein, or semantic distance, as described in Section~\ref{subsec:robustness_definition_significance} (Table~\ref{tab:distance_metrics}).

The metric $R_{class}(h)$ takes values in $[0, 1]$: the value 1 corresponds to a perfectly robust model, and 0 corresponds to a completely non-robust one.

\subsection{The Metric \texorpdfstring{$R_{stab}(f)$}{R\_stab(f)} for Generative Models}
\label{subsec:r_metric_generative}

\begin{figure}[ht!]
\centering
\begin{tikzpicture}[>=stealth, font=\small, every node/.style={align=center}]
  \node[draw, rectangle, minimum width=2cm, minimum height=1cm] (f) at (0,0) {LLM core $f$};

  \node[draw, circle] (x)  at (-5,  1) {$x$};
  \node[font=\scriptsize, align=center] at (-5,  0.3) {\texttt{"How's the weather?"}};
  \node[draw, circle] (xp) at (-5, -1) {$x'$};
  \node[font=\scriptsize, align=center] at (-5, -1.7) {\texttt{"What's the weather like?"}};

  \node[draw, circle, minimum width=2.5cm, minimum height=2.5cm] (p1) at (5,  1.5) {
    \begin{tikzpicture}[scale=0.4]
      \draw[->] (0,0) -- (1.8,0);
      \draw[->] (0,0) -- (0,1.3);
      \foreach \x/\h in {0.2/1, 0.7/0.7, 1.2/0.9} {
        \draw[fill=gray!50] (\x,0) rectangle ++(0.4,\h);
      }
    \end{tikzpicture}\\[-2pt]
    $p(\cdot\mid x)$
  };
  \node[draw, circle, minimum width=2.5cm, minimum height=2.5cm] (p2) at (5, -1.5) {
    \begin{tikzpicture}[scale=0.4]
      \draw[->] (0,0) -- (1.8,0);
      \draw[->] (0,0) -- (0,1.3);
      \foreach \x/\h in {0.2/0.8, 0.7/1, 1.2/0.6} {
        \draw[fill=gray!50] (\x,0) rectangle ++(0.4,\h);
      }
    \end{tikzpicture}\\[-2pt]
    $p(\cdot\mid x')$
  };

  \draw[->] (x.east)  -- node[midway, above=4pt] {input} (f.west);
  \draw[->] (xp.east) -- node[midway, below=4pt] {input} (f.west);

  \draw[->] (f.east) -- (p1.west);
  \draw[->] (f.east) -- (p2.west);

  \draw[<->, thick] (p1) -- node[midway, fill=white, inner sep=1pt] {$D\bigl(p(\cdot\mid x),\,p(\cdot\mid x')\bigr)$} (p2);
\end{tikzpicture}
\caption{Conceptual diagram: two similar prompts $x$ and $x'$ (examples shown) produce different output distributions, whose divergence is measured by $D$.}
\label{fig:concept_divergence}
\end{figure}
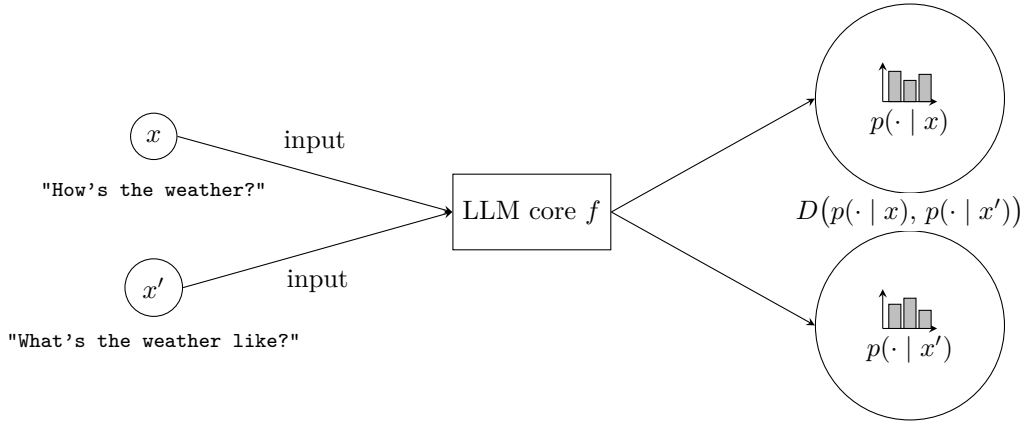

\begin{definition}
\label{def:robustness_metric}
\textbf{Robustness metric $R_{stab}(f)$ for a generative LLM.}
For generative models that autoregressively predict a probability distribution over the next token, one must rigorously define a divergence measure between the outputs for two different inputs $x$ and $x'$. To remove the ambiguity between a distribution over the whole sequence and a distribution over a single token, we define the divergence by averaging over generation steps, as shown in Fig.~\ref{fig:concept_divergence}.

For an input $x$, let a reference response $y^{\star}(x)=(y_1,\dots,y_T)$ be fixed,
for example a ground-truth answer or the answer generated by the model itself under greedy decoding.
We define the stepwise autoregressive distributions
$p_t(\cdot\mid x) := p(\cdot \mid x, y^{\star}_{<t}(x))$ and, similarly, $p_t(\cdot\mid x') := p(\cdot \mid x', y^{\star}_{<t}(x))$.
Here the comparison uses teacher forcing along a single reference trajectory $y^{\star}(x)$: the distributions $p_t(\cdot\mid x, y^{\star}_{<t}(x))$ and $p_t(\cdot\mid x', y^{\star}_{<t}(x))$ are compared under the same prefix $y^{\star}_{<t}(x)$. This isolates sensitivity to the input perturbation from divergence caused by different decoding trajectories.
Then the divergence measure between the outputs averaged over $T$ steps is defined as
$$
D_T(x,x') := \frac{1}{T}\sum_{t=1}^{T} D\bigl(p_t(\cdot\mid x),\,p_t(\cdot\mid x')\bigr),
$$
where $D$ is a divergence measure, for example the Kullback--Leibler or Jensen--Shannon divergence, between probability distributions over the vocabulary $\Vspace$. The resulting robustness metric $R_{stab}(f)$ is defined as:
$$
R_{stab}(f) := \Exp_{x\sim P_{\mathcal{X}}}\Bigl[1-\sup_{x'\in \mathcal{B}_\varepsilon(x)}
\frac{D_T(x,x')}{C_{\max}}\Bigr].
$$
\end{definition}

\begin{remark}
\label{rem:r_stab_protocol}
\textbf{Protocol dependence of $R_{stab}(f)$.}
The metric $R_{stab}(f)$ depends on the chosen comparison protocol, namely the reference trajectory $y^{\star}(x)$ and the teacher forcing scheme. This makes it a protocol-dependent characteristic of the model. The value of the metric is defined relative to a specific way of obtaining $y^{\star}$ and a fixed procedure for the stepwise comparison of distributions. When the protocol changes, for example when greedy decoding is replaced by sampling or when another reference trajectory is used, the numerical values of $R_{stab}(f)$ can change. Therefore, a single protocol for computing the metric must be fixed when models are compared.
\end{remark}

\begin{remark}
\label{rem:dt_asymmetry}
\textbf{On the asymmetry of the construction $D_T(x,x')$.}
The quantity $D_T(x,x')$ is constructed by teacher forcing along the reference trajectory $y^{\star}(x)$ of the original prompt~$x$: the distributions $p_t(\cdot\mid x, y^{\star}_{<t}(x))$ and $p_t(\cdot\mid x', y^{\star}_{<t}(x))$ are compared under the same prefix $y^{\star}_{<t}(x)$. Consequently, $D_T(x,x')$ is a \textbf{directed sensitivity measure}: in general, it does not coincide with $D_T(x', x)$, where teacher forcing would be performed along the trajectory $y^{\star}(x')$. The asymmetry is intentional: it measures the deviation of the output distribution under an input perturbation from the baseline behavior of the model on the original prompt~$x$.

If the base divergence $D$ is symmetric, as is the Jensen--Shannon divergence used in this work, then $D(p_t(\cdot\mid x, y^{\star}_{<t}),\, p_t(\cdot\mid x', y^{\star}_{<t})) = D(p_t(\cdot\mid x', y^{\star}_{<t}),\, p_t(\cdot\mid x, y^{\star}_{<t}))$ for every step~$t$, but the whole construction $D_T$ remains asymmetric because of the fixed trajectory. For tasks that require a symmetric divergence estimate, one can use the symmetrization $\overline{D}_T(x,x') = \frac{1}{2}\bigl(D_T(x,x') + D_T(x',x)\bigr)$. In the context of the definition of $R_{stab}(f)$, where $\sup_{x' \in \mathcal{B}_\varepsilon(x)} D_T(x,x')$ is considered, the directed construction is methodologically justified, because it evaluates the sensitivity of the model to a perturbation specifically relative to the original input~$x$.
\end{remark}

\begin{remark}
\textbf{Computing $R_{stab}(f)$ in practice.}
Computing $R_{stab}(f)$ in practice requires specifying the following choices.
\begin{itemize}
    \item \textbf{Choice of the reference response $y^\star$}: The method for obtaining $y^\star(x)$ must be defined explicitly. In this work, unless stated otherwise, $y^\star(x)$ is the sequence obtained by greedy decoding of the model response to the original prompt $x$ up to the maximum length $T_{\max}$.
    \item \textbf{Choice of the divergence measure $D$ and the normalization $C_{\max}$}: The Kullback--Leibler divergence ($D_{KL}$) can take infinite values if the supports of the distributions do not coincide, which creates problems for numerical stability. To avoid this and to obtain a natural normalization, all experimental evaluations in this work, unless stated otherwise, use the \textbf{Jensen--Shannon divergence} ($\JS$) as the divergence measure $D$. It is defined as:
    $$ \JS(p, q) = \frac{1}{2} D_{KL}(p \parallel m) + \frac{1}{2} D_{KL}(q \parallel m), \quad \text{where } m = \frac{1}{2}(p+q). $$
    With the base-2 logarithm, $\JS(p,q)$ always lies in the range $[0, 1]$, which allows us to take $C_{\max}=1$. With the natural logarithm, denoted $\ln$ here and below, $\JS(p,q) \in [0, \ln 2]$, and accordingly $C_{\max} = \ln 2 \approx 0.693$. In this work, unless stated otherwise, we use the base-2 logarithm and $C_{\max}=1$. This choice yields finite divergence values and makes $R_{stab}$ easier to compare across models and experiments.
\end{itemize}
\end{remark}

\begin{remark}
\label{rem:r_stab_main_metric}
Throughout the subsequent chapters, $R_{stab}(f)$ serves as the primary metric for evaluating the robustness of generative LLMs, using the definition given above. Here $p(\cdot|x)$ is the probability distribution at the output of the model for a given input $x$, which corresponds to the functional representation of the model $f(x)$ in other sections of the work.
\end{remark}

\subsection{The Global Vulnerability Index \texorpdfstring{$V(f)$}{V(f)}}
\label{subsec:vulnerability_metric}

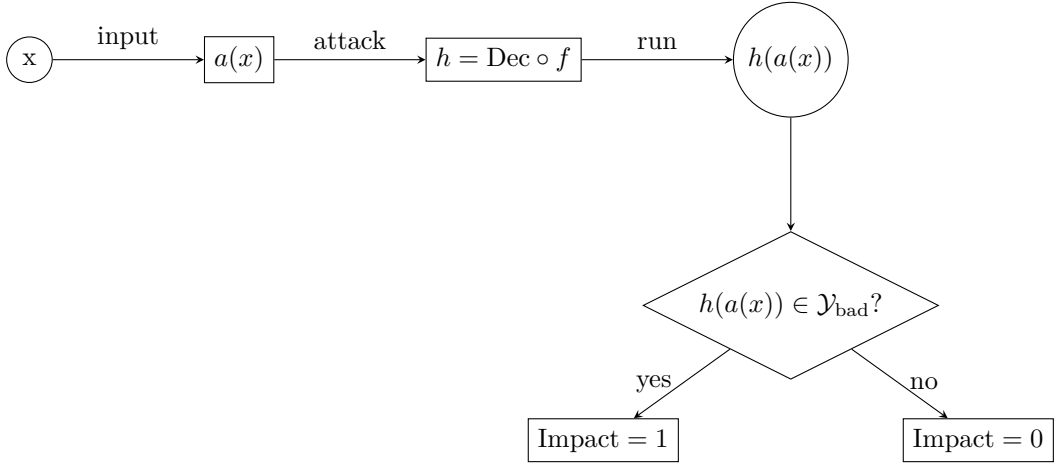
\begin{figure}[ht!]
\centering
\begin{tikzpicture}[->, >=stealth, font=\small]
  \node[draw, circle] (x)                                  {x};
  \node[draw, rectangle, right=2cm of x] (attack)          {$a(x)$};
  \node[draw, rectangle, right=2cm of attack] (model)      {$h=\mathrm{Dec}\circ f$};
  \node[draw, circle,    right=2cm of model] (y)           {$h(a(x))$};

  \node[draw, diamond, aspect=2, below=1.5cm of y] (decide)  {$h(a(x))\in \badY$?};

  \node[draw, rectangle, below left=1cm and 0.5cm of decide]  (impact1) {$\mathrm{Impact}=1$};
  \node[draw, rectangle, below right=1cm and 0.5cm of decide] (impact0) {$\mathrm{Impact}=0$};

  \draw (x) -- node[above]{input} (attack);
  \draw (attack) -- node[above]{attack} (model);
  \draw (model) -- node[above]{run} (y);

  \draw (y) -- (decide);

  \draw (decide) -- node[left]{yes}  (impact1);
  \draw (decide) -- node[right]{no} (impact0);
\end{tikzpicture}
\caption{Attack evaluation: the input $x$ is transformed by $a$, processed by $h=\mathrm{Dec}\circ f$, and assigned $\mathrm{Impact}=1$ if $h(a(x))\in\badY$ and $0$ otherwise.}
\label{fig:attack_impact}
\end{figure}

\begin{definition}
\label{def:vulnerability_metric}
\textbf{Global vulnerability index $V(h)$.}
Let $h=\mathrm{Dec}\circ f$ be the deterministic decision operator induced by the model $f$. Then the global vulnerability index is defined as
$$
V(h)=\mathbb{E}_{x\sim P_{\mathcal{X}},\,a\sim P_{\mathfrak{A}}}
    \bigl[\text{Impact}\bigl(h(a(x))\bigr)\bigr],
$$
where
$$
\text{Impact}(y)=\mathbb{I}\{y\in\badY\}.
$$
In the stochastic case, the probability of an undesirable output can be used instead of $h(a(x))$:
$$
V(f)=\mathbb{E}_{x,a}\Bigl[\mathbb{P}_{y\sim f(a(x))}(y\in\badY)\Bigr].
$$
The set of undesirable outputs $\badY$ is defined depending on the task. It can include, for example, outputs that change the verdict of an LLM-as-a-Judge under a binary or ranking decision rule, bypass the defense, or contain confidential information, as shown in Fig.~\ref{fig:attack_impact}.
\end{definition}

\begin{remark}
\label{rem:impact_disambiguation}
\textbf{On the terminological distinction between the functions Impact and ImpactValue.}
The function $\text{Impact}(y) = \mathbb{I}\{y \in \badY\}$ introduced in this definition is a binary indicator function: it takes the value 0 or~1 depending on whether the output belongs to the set~$\badY$. It must not be confused with $\text{ImpactValue}(v)$ from Chapter~\ref{ch:ch1}, Definition~\ref{def:attack_risk}. The quantity $\text{ImpactValue}(v)$ denotes a quantitative estimate of the damage from exploiting the vulnerability~$v$ on the CVSS scale and takes arbitrary values in $[0,1]$. $\text{Impact}$ is an indicator, whereas $\text{ImpactValue}$ is a continuous severity estimate. The arguments of the functions also differ: $\text{Impact}$ takes the model output~$y$, and $\text{ImpactValue}$ takes the vulnerability~$v$.
\end{remark}

\begin{remark}
\label{rem:v_f_shorthand}
\textbf{Distinction between $V(h)$ and $V(f)$.}
Definition~\ref{def:vulnerability_metric} defines the vulnerability $V(h)$ for the deterministic decision operator $h$ through the indicator function $\text{Impact}(y) = \mathbb{I}\{y \in \badY\}$. The generalization to the stochastic case $V(f)$ replaces the indicator by the probability $\mathbb{P}_{y \sim f(a(x))}(y \in \badY)$. The stochastic and deterministically decoded definitions describe different output protocols and need not agree. In the remainder of this chapter, $V(f)$ denotes $V(h_\theta)$ with $h_\theta=\mathrm{Dec}\circ f_\theta$ unless stochastic generation is explicitly specified.
\end{remark}

\begin{remark}
\label{rem:global_vs_local_vulnerability}
The global vulnerability index $V(f)$ generalizes the concept of the local vulnerability metric $\Vulfunc_{\text{local}}(f)$ presented in Chapter~\ref{ch:ch1}, Section~\ref{subsec:vulnerabilities_attacks}. The global vulnerability $V(f)$ characterizes the expected effectiveness of attacks from a given set $\mathfrak{A}$ over the whole input distribution, whereas the local vulnerability $\Vulfunc_{\text{local}}(f)$ focuses on the maximum probability of generating undesirable content under perturbations of a specific input within its $\varepsilon$-neighborhood. If the attack set $\mathfrak{A}$ in the definition of $V(f)$ consists of a single attack that aims to generate undesirable content $\badY$ by searching for the worst-case perturbation $x' \in \mathcal{B}_\varepsilon(x)$ for each $x$, then $V(f)$ is equivalent to $\Vulfunc_{\text{local}}(f)$ from Chapter \ref{ch:ch1}.
\end{remark}
The index $V(f)$ lies in $[0,1]$. A value of 0 means that an undesirable output occurs with probability zero under the specified input and attack distributions. A value of 1 means that it occurs almost surely under those distributions.

\begin{remark}
\textbf{Accounting for sampling stochasticity.}
Empirical estimation of the metrics $R_{stab}(f)$ and $V(f)$ must account for the stochasticity of the LLM response generation process, which is controlled by the sampling parameters: temperature, nucleus sampling, top-k. To ensure reproducibility and comparability of results, all sampling parameters are fixed during the experiments. Unless stated otherwise, this work uses greedy decoding, equivalent to temperature 0, which ensures deterministic and reproducible results.
\end{remark}

\subsection{The Risk Function \texorpdfstring{$\text{Risk}(a, f)$}{Risk(a, f)}}
\label{subsec:risk_function}

\begin{definition}
\label{def:risk_function}
\textbf{Attack risk ($\text{Risk}(a,h)$).}
The risk associated with a specific attack or attack type $a \in \mathfrak{A}$ on the decision operator $h=\mathrm{Dec}\circ f$ is defined as:
$$
\text{Risk}(a, h) = \mathbb{E}_{x \sim P_{\mathcal{X}}} [\mathbf{1}_{h(a(x)) \in \badY}] \cdot \text{Severity}(a, \badY),
$$
where the first factor is the success probability of the attack $a$, denoted $P_{\text{success}}(a,h)$ and estimated, for example, by the attack success rate (ASR), and $\text{Severity}(a, \badY)$ reflects the expected damage or the severity of the consequences if the attack $a$ succeeds and leads to an undesirable output.
\end{definition}

In the stochastic setting, one can use the analogous form
$$
\text{Risk}(a, f)=\mathbb{E}_{x\sim P_{\mathcal X}}\Bigl[\mathbb P_{y\sim f(a(x))}(y\in\badY)\Bigr]\cdot \text{Severity}(a,\badY).
$$

Alternatively, the attack risk can be defined as a sum over all possible vulnerabilities, as in Definition~\ref{def:attack_risk} from Chapter~\ref{ch:ch1}:
$$
\text{Risk}(a, f_\theta) = \sum_{v \in \mathcal{V} : (a, v) \in \mathcal{I}_{AV}} P_{\text{exploit}}(v | a, f_\theta) \cdot \text{ImpactValue}(v),
$$
where $\text{ImpactValue}(v)$ is a quantitative estimate of the damage, expressed, for example, in points or in monetary terms.

The definition of $\text{Risk}$ through $\text{Severity}(a, \badY)$ is an alternative formulation compared with Definition~\ref{def:attack_risk} of Chapter~\ref{ch:ch1}, which uses $\text{ImpactValue}(v)$. The first formulation expresses risk in terms of an attack's success probability and severity. The second aggregates risk over vulnerabilities. When aggregated over all vulnerabilities, the two formulations are consistent.

\begin{remark}
\label{rem:risk_normalization}
The definition of robustness through the risk function ($\text{AdvRobustness}(f) = 1 - \sup_{a \in \mathfrak{A}} \text{Risk}(a, f)$) assumes that $\text{Risk}(a, f)$ is normalized to the range $[0,1]$. For the probability-times-severity formulation, requiring $0\le\text{Severity}\le1$ is sufficient. For the sum over vulnerabilities, the aggregate risk must be bounded or normalized to $[0,1]$. Bounding each $\text{ImpactValue}(v)$ separately does not suffice.
\end{remark}

For models with a discrete output, the vulnerability is related to the metric of robustness to small perturbations through a monotonically increasing function $g$ of the non-robustness $1-R_{class}(h)$:
$$
V(h) \leq g(1-R_{class}(h)),
$$
where $h = \mathrm{Dec} \circ f$ is the deterministic decision operator, and $g$ depends on the specific definition of $V$ and $R_{class}$. For the class of localized attacks, $g(t) = t$, and the inequality $V(h) \le 1 - R_{class}(h)$ is rigorously proved in Section~\ref{subsec:r_v_theorem}. For non-localized attacks, the inequality is replaced by the approximate relation $V(h) \approx g(1-R_{class}(h))$, where $g$ defines a calibrated empirical model, which is verified on data and is not a universal rigorous bound. The difference in status between the two cases is discussed in detail in Remark~\ref{rem:empirical_v_r_model}.

\section{Properties of and Relations Between the Metrics}
\label{sec:metric_properties}

Three groups of results are established below: the monotonicity of $R_{class}$ and $R_{stab}$ with respect to a strengthened defense, an upper bound on the vulnerability $V$ through $R_{class}$, and a lemma on model composition.

\subsection{The Monotonicity Property of \texorpdfstring{$R_{class}(h)$}{R\_class(h)} and \texorpdfstring{$R_{stab}(f)$}{R\_stab(f)}}
\label{subsec:r_monotonicity}

Both metrics are monotone: if the defense of the model is strengthened, the values of $R_{class}(h)$ and $R_{stab}(f)$ do not decrease. Monotonicity serves as a consistency condition between the definition of the metric and the intuition of a ``strengthened defense''.

\begin{remark}
\textbf{Monotonicity of $R_{class}(h)$ and $R_{stab}(f)$.}
Let $h$ be the deterministic decision operator of the original model, and let $h'$ correspond to the model with a strengthened defense. Then $R_{class}(h') \geq R_{class}(h)$ and $R_{stab}(f') \geq R_{stab}(f)$.
\end{remark}
\begin{proof}
\begin{enumerate}
    \item For $R_{class}(h)$:

    Let $h'$ be the decision operator of the model with a defense strengthened relative to $h$. A strengthened defense for $R_{class}(h)$ means that the following condition holds for any input $x$:
    $$
    \sup_{x'\in\mathcal{B}_\varepsilon(x)}\mathbb{I}\{h'(x)\neq h'(x')\} \leq \sup_{x'\in\mathcal{B}_\varepsilon(x)}\mathbb{I}\{h(x)\neq h(x')\}.
    $$
    This directly gives:
    $$
    1-\sup_{x'\in\mathcal{B}_\varepsilon(x)}\mathbb{I}\{h'(x)\neq h'(x')\} \geq 1-\sup_{x'\in\mathcal{B}_\varepsilon(x)}\mathbb{I}\{h(x)\neq h(x')\}.
    $$
    Taking the expectation over the distribution $P_{\mathcal{X}}$, we obtain:
    \begin{align*}
    R_{class}(h') &= \mathbb{E}_{x \sim P_{\mathcal{X}}}\Bigl[1-\sup_{x'\in\mathcal{B}_\varepsilon(x)}\mathbb{I}\{h'(x)\neq h'(x')\}\Bigr] \\
    &\geq \mathbb{E}_{x \sim P_{\mathcal{X}}}\Bigl[1-\sup_{x'\in\mathcal{B}_\varepsilon(x)}\mathbb{I}\{h(x)\neq h(x')\}\Bigr] \\
    &= R_{class}(h).
    \end{align*}
    Hence $R_{class}(h') \geq R_{class}(h)$.

    \item For $R_{stab}(f)$:
    
    A strengthened defense $f'$ relative to $f$ for $R_{stab}(f)$ means that for any $x \in \mathcal{X}$ and the corresponding divergence measures $D_T'(x,x')$ for $f'$ and $D_T(x,x')$ for $f$, the following holds:
    $$
    \sup_{x' \in \mathcal{B}_\varepsilon(x)} \frac{D_T'(x,x')}{C_{max}} \leq \sup_{x' \in \mathcal{B}_\varepsilon(x)} \frac{D_T(x,x')}{C_{max}}.
    $$
    Then:
    $$
    1 - \sup_{x' \in \mathcal{B}_\varepsilon(x)} \frac{D_T'(x,x')}{C_{max}} \geq 1 - \sup_{x' \in \mathcal{B}_\varepsilon(x)} \frac{D_T(x,x')}{C_{max}}.
    $$
    Taking the expectation over $x \sim P_{\mathcal{X}}$ of both sides, we obtain $R_{stab}(f') \geq R_{stab}(f)$.
\end{enumerate}
\end{proof}

\begin{remark}
\label{rem:enhanced_protection_definition}
\textbf{On the definition of a ``strengthened defense'' and the scope of applicability.}
Monotonicity follows from the pointwise definition of a ``strengthened defense'' and the linearity of expectation. The result is not a standalone theorem: it verifies that the definition of the metric agrees with intuition. Without monotonicity, the metric would not allow a comparison of models with different levels of defense. Which specific transformations, including input filtering, output smoothing, and ensembling, satisfy the condition of a pointwise decrease in deviation is established experimentally in Chapters~\ref{ch:ch3}--\ref{ch:ch4}.
\end{remark}

\subsection{An Upper Bound on Vulnerability Through Robustness: \texorpdfstring{$V \le 1 - R_{class}$}{V <= 1 - R\_class}}
\label{subsec:r_v_theorem}

For models with a discrete output, where $h = \mathrm{Dec} \circ f$, we describe the relationship between $R_{class}(h)$ and $V(h)$ in two regimes. For localized attacks ($a(x) \in \mathcal{B}_\varepsilon(x)$) satisfying the assumptions of Theorem~\ref{thm:robust_vs_vuln}, the inequality $V(h) \le 1 - R_{class}(h)$ holds. Non-localized attacks, such as prompt injection and jailbreak, go beyond $\mathcal{B}_\varepsilon(x)$, and the rigorous bound does not apply to them. Instead, we propose an empirical regression model $V(h) \approx g(1 - R_{class}(h))$, whose parameters are calibrated on data (Remark~\ref{rem:empirical_v_r_model}). The separation into two regimes follows from the nature of the attacks: theoretical guarantees are possible only under a controlled perturbation budget.

\begin{theorem}
\label{thm:robust_vs_vuln}
\textbf{Upper bound on vulnerability through robustness.}
Let $h = \mathrm{Dec} \circ f$ be a deterministic decision operator and
$$
R_{class}(h)=\mathbb{E}_{x\sim P_{\mathcal{X}}}\!\Bigl[1-\sup_{x'\in\mathcal{B}_\varepsilon(x)}
    \mathbb{I}\{h(x)\neq h(x')\}\Bigr],
$$
$$
V(h)=\mathbb{E}_{x\sim P_{\mathcal{X}},\,a\sim P_{\mathfrak{A}}}
    \bigl[\text{Impact}\bigl(h(a(x))\bigr)\bigr],
$$
where
\begin{itemize}
    \item $(\mathcal{X},P_{\mathcal{X}})$ is the input space and the input distribution,
    \item $\mathfrak{A}$ is the set of adversarial transformations $a:\mathcal{X}\!\to\!\mathcal{X}$ with the distribution $P_{\mathfrak{A}}$,
    \item $\mathcal{B}_\varepsilon(x)=\{x'\in\mathcal{X}:d(x,x')\le\varepsilon\}$, where the metric $d$ defines the neighborhood,
    \item $\text{Impact}(y) = \mathbb{I}\{y \in \badY\}$ is a measurable impact function, and $\badY$ is the set of undesirable outputs,
    \item $\mathbb{I}\{\cdot\}$ is the indicator function.
\end{itemize}

Assume
\begin{enumerate}
\item[{\textrm{(A0)}}] \textbf{Correctness on clean inputs:} $\mathbb{P}_{x \sim P_{\mathcal{X}}}(h(x) \in \badY) = 0$, i.e., the model does not produce undesirable outputs on unperturbed inputs.
\item[{\textrm{(A1)}}] \textbf{Localization of attacks:} $a(x)\in\mathcal{B}_\varepsilon(x)$ almost surely.
\item[{\textrm{(A2)}}] \textbf{Link between impact and output change:} the event $\text{Impact}(h(a(x)))=1$ implies $h(x)\neq h(a(x))$, that is, an undesirable output on a perturbed input is possible only if it differs from the output on the unperturbed input.
\end{enumerate}
For the indicator impact function $\text{Impact}(y)=\mathbb{I}\{y\in\badY\}$, condition (A2) holds automatically under (A0): if $\text{Impact}(h(a(x)))=1$ and $h(x)\notin\badY$ by (A0), then $h(a(x))\in\badY$ and $h(x)\neq h(a(x))$. Under (A0)--(A2), the following inequality holds:
\begin{equation}\label{eq:simple_bound_ch2_part2}
    V(h)\;\le\;1-R_{class}(h).
\end{equation}
If condition (A0) is relaxed, that is, if the model can produce undesirable outputs on clean inputs with probability $\beta = \mathbb{P}_{x \sim P_{\mathcal{X}}}(h(x) \in \badY)$, the bound takes the form:
\begin{equation}\label{eq:relaxed_bound_ch2_part2}
    V(h)\;\le\;\beta + 1-R_{class}(h).
\end{equation}
\end{theorem}

\begin{proof}[Proof of inequality \eqref{eq:simple_bound_ch2_part2}]
By condition (A0), $h(x) \notin \badY$ holds almost surely with respect to $P_{\mathcal{X}}$.
By (A2), which holds under (A0):
$$
\text{Impact}(h(a(x))) \;\le\;\mathbb{I}\bigl\{h(x)\neq h(a(x))\bigr\}.
$$
Since $a(x) \in \mathcal{B}_\varepsilon(x)$ (by A1),
$$
\mathbb{I}\bigl\{h(x)\neq h(a(x))\bigr\} \;\le\;\sup_{x'\in\mathcal{B}_\varepsilon(x)} \mathbb{I}\{h(x)\neq h(x')\}.
$$
Consequently,
$$
\text{Impact}(h(a(x))) \;\le\;\sup_{x'\in\mathcal{B}_\varepsilon(x)} \mathbb{I}\{h(x)\neq h(x')\}.
$$
Averaging over attacks $a \sim P_{\mathfrak{A}}$ and inputs $x \sim P_{\mathcal{X}}$:
$$
V(h) = \mathbb{E}_{x,a}[\text{Impact}(h(a(x)))] \;\le\;
    \mathbb{E}_{x}\left[\sup_{x'\in\mathcal{B}_\varepsilon(x)}
            \mathbb{I}\{h(x)\neq h(x')\}\right]
    \;=\;1-R_{class}(h).
$$

\paragraph{Proof of the bound \eqref{eq:relaxed_bound_ch2_part2} under a relaxed condition (A0).}
If condition (A0) does not hold, we split the expectation into two cases:
\begin{multline*}
V(h) = \mathbb{E}_{x,a}\bigl[\text{Impact}(h(a(x))) \cdot \mathbb{I}\{h(x) \notin \badY\}\bigr] \\
    + \mathbb{E}_{x,a}\bigl[\text{Impact}(h(a(x))) \cdot \mathbb{I}\{h(x) \in \badY\}\bigr].
\end{multline*}
The first term is bounded as above: it is at most $1 - R_{class}(h)$. The second term does not exceed $\mathbb{P}(h(x) \in \badY) = \beta$, because $\text{Impact} \le 1$. In total: $V(h) \le \beta + 1 - R_{class}(h)$.

\end{proof}

\begin{remark}
\textbf{On possible generalizations.}
The bound \eqref{eq:simple_bound_ch2_part2} can be relaxed for cases where attacks leave $\mathcal{B}_\varepsilon(x)$ only with some probability, or where condition (A2) holds in a weakened form. Such generalizations lead to more complex calibrated upper bounds, and this work does not use them as rigorous theoretical results.
\end{remark}

\begin{remark}
\label{rem:theorem_applicability_scope}
\textbf{Scope of applicability of the bound.}
Inequality \eqref{eq:simple_bound_ch2_part2} applies to attacks that satisfy conditions (A1) and (A2): localized adversarial perturbations for which $a(x) \in \mathcal{B}_\varepsilon(x)$. Non-localized attacks, in particular prompt injection and jailbreak, fall outside this bound. For them, Remark~\ref{rem:empirical_v_r_model} proposes a calibrated empirical model. The two-regime framework described at the beginning of Section~\ref{subsec:r_v_theorem} covers both regimes of the attack space and forms the framework for interpreting the experimental results of Chapter~\ref{ch:ch4}.
\end{remark}

\begin{remark}
\label{rem:empirical_v_r_model}
\textbf{Empirical relationship between $V$ and $R_{class}$ when (A1) does not hold.}
When conditions (A1) and (A2) are violated, for example for prompt injection attacks, the rigorous upper bound \eqref{eq:simple_bound_ch2_part2} does not apply. However, a monotone dependence between the vulnerability $V(h)$ and the non-robustness $(1-R_{class}(h))$ is often observed. It can be approximated empirically by a function of the form:
$$
    V(h) \approx g(1-R_{class}(h)), \quad \text{where} \quad g(t)=\min\{\,1,\;k_{\mathrm{eff}}\,(t+\mathbb{E}\eta)^\alpha\}.
$$
Here $\mathbb{E}\eta$ is the mean probability that the attack leaves $\mathcal{B}_\varepsilon(x)$, and the parameters $k_{\mathrm{eff}} > 0$ and $\alpha \in (0, 1]$ are calibrated on a validation set of attacks. This relation is an empirical approximation calibrated to data. It provides a rough vulnerability estimate within the studied class of models and attacks, rather than a universal rigorous upper bound.
\end{remark}

\begin{remark}
\label{rem:empirical_proposition_status}
\textbf{Epistemological status of the empirical model of Remark~\ref{rem:empirical_v_r_model}.}
Unlike Theorem~\ref{thm:robust_vs_vuln}, which is rigorously proved under conditions (A0)--(A2), the empirical model of Remark~\ref{rem:empirical_v_r_model} represents an empirical regularity, that is, a hypothesis: the monotone dependence $V(h) \approx g(1-R_{class}(h))$ was observed experimentally but is not derived deductively from axioms. The parameters $k_{\mathrm{eff}}$ and $\alpha$ are calibrated on data and can vary across families of models and attacks.
\end{remark}

\begin{remark}
\label{rem:rstab_v_gap}
\textbf{On the relation between $R_{stab}(f)$ and $V(f)$.}
Theorem~\ref{thm:robust_vs_vuln} is proved for the metric $R_{class}(h)$, where a change of the model output is binary. For the generative metric $R_{stab}(f)$, which is based on the divergence measure $D(p_1, p_2)$, there is no direct analog of the inequality $V(h) \le 1 - R_{class}(h)$. The obstacle is that condition~(A2), which links $\text{Impact}$ to a change of a discrete output, does not transfer directly to the continuous divergence~$D_T$. For full output-sequence distributions $p$ and $q$, Pinsker's inequality $|p(\badY)-q(\badY)| \le \sqrt{\tfrac{1}{2} D_{KL}(p \| q)}$ bounds changes in event probabilities through their KL divergence. This does not directly yield a bound from $D_T(x,x')$, which averages token-level divergences along a single reference trajectory. The construction of a rigorous upper bound on $V(f)$ through $R_{stab}(f)$ remains an open problem. The empirical relation between $R_{stab}(f)$ and $R_{class}(h)$ is studied in Section~\ref{subsec:empirical_rstab_validation} of Chapter~\ref{ch:ch4}.
\end{remark}

\subsection{Lemma on Model Composition}
\label{subsec:composition_lemma}

When several LLM components are applied sequentially, the robustness of the composition is determined by the robustness of the individual components.
\begin{lemma}
\label{lemma:composition_rclass}
\textbf{On model composition for $R_{class}$.}
Let $g\colon X\to Y$ and $f\colon Y\to Z$ be deterministic models, and let the composition be $h=f\circ g\colon X\to Z$.
Then for every input $x\in X$ the pointwise inequality $S_h(x)\ge S_g(x)$ holds, where $S_u(x)=1-\sup_{x'\in\mathcal{B}_\varepsilon(x)}\mathbb{I}\{u(x)\ne u(x')\}$ for a deterministic map $u$, and, consequently,
$$
R_{class}(h)\ge R_{class}(g).
$$
In particular, if for some $x$ the first component $g$ is robust on $\mathcal{B}_\varepsilon(x)$, that is, $g(x')=g(x)$ for all $x'\in\mathcal B_\varepsilon(x)$, then the composition $h$ is also robust on this input: $h(x')=h(x)$ for all $x'\in\mathcal B_\varepsilon(x)$.
\end{lemma}
\begin{proof}
Define
$$
  S_h(x)=1-\sup_{x'\in\mathcal{B}_\varepsilon(x)}
        \mathbb{I}\{h(x)\ne h(x')\}\in\{0,1\}.
$$
Then $R_{class}(h)=\mathbb{E}_{x\sim P_{\mathcal{X}}}[S_h(x)]$.

If $S_g(x)=1$, then $g(x')=g(x)$ for all $x'\in\mathcal{B}_\varepsilon(x)$.
Consequently,
$
    (f\circ g)(x')=f\!\bigl(g(x')\bigr)=f\!\bigl(g(x)\bigr)=(f\circ g)(x),
$
and therefore $S_{f\circ g}(x)=1$. If instead $S_g(x)=0$, then trivially $S_{f\circ g}(x) \ge 0 = S_g(x)$.

Consequently, $S_{f\circ g}(x)\ge S_g(x)$ holds for all $x\in X$.
Taking expectations, we obtain $R_{class}(f\circ g)=\E[S_{f\circ g}(x)]\ge\E[S_g(x)]=R_{class}(g)$.
\end{proof}

\begin{remark}
\textbf{On the stochasticity of LLMs.}
The lemma assumes deterministic models, corresponding to greedy decoding with deterministic argmax selection at each generation step. In general, LLMs are stochastic. Extending the lemma to the stochastic case requires a definition of robustness for distributions, as in $R_{stab}$, and an account of how stochasticity propagates through the composition, which is a harder problem.
\end{remark}

In the deterministic case, the inequality $R_{class}(h)\ge R_{class}(g)$ is global: the second component $f$ can only ``merge'' distinct outputs of $g$, but cannot ``split'' coinciding ones. The robustness of the composition does not decrease. At the same time, the inequality is a worst-case bound and assumes that both components are deterministic, with a fixed direction of composition $h = f \circ g$. Arbitrary addition of components to the system does not guarantee that robustness is preserved: the reverse order of composition, stochastic decoding, or feedback loops between components make the bound inapplicable. For stochastic models, an analogous statement requires a separate justification, because noise in~$f$ can break robustness even when $g$ is robust. The relation of this result to the bound on $V_d$ for recursive chains is discussed in Remark~\ref{rem:composition_vs_chain}.

\section{Analytical Estimates and Upper Bounds}
\label{sec:analytical_bounds}

The practical use of robustness metrics requires analytical estimates and bounds that allow one to compute or approximate these metrics. This section covers upper bounds on $R_{class}$ for bounded $\varepsilon$, vulnerability to adaptive attacks, and asymptotic properties of the estimates as the dataset size grows.

\subsection{Upper Bounds on \texorpdfstring{$R_{class}(h)$}{R\_class(h)} for Bounded \texorpdfstring{$\varepsilon$}{epsilon}}
\label{subsec:r_upper_bounds}

Consider upper bounds on the robustness metric $R_{class}(h)$ for a bounded perturbation radius $\varepsilon$.
\begin{theorem}
\textbf{Upper bound on $R_{class}(h)$.}
Let $h = \mathrm{Dec} \circ f$ be a deterministic decision operator that maps from the input space $\mathcal{X}$, equipped with a metric $d$, to a discrete output space.
Then, for bounded $\varepsilon$:
$$
R_{class}(h) \leq 1 - \Pbb(h(X) \neq h(X') \text{ and } d(X, X') \leq \varepsilon)
$$
where $X$ and $X'$ are independent random variables with the same distribution $P_{\mathcal{X}}$ on $\mathcal{X}$.
\end{theorem}
\begin{proof}
By the definition of $R_{class}(h)$:
$$
R_{class}(h) = \E_{X \sim P_{\mathcal{X}}}\left[1 - \sup_{x_0 \in \B_\varepsilon(X)} \I\{h(X) \neq h(x_0)\}\right].
$$
Denote the event $A_X^c = \{\exists x_0 \in \B_\varepsilon(X) : h(X) \neq h(x_0)\}$. Then $R_{class}(h) = 1 - \Pbb_X(A_X^c)$.
Consider the event $E = \{ (X,X') \mid h(X) \neq h(X') \text{ and } d(X, X') \leq \varepsilon \}$.
If the event $E$ occurs for realizations $(x, x')$, then $x_0 = x'$ lies in $\B_\varepsilon(x)$ and $h(x) \neq h(x_0)$, so $A_x^c$ is true. The event $E$ implies $A_X^c$ for the first element of the pair, whence $\Pbb(E) \le \Pbb_X(A_X^c)$.
Then $1 - \Pbb_X(A_X^c) \le 1 - \Pbb(E)$, which gives
$$
R_{class}(h) \le 1 - \Pbb(h(X) \neq h(X') \text{ and } d(X, X') \leq \varepsilon),
$$
which completes the proof.
\end{proof}

\begin{remark}
\label{rem:r_upper_bound_interpretation}
\textbf{Interpretation of the upper bound on $R_{class}(h)$.}
The inequality $R_{class}(h) \leq 1 - \Pbb_{X,X'}(E)$ admits an intuitive interpretation. The more often two random inputs at a distance of at most $\varepsilon$ yield different outputs, the lower the upper bound on robustness. The central step of the proof is the relation $\Pbb_X(A_X^c) \geq \Pbb_{X,X'}(E)$. It means that the worst-case perturbation in $\B_\varepsilon(x)$ succeeds at least as often as a random one.
\end{remark}

\subsection{Estimating Vulnerability to Adaptive Attacks}
\label{subsec:v_lower_bounds}

Adaptive attacks adjust to a specific model and choose the best action for each input.
\begin{definition}
\textbf{Vulnerability to a pointwise-optimal adaptive attack.}
Let $\mathfrak{A}_{\text{adap}}$ be the set of adaptive attacks that can use information about the model $f$, and let $h=\mathrm{Dec}\circ f$ be the induced deterministic decision operator.

The vulnerability to a pointwise-optimal adaptive attack $V_{\text{opt\_adapt}}(f)$ is defined as:
$$
V_{\text{opt\_adapt}}(f) = \mathbb{E}_{x\sim P_{\mathcal{X}}}\bigl[\max_{a \in \mathfrak{A}_{\text{adap}}} \text{Impact}(h(a(x)))\bigr]
$$
where $\text{Impact}(y) = \mathbb{I}\{y \in \badY\}$.
\end{definition}
The quantity $V_{\text{opt\_adapt}}(f)$ models an adversary that can choose an optimal strategy from $\mathfrak{A}_{\text{adap}}$ separately for each input $x$. If $\mathfrak{A}\subseteq\mathfrak{A}_{\text{adap}}$, then, for a fixed attack distribution $P_\mathfrak{A}$ that does not depend on $x$, the mean vulnerability $V(f) = \mathbb{E}_{a\sim P_\mathfrak{A}}\mathbb{E}_{x\sim P_{\mathcal{X}}}[\text{Impact}(h(a(x)))]$ does not exceed $V_{\text{opt\_adapt}}(f)$.
The relation to the worst-case fixed attack $V_{\text{worst}}(f) = \sup_{a \in \mathfrak{A}} \mathbb{E}_x [\text{Impact}(h(a(x)))]$ is less straightforward. By Jensen's inequality, the maximum inside the expectation is not less than the maximum outside it: $V_{\text{opt\_adapt}}(f) \ge \max_{a \in \mathfrak{A}_{\text{adap}}} \mathbb{E}_x [\text{Impact}(h(a(x)))]$. Under the same inclusion, this gives $V_{\text{opt\_adapt}}(f) \ge V_{\text{worst}}(f)$.

\begin{remark}
\label{rem:adaptive_attack_bound_updated}
The metric $V_{\text{opt\_adapt}}(f)$ is useful for estimating the ``ceiling'' of vulnerability when the adversary has substantial adaptive capabilities. A comparison of $V_{\text{opt\_adapt}}(f)$ with $V(f)$ for some distribution $P_\mathfrak{A}$ over attacks, or with $V_{\text{worst}}(f)$, can show how much the adversary benefits from adapting the attack to each input.
\end{remark}

For any attack distribution supported on $\mathfrak{A}_{\text{adap}}$, the average vulnerability is bounded above by $V_{\text{opt\_adapt}}(f)$, the expected maximum impact over the available adaptive strategies.

\subsection{Asymptotic Properties as the Dataset Size \texorpdfstring{$\to \infty$}{-> infinity}}
\label{subsec:asymptotic_properties}

As the size of the dataset used for the empirical estimation of robustness metrics grows, the asymptotic properties of the resulting estimates become essential.

\begin{proposition}
\textbf{Convergence of the empirical estimate $\hat{R}_{class}$ to $R_{class}$.}
Let $\hat{R}_{class,n}(h)$ be the empirical estimator of $R_{class}(h)$ based on $n$ inputs $x_1,\ldots,x_n$ drawn independently from $P_{\mathcal{X}}$:
$$
\hat{R}_{class,n}(h) = \frac{1}{n} \sum_{i=1}^n \Big[1-\sup_{x'\in\mathcal{B}_\varepsilon(x_i)}\mathbb{I}\{h(x_i)\neq h(x')\}\Big]
$$
Then, as $n \to \infty$:
$$
\hat{R}_{class,n}(h) \xrightarrow{P} R_{class}(h)
$$
where $\xrightarrow{P}$ denotes convergence in probability.
\end{proposition}
\begin{proof}
Denote $Z_i = 1-\sup_{x'\in\mathcal{B}_\varepsilon(x_i)}\mathbb{I}\{h(x_i)\neq h(x')\}$. The quantities $Z_i$ are independent identically distributed random variables with $\mathbb{E}[Z_i] = R_{class}(h)$ and $Z_i \in [0, 1]$.

By the law of large numbers:
$$
\hat{R}_{class,n}(h) = \frac{1}{n} \sum_{i=1}^n Z_i \xrightarrow{P} \mathbb{E}[Z_i] = R_{class}(h)
$$
Moreover, by the central limit theorem:
$$
\sqrt{n}(\hat{R}_{class,n}(h) - R_{class}(h)) \xrightarrow{d} \mathcal{N}(0, \sigma^2)
$$
where $\sigma^2 = \text{Var}(Z_i)$ and $\xrightarrow{d}$ denotes convergence in distribution.

Hence $\hat{R}_{class,n}(h) \xrightarrow{P} R_{class}(h)$ as $n \to \infty$, which completes the proof. Analogous statements hold for the empirical estimates of $R_{stab}(f)$ and $V(f)$ under appropriate conditions on the moments of the random variables.
\end{proof}

The proposition is a direct consequence of the law of large numbers for bounded random variables. It guarantees that, for a sufficiently large dataset, the empirical estimate of the robustness metric approximates the true value arbitrarily well with probability arbitrarily close to one.

\section{Robustness of Compositional Systems}
\label{sec:compositional_systems}

Systems based on large language models often combine several interacting components. This section addresses the robustness evaluation of LLM pipelines with retriever components, recursive chains of LLM agents, and majority-vote ensembles.

\subsection{LLM + Retriever Pipelines}
\label{subsec:llm_retriever}

Pipelines that combine an LLM with an information retrieval component, or retriever, are common in question answering systems. The robustness of such systems is determined by the robustness of both the LLM and the retriever component.

\begin{proposition}
\label{thm:llm_retriever_robustness}
\textbf{Compositional robustness of LLM + retriever.}
Let $r: \mathcal{X} \to 2^{\mathcal{D}}$ be a retriever component that retrieves relevant documents from the database $\mathcal{D}$ and satisfies the Lipschitz condition:
$$
\forall x, x' \in \mathcal{X}: d_D(r(x), r(x')) \leq K_r \cdot d_X(x, x')
$$
where $K_r > 0$ is the Lipschitz constant of the retriever, $d_X$ is a metric on the query space, and $d_D$ is a metric on the space of document sets.

Let $f: \mathcal{X} \times 2^{\mathcal{D}} \to \mathcal{Y}$ be a language model that generates a response based on the query and the retrieved documents and satisfies the Lipschitz conditions:
\begin{multline*}
\forall x \in \mathcal{X}, \forall D_1, D_2 \in 2^{\mathcal{D}}: \\
d_Y(f(x, D_1), f(x, D_2)) \leq K_{f,D} \cdot d_D(D_1, D_2)
\end{multline*}
\begin{multline*}
\forall x_1, x_2 \in \mathcal{X}, \forall D \in 2^{\mathcal{D}}: \\
d_Y(f(x_1, D), f(x_2, D)) \leq K_{f,X} \cdot d_X(x_1, x_2)
\end{multline*}
where $K_{f,D}, K_{f,X} > 0$ are the Lipschitz constants of the LLM with respect to the documents and the query, and $d_Y$ is a metric on the outputs.

Then the compositional system $s(x) = f(x, r(x))$ satisfies the Lipschitz condition with the constant $K_s \leq K_{f,D} \cdot K_r + K_{f,X}$:
$$
\forall x, x' \in \mathcal{X}: d_Y(s(x), s(x')) \leq (K_{f,D} \cdot K_r + K_{f,X}) \cdot d_X(x, x')
$$
\end{proposition}
\begin{proof}
Consider two arbitrary queries $x, x' \in \mathcal{X}$. We bound the distance between the corresponding outputs of the system:
\begin{align*}
d_Y(s(x), s(x')) &= d_Y(f(x, r(x)), f(x', r(x'))) \\
&\leq d_Y(f(x, r(x)), f(x, r(x'))) + d_Y(f(x, r(x')), f(x', r(x'))) \\
&\quad (\text{triangle inequality})
\end{align*}
By the Lipschitz condition for $f$ in the second argument, that is, the documents:
$$
d_Y(f(x, r(x)), f(x, r(x'))) \leq K_{f,D} \cdot d_D(r(x), r(x')).
$$
By the Lipschitz condition for $r$:
$$
d_D(r(x), r(x')) \leq K_r \cdot d_X(x, x').
$$
Consequently, the first term does not exceed $K_{f,D} \cdot K_r \cdot d_X(x, x')$.

By the Lipschitz condition for $f$ in the first argument, that is, the query:
$$
d_Y(f(x, r(x')), f(x', r(x'))) \leq K_{f,X} \cdot d_X(x, x').
$$
Then:
\begin{align*}
d_Y(s(x), s(x')) &\leq K_{f,D} \cdot K_r \cdot d_X(x, x') + K_{f,X} \cdot d_X(x, x') \\
&= (K_{f,D} \cdot K_r + K_{f,X}) \cdot d_X(x, x')
\end{align*}
It follows that the system $s(x) = f(x, r(x))$ satisfies the Lipschitz condition with the constant $K_s \leq K_{f,D} \cdot K_r + K_{f,X}$, which completes the proof.
\end{proof}

\begin{remark}
\label{rem:compositional_robustness_formulation}
\textbf{Applying the bound to composite systems.}
Applying Proposition~\ref{thm:llm_retriever_robustness} in practice requires choosing the metric $d_D$ and estimating the Lipschitz constants.

\textbf{The metric $d_D$.} The Jaccard distance $d_D(D_1, D_2) = 1 - |D_1 \cap D_2| / |D_1 \cup D_2|$ accounts only for the composition of the document sets. When the differences are semantic, an embedding-based distance is preferable: the cosine measure between the averaged vectors of the sets or the Earth Mover's Distance between sets of embeddings.

\textbf{Lipschitz constants.} An analytical computation of $K_r$, $K_{f,D}$, $K_{f,X}$ for neural networks is, as a rule, impossible. The alternative is an empirical estimate on a validation set $\mathcal{D}_{val}$. For the retriever component, pairs of close queries $(x_i, x'_i)$ are generated and the ratio $\rho_i = d_D(r(x_i), r(x'_i)) / d_X(x_i, x'_i)$ is computed. The estimate $\hat{K}_r$ is taken as the maximum of $\rho_i$, or as the 99th percentile for robustness to outliers. The constants $\hat{K}_{f,D}$ and $\hat{K}_{f,X}$ are estimated similarly on the corresponding spaces. A plug-in estimate for the composite system is $\hat{K}_s:=\hat{K}_{f,D}\hat{K}_r+\hat{K}_{f,X}$. This is an empirical sensitivity estimate, not a certified global Lipschitz bound: the sample maximum of $\rho_i$ never exceeds the true constant $K_r$, and the 99th percentile does not bound even the observed ratios. The contributions $\hat{K}_{f,D}\hat{K}_r$ and $\hat{K}_{f,X}$ can be compared to assess which pathway dominates this estimate.
\end{remark}

\subsection{Recursive Chains of LLM Agents}
\label{subsec:recursive_chains}

In recursive chains of LLM agents, one agent's output becomes another's input, creating a risk that errors accumulate and become amplified.

\begin{proposition}
\label{thm:chain_vulnerability}
\textbf{Upper bound on vulnerability in recursive chains under the cascading failure model.}
Let $V_1$ denote the vulnerability of a single LLM agent, defined as the probability that it returns an incorrect result under attack. We introduce error indicators $E_i \in \{0,1\}$, where $E_i = 1$ means an incorrect result of the $i$-th agent under attack. Consider the \textbf{cascading fatal failure model} under the following assumptions:
\begin{enumerate}
    \item[(C1)] \textbf{Independence of errors:} the indicators $E_1, \ldots, E_d$ are modeled as statistically independent random variables on a common probability space of the attack;
    \item[(C2)] \textbf{Homogeneity:} all $d$ agents have the same error probability $\mathbb{P}(E_i = 1) = V_1$;
    \item[(C3)] \textbf{No compensation:} the result of the chain is incorrect if and only if at least one agent produced an incorrect result, that is, subsequent links do not compensate for the error of the preceding ones.
\end{enumerate}
Under assumptions (C1)--(C3), the vulnerability index of a chain of $d$ sequential agents equals:
$$
V_d = 1 - (1 - V_1)^d.
$$
If assumption (C3) is weakened to a one-sided condition, namely that the correctness of all agents implies the correctness of the chain but partial compensation of errors is possible, then, with (C1)--(C2) preserved, the equality becomes an upper bound:
$$
V_d \leq 1 - (1 - V_1)^d.
$$
\end{proposition}
\begin{proof}
Consider a chain of $d$ sequential LLM agents $f_1, f_2, \ldots, f_d$, where the output of $f_i$ is the input of $f_{i+1}$, and the composition is given by the function $F_d = f_d \circ f_{d-1} \circ \ldots \circ f_1$. Denote $A = \{$the result of the chain is incorrect$\}$ and $B = \{E_1 = 1 \text{ or } \ldots \text{ or } E_d = 1\}$, the event that at least one agent makes an error.

By independence (C1) and homogeneity (C2), the probability that no agent makes an error equals $\mathbb{P}(\overline{B}) = \prod_{i=1}^d (1 - V_1) = (1 - V_1)^d$, whence $\mathbb{P}(B) = 1 - (1 - V_1)^d$.

Under the full condition (C3), the events $A$ and $B$ coincide, so $V_d = \mathbb{P}(A) = \mathbb{P}(B) = 1 - (1 - V_1)^d$.

Under the weakened (C3), the inclusion $A \subseteq B$ is preserved: if no agent makes an error, the result of the chain is correct, that is, $\overline{B} \subseteq \overline{A}$. Hence $V_d = \mathbb{P}(A) \le \mathbb{P}(B) = 1 - (1 - V_1)^d$.

For a heterogeneous chain with vulnerabilities $V_1^{(i)}$, similarly $V_d = 1 - \prod_{i=1}^d (1 - V_1^{(i)})$ under the full (C3), and if $V_1^{(i)} \leq V_1$ for all $i$, then $V_d \leq 1 - (1 - V_1)^d$.
\end{proof}

Under assumptions (C1)--(C3), vulnerability is nondecreasing with chain depth and increases strictly when $0<V_1<1$, which places stringent requirements on individual agents.

\begin{remark}
\label{rem:chain_vulnerability_inequality}
\textbf{Scope of applicability of the bound $V_d \leq 1 - (1 - V_1)^d$.}
The bound relies on the simplifying assumptions (C1)--(C3), which hold only approximately in practice. Models of the same family often share systematic biases, which produces correlated errors and violates condition C1. If the error indicators are positively associated, the probability that no agent fails is at least $\prod_{i}(1-V_1^{(i)})$, so the actual vulnerability does not exceed the value given by the independence model; for $d = 2$ this holds for any positive correlation. For $d \ge 3$, pairwise correlations alone do not determine the direction of the deviation. Heterogeneity of the agents, which violates condition C2, leads to the expression $V_d = 1 - \prod_{i=1}^d (1 - V_1^{(i)})$. If $V_1^{(i)} \leq V_1$ for all $i$, it reduces to $V_d \leq 1 - (1 - V_1)^d$. Meanwhile, subsequent agents can partially compensate for the errors of the preceding ones. This violates condition C3 and lowers the actual vulnerability below the upper bound. The formula $V_d \leq 1 - (1 - V_1)^d$ is a coarse upper bound under the assumption of cascading failure and independent errors, not a universal law for arbitrary agent chains.
\end{remark}

\begin{remark}
\label{rem:composition_vs_chain}
\textbf{Relation between the composition lemma and the bound on $V_d$.}
Lemma~\ref{lemma:composition_rclass} states that the robustness of a composition $R_{class}(f\circ g)\ge R_{class}(g)$ does not decrease when a deterministic component is added. Under assumptions (C1)--(C3), Proposition~\ref{thm:chain_vulnerability} gives $V_d=1-(1-V_1)^d$, which is nondecreasing with chain depth. Under the weakened condition (C3), only the upper bound is guaranteed to be nondecreasing. There is no contradiction between these results. The lemma describes a property of one fixed composition relative to its first component, whereas the bound on $V_d$ models the probability that at least one link fails in a chain of $d$ independent agents. The first result characterizes the structural monotonicity of robustness, the second the accumulation of error probability in a stochastic scenario.
\end{remark}

\subsection{Majority-Vote Ensembles}
\label{subsec:vote_ensemble}

An ensemble that combines model outputs by majority voting can improve LLM-system robustness, primarily in classification tasks with discrete outputs. Committees of models have also proved effective in increasing the robustness of LLM-as-a-Judge systems to prompt injections.

\begin{theorem}
\label{thm:vote_ensemble_stability}
\textbf{Robustness of a majority-vote ensemble in the presence of a robust majority with a common verdict.}
Let $f_1,\dots,f_k$ be $k$ models with a discrete output, $k$ odd, and let $F_k$ be the majority voting mechanism.
For an input $x$, denote by $v_i(x)=f_i(x)$ the verdict of the $i$-th model and let
$$
S_i(x)=\mathbb I\{f_i(x)=f_i(x') \text{ for all } x'\in\mathcal B_\varepsilon(x)\}.
$$
Assume that there exists a label $c(x)$ such that for at least $\lceil k/2\rceil$ models both of the following hold:
\begin{enumerate}
    \item $S_i(x)=1$;
    \item $f_i(x)=c(x)$.
\end{enumerate}
Then
$$
F_k(x')=c(x)\qquad \text{for all } x'\in\mathcal B_\varepsilon(x),
$$
that is, the ensemble is robust on the input $x$.
\end{theorem}
\begin{proof}
By assumption, there exist at least $\lceil k/2\rceil$ models that preserve the same verdict $c(x)$ on the whole set $\mathcal B_\varepsilon(x)$. Consequently, for any $x'\in\mathcal B_\varepsilon(x)$, at least $\lceil k/2\rceil=(k+1)/2$ of the $k$ models, that is, a strict majority, output $c(x)$. Since $k$ is odd, majority voting uniquely returns $c(x)$ for all $x'\in\mathcal B_\varepsilon(x)$. The theorem is proved.
\end{proof}

\begin{corollary}
\label{cor:condorcet_ensemble}
\textbf{Committee accuracy under the Condorcet jury theorem.}
Let $c^*(x)$ denote the true verdict on the input $x$, and let $f_1,\dots,f_k$ be $k$ models with odd $k$, each of which gives a verdict $f_i(x)$ that coincides with $c^*(x)$ with probability $r > 0.5$. Let the errors of the models be mutually independent. Then the probability that the majority vote of the committee yields $c^*(x)$ satisfies:
$$
P_{ens}(k,r) = \sum_{i=\lceil k/2 \rceil}^{k} \binom{k}{i} r^i (1-r)^{k-i} \xrightarrow[k \to \infty]{k \text{ odd}} 1 \quad \text{for } r > 0.5.
$$
\end{corollary}
\begin{proof}
The number of correct verdicts $S_k = \sum_{i=1}^k \mathbb{I}\{f_i(x) = c^*(x)\}$ has the distribution $\mathrm{Bin}(k, r)$. For $r > 0.5$, by the law of large numbers $S_k/k \xrightarrow{P} r > 1/2$, whence $\mathbb{P}(S_k \ge \lceil k/2 \rceil) \to 1$.
\end{proof}

\begin{remark}
While Theorem~\ref{thm:vote_ensemble_stability} guarantees robustness in the presence of a deterministic robust majority, Corollary~\ref{cor:condorcet_ensemble} reveals the probabilistic mechanism by which the accuracy of a committee increases with its size. The independence of verdicts is an idealization here. On overlapping training data, models acquire common biases, and full independence is unattainable. Meanwhile, the experimental data of Chapter~\ref{ch:ch4} show that heterogeneous committees retain a sufficiently low error correlation for the method to be effective in practice.
\end{remark}

Under the assumptions of Corollary~\ref{cor:condorcet_ensemble}, committee accuracy approaches one as the odd number of models $k$ increases, provided that each model has accuracy $r>0.5$. This accuracy result alone does not establish an increase in $R_{class}$. A statement about $R_{class}$ requires a stronger assumption: if, for each model, the event $\{f_i(x')=c^*(x) \text{ for all } x'\in\mathcal B_\varepsilon(x)\}$ has probability $\rho>0.5$ under $P_{\mathcal{X}}$ and these events are mutually independent, then Theorem~\ref{thm:vote_ensemble_stability} gives $R_{class}(F_k)\ge P_{ens}(k,\rho)$.

\begin{remark}
\label{rem:vote_ensemble_generative}
\textbf{Applicability of the theorem to $R_{stab}(f)$.}
The proved theorem applies rigorously to $R_{class}(h)$, where model outputs are discrete and unambiguously comparable. For generative models, a separate probability-averaging argument uses the joint convexity of the Jensen--Shannon divergence, discussed in Remark~\ref{rem:jsd_convexity_gap} of Chapter~\ref{ch:ch3}. If all component models and the ensemble are evaluated along the same reference trajectory, with the same generation horizon and fixed mixture weights, the ensemble's $D_T$ is at most the corresponding weighted mean of the component divergences. This argument concerns probability averaging, not voting over decoded verdicts. The experimental data of Chapter~\ref{ch:ch4} confirm that ensembling increases the robustness of LLM-as-a-Judge systems, where the final verdicts are aggregated.
\end{remark}

\section{Comparison with Existing Metrics}
\label{sec:metrics_comparison}

We compare the proposed robustness metrics with existing measures of language-model quality and safety.

\subsection{Completeness Criteria}
\label{subsec:completeness_criteria}

The completeness and applicability of robustness metrics are assessed against the following criteria:

\begin{itemize}
    \item \textbf{Invariance:} The metric is invariant to inessential transformations of the model or the data.
    \item \textbf{Sensitivity:} Changes that affect the robustness of the model are reflected in the values of the metric.
    \item \textbf{Computability:} The metric can be computed efficiently in practice.
\end{itemize}

We compare the proposed metrics $R_{stab}(f)$ and $V(f)$ with existing measures: BLEU/ROUGE-like metrics~\cite{papineni2002bleu,lin2004rouge}, which score the overlap between an output and a reference text, and a perplexity-based score, which thresholds the perplexity of the input as in the Perplexity Check filter of Chapter~\ref{ch:ch3}~\cite{jain_baseline_2023}. The ratings in Table~\ref{tab:metrics_comparison_updated_ch2} are qualitative assessments against the three criteria above, not measured values:

\begin{table}[ht]
\centering
\small
\begin{tabular}{|l|c|c|c|}
\hline
\textbf{Metric} & \textbf{Invar.} & \textbf{Sens.} & \textbf{Comput.} \\
\hline
$R_{stab}(f)$ & High & High & Med./High \\
$V(f)$ & High & High & High \\
BLEU/ROUGE-like metrics & Medium & Low & High \\
Perplexity-based score & Low & Medium & High \\
\hline
\end{tabular}
\caption{Qualitative comparison of metric properties.}
\label{tab:metrics_comparison_updated_ch2}
\end{table}

Table~\ref{tab:metrics_comparison_updated_ch2} presents a qualitative assessment of the proposed metrics and comparison measures in terms of invariance, sensitivity, and computational feasibility. The proposed metrics rate high on invariance and sensitivity. Their exact computation, however, requires more resources, because it involves either an exhaustive search for the $\sup$ or a simulation of many attacks.

\subsection{Application of the Metrics: Model Comparison}
\label{subsec:case_study}

Chapter~\ref{ch:ch4} uses experimental data to illustrate how $R_{stab}$, $V$, and $\text{Risk}$ are computed.

\subsection{Relation Between the Theoretical Metric \texorpdfstring{$R_{stab}(f)$}{R\_stab(f)} and Empirical Indicators}
\label{subsec:r_stab_empirical_bridge}

The metric $R_{stab}(f)$ is defined through the stepwise divergence of probability distributions and, in general, requires access to the internal distributions of the model at each generation step. For proprietary models (GPT-4, Claude-3-Opus), such access is not available, which makes a direct computation of $R_{stab}(f)$ impossible in the black-box setting. Even for open-weight models, a direct computation of $D_T(x,x')$ for all $x' \in \mathcal{B}_\varepsilon(x)$ is computationally expensive.

The experiments in Chapter~\ref{ch:ch4} therefore use the empirical estimate
$$
\widehat{R}_{emp}(f) = 1 - \mathrm{ASR}(f, \mathfrak{A}_{eval}),
$$
which estimates $1-V(h)$ for the fixed attack family $\mathfrak{A}_{eval}$. Under the assumptions below, this population quantity is an upper bound on $R_{class}(h)$. Its use as a proxy for $R_{stab}(f)$ requires separate empirical justification.

\begin{proposition}
\label{prop:r_emp_proxy}
\textbf{Justification of $\widehat{R}_{emp}$ as a proxy for $R_{class}$.}
For tasks with a discrete output, including LLM-as-a-Judge with a finite set of verdicts, let the following conditions hold:
\begin{enumerate}
    \item the pairs $(x_i,a_i)$ are sampled i.i.d. from $P_{\mathcal{X}}\otimes P_{\mathfrak{A}}$, with $P_{\mathfrak{A}}$ supported on $\mathfrak{A}_{eval}$;
    \item the success of an attack is determined by the indicator $\mathrm{Impact}(h(a(x))) = \mathbb{I}\{h(a(x)) \in \badY\}$ from Theorem~\ref{thm:robust_vs_vuln};
    \item assumptions (A0)--(A1) of the same theorem hold.
\end{enumerate}
Then the estimate $\widehat{R}_{emp}(f) = 1 - \mathrm{ASR}$ is consistent for the population quantity $1 - V(h)$, and for the population limit $1 - V(h) \ge R_{class}(h)$ holds. For a finite sample of size $n$, the inequality $\widehat{R}_{emp}(f) \ge R_{class}(h)$ holds up to a random error of order $O_p(n^{-1/2})$.
\end{proposition}

\begin{proof}
By definition, $\mathrm{ASR} = \frac{1}{n}\sum_{i=1}^{n} \mathrm{Impact}(h(a_i(x_i)))$ is a mean of i.i.d. bounded terms with expectation $V(h) = \mathbb{E}[\mathrm{Impact}(h(a(x)))]$. By the strong law of large numbers, $\mathrm{ASR} \to V(h)$ almost surely, whence $\widehat{R}_{emp}(f) = 1 - \mathrm{ASR} \to 1 - V(h)$, which proves consistency. By Theorem~\ref{thm:robust_vs_vuln}, under (A0)--(A1), $V(h) \le 1 - R_{class}(h)$ holds, that is, $1 - V(h) \ge R_{class}(h)$ for the population quantities. By the central limit theorem, the deviation of $\widehat{R}_{emp}(f)$ from its limit is of order $O_p(n^{-1/2})$, which gives the stated finite-sample bound.
\end{proof}

The proposition gives a theoretical interpretation of the empirical results of Chapter~\ref{ch:ch4}: the population limit of $1-\mathrm{ASR}$ is $1-V(h)$, which is an upper bound on $R_{class}(h)$ under the stated assumptions, while $\mathrm{ASR}$ is a consistent estimator of $V(h)$ for the attack family $\mathfrak{A}_{eval}$. For non-localized attacks such as prompt injection, the relation is established through the calibrated model, see Remark~\ref{rem:empirical_v_r_model}.

\section{Chapter Conclusions}
\label{sec:conclusions_ch2_part2} 

The problem of evaluating LLM robustness is formalized through the definition of the input space, the output space, the attack family, and the sets of goals and consequences. The chapter distinguishes adversarial robustness $\text{AdvRobustness}(f)$ from robustness to small input perturbations, measured by $R_{stab}(f)$. Formal metrics are proposed: $R_{class}(h)$ for classification tasks, $R_{stab}(f)$ for generative models, $V(f)$ for global vulnerability, and $\text{Risk}(a, f)$ for quantifying the risk of a specific attack. For practical application, the normalization parameter $C_{max}$ and the choice of the divergence measure $D$ are specified.

The monotonicity of $R_{class}(h)$ and $R_{stab}(f)$ with respect to a strengthened defense is verified as a consistency condition of the definitions. Theorem~\ref{thm:robust_vs_vuln} establishes $V(h)\le1-R_{class}(h)$ for localized attacks when clean inputs produce no undesirable outputs, and $V(h)\le\beta+1-R_{class}(h)$ for localized attacks when the clean-input failure probability is $\beta$. The composition property for $R_{class}$ is also proved. Analytical upper bounds for $R_{class}(h)$ are obtained, and the consistency of empirical estimates is confirmed as a consequence of the law of large numbers for bounded random variables.

For compositional systems, we derive robustness bounds for LLM--retriever pipelines using Lipschitz constants, an upper bound on the vulnerability of recursive agent chains, and conditions for the robustness of majority-vote ensembles with respect to $R_{class}$.

The proposed framework has several limitations. Exact computation of the metrics for large models requires substantial computational resources. The rigorous bound $V \le 1 - R_{class}$ covers only localized attacks, whereas the transfer of the results to the metric $R_{stab}(f)$ for generative models remains an open problem. The choice of adequate distance measures in the discrete token space is a separate unsolved problem. Chapter~\ref{ch:ch4} presents the experimental validation of these metrics and their application to real models.

\chapter{Methods for Improving the Robustness of Large Language Models}\label{ch:ch3}

\section*{Chapter Overview}
The local sensitivity of a model to input perturbations determines its robustness. In this chapter, we establish a necessary condition and a sufficient condition for high values of the metric $R_{stab}(f)$ that connect the local properties of a model with its global robustness. Building on these results, the chapter systematizes architectural modifications, training methods, and multi-layer defense mechanisms that operate in the discrete token space. We also consider the trade-off between robustness and the functional utility of the model.

Known approaches are presented in Sections~\ref{sec:robustness_foundations_ch3_part3}--\ref{sec:training_robustness_methods_ch3_part3}. Section~\ref{sec:algorithms_complexity_ch3_part3} presents the methods proposed by the author of this thesis.

\section{Theoretical Foundations of Robustness Improvement}
\label{sec:robustness_foundations_ch3_part3}

This section formalizes the central concepts of robustness for large language models, which operate in the discrete token space $\Xspace = \Vspace^*$. Building on the notions of robustness introduced in Chapter~\ref{ch:ch1} (Section~\ref{sec:theoretical_foundations}) and Chapter~\ref{ch:ch2} (Section~\ref{subsec:r_metric_generative}), we trace the relationship between the local sensitivity of a model and its overall robustness.

\subsection{Formalization of Robustness and Local Sensitivity}
\label{subsec:robustness_formalization_ch3_part3}

The robustness metric $R_{stab}(f)$ for generative models is introduced in Definition~\ref{def:robustness_metric} of Chapter~\ref{ch:ch2}. In all theoretical statements of this chapter, the divergence of outputs is understood as $D_T(x,x')$ from the same definition. The shorthand notation $D\bigl(p(\cdot\mid x),\,p(\cdot\mid x')\bigr)$ appears below only where the two distributions are specified explicitly, namely in Lemma~\ref{lem:ensemble_robustness_ch3_part3} and in the regularizers of Section~\ref{subsec:robust_training_ch3_part3}. Elsewhere it is not used, to avoid ambiguity between the per-step distribution and the full distribution over sequences. A value of $R_{stab}(f)$ close to~1 corresponds to high robustness.

The local sensitivity of a model reflects the degree to which the output distribution changes under small input perturbations. In discrete spaces, the analog of the Lipschitz constant is the maximal shift of the divergence $D_T$ under the smallest admissible changes of the input, in particular for $d(x,x')=1$ in the Levenshtein metric.

\begin{lemma}
\label{lem:robustness_necessary_condition_ch3_part3}
\textbf{Necessary condition for robustness.}
If a model $f: \Xspace \rightarrow \ProbSpace{\Yspace}$ has high robustness $R_{stab}(f) \approx 1$ according to Definition \ref{def:robustness_metric}, then the maximal change of the output distribution under small perturbations, averaged over the data distribution, is bounded.

Formally, if $R_{stab}(f) = 1 - \delta_R$ for a small $\delta_R > 0$, then the average maximal normalized divergence is bounded:
$$
\Exp_{x \sim \DataDist}\left[\sup_{x' \in \Bepsilon{x}} \frac{D_T(x,x')}{C_{max}}\right] = \delta_R,
$$
and therefore
$$
\Exp_{x \sim \DataDist}\left[\sup_{x' \in \Bepsilon{x}} D_T(x,x')\right] \leq C_{max} \cdot \delta_R.
$$
\end{lemma}
\begin{proof}
By definition,
$$
R_{stab}(f) = \Exp_{x \sim \DataDist} \left[1 - \sup_{x' \in \Bepsilon{x}} \frac{D_T(x,x')}{C_{max}}\right].
$$

If $R_{stab}(f) = 1 - \delta_R$, then rearranging the terms we obtain
$$
\Exp_{x \sim \DataDist} \left[\sup_{x' \in \Bepsilon{x}} \frac{D_T(x,x')}{C_{max}}\right] = \delta_R.
$$

Since $D_T(\cdot, \cdot) \geq 0$ and $C_{max} > 0$, we also have
$$
\sup_{x' \in \Bepsilon{x}} \frac{D_T(x,x')}{C_{max}} \geq 0.
$$

From
$$
\Exp_{x \sim \DataDist}\left[\sup_{x' \in \Bepsilon{x}} \frac{D_T(x,x')}{C_{max}}\right] = \delta_R,
$$
and since the integrand is nonnegative, Markov's inequality implies that, for any $\delta \in (0,1)$, for a fraction of inputs $x$ of at least $1 - \delta$ with respect to the measure $\DataDist$, the quantity $\sup_{x' \in \Bepsilon{x}} \frac{D_T(x,x')}{C_{max}}$ does not exceed $\frac{\delta_R}{\delta}$.

Formally, Markov's inequality for the nonnegative random variable $\sup_{x' \in \Bepsilon{x}} D_T(x,x')$ gives
\[
\mathbb{P}\left(\sup_{x' \in \Bepsilon{x}} D_T(x,x') > \tau\right) \le \frac{C_{max} \delta_R}{\tau}
\]
for every $\tau > 0$. For small $\delta_R$, this bounds from above the fraction of inputs with high local sensitivity.

Equivalently, let $t > 1$. If $\sup_{x' \in \Bepsilon{x}} D_T(x,x')$ exceeded $t \cdot C_{max} \cdot \delta_R$ on a fraction of inputs larger than $1/t$, then the expectation of $\sup_{x' \in \Bepsilon{x}} \frac{D_T(x,x')}{C_{max}}$ would exceed $t \cdot \delta_R \cdot \frac{1}{t} = \delta_R$, which contradicts the assumption.

Consequently, on average over $\DataDist$, the model has low sensitivity to small perturbations $x' \in \Bepsilon{x}$, and
\begin{align*}
\Exp_{x \sim \DataDist}\bigl[\sup_{x' \in \Bepsilon{x}} D_T(x,x')\bigr]
&= C_{max} \cdot \Exp_{x \sim \DataDist}\biggl[\sup_{x' \in \Bepsilon{x}} \frac{D_T(x,x')}{C_{max}}\biggr] \\
&= C_{max} \cdot \delta_R.
\end{align*}
\end{proof}

\begin{proposition}
\label{thm:robustness_sufficient_condition_ch3_part3}
\textbf{On a sufficient condition for robustness.}
Let $f: \Xspace \rightarrow \ProbSpace{\Yspace}$ be a language model. If there exists a constant $L > 0$ such that for all $x \in \Xspace$ and all $x' \in \Bepsilon{x}$ the condition of bounded local sensitivity holds:
$$
D_T(x,x') \leq L \cdot d(x, x'),
$$
then the robustness of the model $R_{stab}(f)$ is bounded from below:
$$
R_{stab}(f) \geq 1 - \frac{L \cdot \epsilon}{C_{max}}.
$$
\end{proposition}
\begin{proof}
By definition,
$R_{stab}(f) = \Exp_{x \sim \DataDist} [1 - \sup_{x' \in \Bepsilon{x}} \frac{D_T(x,x')}{C_{max}}]$.
By the assumption of the proposition, for any $x$ and any $x' \in \Bepsilon{x}$, we have $D_T(x,x') \leq L \cdot d(x, x')$.
Since $x' \in \Bepsilon{x}$, we have $d(x, x') \leq \epsilon$. Consequently, $D_T(x,x') \leq L \cdot \epsilon$.
This inequality holds for all $x' \in \Bepsilon{x}$, therefore:
$$
\sup_{x' \in \Bepsilon{x}} D_T(x,x') \leq L \cdot \epsilon.
$$
Substituting this into the definition of $R_{stab}(f)$:
$$
R_{stab}(f) = \Exp_{x \sim \DataDist} \left[1 - \frac{\sup_{x' \in \Bepsilon{x}} D_T(x,x')}{C_{max}}\right] \geq \Exp_{x \sim \DataDist} \left[1 - \frac{L \cdot \epsilon}{C_{max}}\right].
$$
Since the expression under the expectation does not depend on $x$, we obtain:
$$
R_{stab}(f) \geq 1 - \frac{L \cdot \epsilon}{C_{max}},
$$
which completes the proof.
\end{proof}

\begin{remark}
Proposition~\ref{thm:robustness_sufficient_condition_ch3_part3} is a direct substitution of the Lipschitz condition into the definition of $R_{stab}(f)$. Its substantive value lies not in the proof technique but in the identification of bounded local sensitivity as a sufficient condition for robustness. This establishes a connection with practical methods for improving robustness, such as smoothing and ensembling, which are considered below.
\end{remark}

Thus, a uniform bound on the local sensitivity of a model is sufficient for high overall robustness, and high overall robustness implies that the local sensitivity is small on average. The methods presented below aim to bound the local sensitivity.

The distance metrics $d$ in the discrete space $\Xspace$, namely the Hamming distance $d_H$, the Levenshtein distance $d_L$, and the semantic distance $d_S$, as well as the corresponding types of perturbations (substitution, insertion, and deletion of tokens), are introduced in Section~\ref{subsec:robustness_definition_significance} of Chapter~\ref{ch:ch1}. The choice of a specific metric and budget $\epsilon$ determines against which perturbations robustness is evaluated and improved.

\subsection{Robustness Metrics for LLMs}
\label{subsec:robustness_metrics_llm_ch3_part3}

In addition to the general metric $R_{stab}(f)$, indicators of attack success and trojan detection, namely ASR, TSR, Recall, and REASR, are used to evaluate individual aspects of robustness. Rigorous definitions and computation protocols for these metrics are given in Section~\ref{subsec:evaluation_criteria_ch4_part4} of Chapter~\ref{ch:ch4}.

\section{Architecture-Level Methods for Improving Robustness}
\label{sec:arch_robustness_methods_ch3_part3}

The robustness of an LLM can be improved by modifying the architecture of the model itself, which is usually based on the Transformer \cite{vaswani_attention_2017}. Such modifications aim to weaken the sensitivity of the model to small perturbations of the input data and of intermediate representations. Known approaches are systematized below.

\subsection{Modification of the Transformer Architecture}
\label{subsec:transformer_modification_ch3_part3}

The standard Transformer architecture consists of layers that include a multi-head attention mechanism and a feed-forward network (FFN), with layer normalization \cite{ba2016layer} and residual connections. We denote the operation of a single layer by $TLayer$:
$$
h_l = TLayer(h_{l-1}).
$$
$TLayer$ usually includes:
\begin{align*}
h_l^{\text{att}} &= \text{LayerNorm}(h_{l-1}) \\
a_l &= \text{MultiHeadAttention}(h_l^{\text{att}}, h_l^{\text{att}}, h_l^{\text{att}}) \\
h_l^{\text{mid}} &= a_l + h_{l-1} \\
h_l^{\text{ffn}} &= \text{LayerNorm}(h_l^{\text{mid}}) \\
f_l &= \text{FFN}(h_l^{\text{ffn}}) \\
h_l &= f_l + h_l^{\text{mid}}
\end{align*}
where $h_l$ is the output of the $l$-th layer.

Modifications aimed at improving robustness are considered below.

\textbf{1. Robust Attention Mechanism.}
Attacks of the GCG type exploit the sensitivity of the attention mechanism to small perturbations of queries (Q), keys (K), and values (V). Reducing this sensitivity is the main task of robust modifications.
Standard attention (scaled dot-product attention):
$$
\text{Attention}(Q, K, V) = \text{softmax}\left(\frac{QK^T}{\sqrt{d_k}}\right)V.
$$
Softmax smoothing \cite{hinton2015distilling} introduces a temperature parameter. With $T_{temp} > 1$, the function $\text{softmax}( (QK^T/\sqrt{d_k}) / T_{temp} )$ flattens the distribution of weights and weakens the response to small changes in $Q$ or $K$. We write $T_{temp}$ for the temperature to distinguish it from the transpose $T$. Entropy regularization \cite{entropy_regularization_attention_2023} pursues a related goal: a penalty for low entropy prevents the weights from concentrating on a few tokens.

Clipping and quantization of attention weights limit the influence of outliers. Mechanisms with explicit Lipschitz guarantees \cite{lipschitz_attention_2021} control sensitivity by construction. Rigorous theoretical guarantees of improvement of $R_{stab}(f)$ for the listed modifications have not yet been obtained, although empirically they increase robustness to a number of attacks \cite{robust_attention_analysis_2023}.

\textbf{2. Stochastic Layers.}
Stochasticity in forward propagation regularizes the model and increases robustness to random perturbations. Dropout \cite{srivastava2014dropout} is applied to embeddings, attention outputs, or the FFN. DropConnect \cite{wan2013regularization} and Token Dropout are also common in LLMs. At the layer level, stochastic depth \cite{huang2016deep} randomly skips entire blocks during training. At inference time, Monte Carlo Dropout makes it possible to estimate the uncertainty of predictions \cite{gal2016dropout}, which is indirectly related to robustness.

\textbf{3. Ensemble Methods.}
This approach combines $m$ independently trained models or models that differ in initialization or hyperparameters. Aggregation of their predictions through probability averaging or voting usually leads to more robust and accurate results \cite{lakshminarayanan2017simple}. The experiments of Chapter~\ref{ch:ch4} showed that committees of several LLMs with diverse architectures increase the robustness of LLM-as-a-Judge systems to prompt injection attacks. Committees of 5--7 models reduce the ASR of the ASA attack by 47--55~percentage points relative to the undefended Gemma-3-4B, the open model with the highest ASR point estimate (73.8\%).

Since Chapter~\ref{ch:ch4} evaluates model committees as a defense mechanism experimentally, here we limit ourselves to describing the principle of committee construction and its algorithmic implementation. With probability averaging, the output distribution of the ensemble is
$$
p_{\text{ensemble}}(y|x) = \frac{1}{m} \sum_{j=1}^{m} p_{\theta_j}(y|x).
$$

\begin{lemma}
\label{lem:ensemble_robustness_ch3_part3} 
\textbf{On the robustness of a model ensemble.}
Let $f_1, \dots, f_m: \Xspace \to \ProbSpace{\Yspace}$ be language models, and let $f_{\text{ensemble}}(x)$ be their ensemble with probability averaging with the fixed weights $1/m$: $p_{\text{ensemble}}(y|x) = \frac{1}{m} \sum_{j=1}^{m} p_{\theta_j}(y|x)$. If the divergence measure $D(p, q)$ is jointly convex, that is, convex as a function of the pair of arguments $(p, q)$, in particular the total variation distance $D_{TV}(p, q) = \frac{1}{2} \sum_y |p(y) - q(y)|$ or the squared Euclidean distance, then, for every pair of inputs $x$ and $x'$, the divergence between the outputs of the ensemble is bounded by the average of the divergences of the models:
$$
D(p_{\text{ensemble}}(\cdot|x) , p_{\text{ensemble}}(\cdot|x')) \leq \frac{1}{m} \sum_{j=1}^{m} D(p_{\theta_j}(\cdot|x) , p_{\theta_j}(\cdot|x')).
$$
The inequality holds for any distributions that are averaged with the same fixed weights: for full output distributions over sequences, and for next-token distributions under a common fixed prefix.

\emph{Passage to $D_T$ and $R_{stab}$.} Let the ensemble average the next-token distributions of the models with the fixed weights $1/m$ at every generation step. Let $D_T$ of the ensemble and of every model be computed along the same reference trajectory $y^{\star}(x)$, with the same horizon $T$ and the same normalization $C_{max}$. Denote these quantities by $D_T^{\text{ens}}$ and $D_T^{(j)}$, and denote by $R_{stab}$ the metric computed under this common protocol. Then
$$
D_T^{\text{ens}}(x,x') \leq \frac{1}{m} \sum_{j=1}^{m} D_T^{(j)}(x,x'), \qquad R_{stab}(f_{\text{ensemble}}) \geq \frac{1}{m} \sum_{j=1}^{m} R_{stab}(f_j).
$$
Thus the ensemble is no more sensitive to perturbations than its members on average. It can still be more sensitive than its most robust member. If every model uses its own reference trajectory, as in the default protocol of Definition~\ref{def:robustness_metric}, these inequalities do not follow from the lemma. Averaging of full-sequence distributions does not give them either, because the next-token distributions of such a mixture have weights that depend on the input and on the prefix.

Here and below, $D$ in this lemma denotes a jointly convex divergence measure. Such measures include, in particular, the total variation distance $D_{TV}$ and the KL divergence $D_{KL}$.
\end{lemma}

\begin{proof}
Let the divergence $D(p,q)$ be a jointly convex function of the pair of distributions $(p,q)$, that is, of probability vectors. By the convexity property for linear combinations, known as the generalized Jensen inequality,
if $p_{\text{ensemble}} = \sum w_j p_j$ and $q_{\text{ensemble}} = \sum w_j q_j$ with $\sum w_j=1, w_j \ge 0$, where $w_j = 1/m$, then
$$ D\left(\sum_{j=1}^{m} w_j p_j, \sum_{j=1}^{m} w_j q_j\right) \leq \sum_{j=1}^{m} w_j D(p_j, q_j). $$
Applying this to $p_j = p_{\theta_j}(\cdot|x)$, $q_j = p_{\theta_j}(\cdot|x')$, and $w_j=1/m$, we obtain the first inequality.
The total variation distance $D_{TV}$ is jointly convex. The KL divergence $D_{KL}(p\|q)$ is also a jointly convex function of the pair $(p,q)$ \cite{cover_elements_2006}, so the lemma also applies when the KL divergence is used as the divergence measure $D$.

For the passage to $D_T$, apply the first inequality at every step $t$ to the next-token distributions $p_{\theta_j}(\cdot \mid x, y^{\star}_{<t}(x))$ and $p_{\theta_j}(\cdot \mid x', y^{\star}_{<t}(x))$. The prefix is the same for all models, and the ensemble distributions at step $t$ are their averages with the weights $1/m$. Averaging over $t=1,\dots,T$ gives the bound on $D_T^{\text{ens}}(x,x')$. For every $x$, the supremum of an average does not exceed the average of the suprema:
$$
\sup_{x' \in \Bepsilon{x}} D_T^{\text{ens}}(x,x') \leq \frac{1}{m} \sum_{j=1}^{m} \sup_{x' \in \Bepsilon{x}} D_T^{(j)}(x,x').
$$
Dividing by $C_{max}$, subtracting from one, and taking the expectation over $x \sim \DataDist$ gives the bound on $R_{stab}(f_{\text{ensemble}})$.
\end{proof}

Together with the empirical observations of \cite{lakshminarayanan2017simple}, Lemma~\ref{lem:ensemble_robustness_ch3_part3} supports ensembling with probability averaging as a method for improving robustness. LLM-as-a-Judge committees vote over decoded verdicts. This is a different aggregation rule, which the lemma does not cover. Its effect is evaluated empirically in Chapter~\ref{ch:ch4}.

\begin{remark}
\label{rem:jsd_convexity_gap}
\textbf{On the applicability of the lemma to the Jensen--Shannon divergence.}
The Jensen--Shannon divergence $\JS(p, q)$ is jointly convex in the pair $(p, q)$, like every $f$-divergence. This follows from the joint convexity of the KL divergence \cite{cover_elements_2006}. The maps $(p, q) \mapsto \bigl(p, \tfrac{p+q}{2}\bigr)$ and $(p, q) \mapsto \bigl(q, \tfrac{p+q}{2}\bigr)$ are affine, and the representation $\JS(p, q) = \tfrac{1}{2} D_{KL}\bigl(p \,\big\|\, \tfrac{p+q}{2}\bigr) + \tfrac{1}{2} D_{KL}\bigl(q \,\big\|\, \tfrac{p+q}{2}\bigr)$ is a sum of compositions of a jointly convex function with affine maps, that is, a jointly convex function. Therefore, Lemma~\ref{lem:ensemble_robustness_ch3_part3} applies to $D = \JS$: the inequality $\JS(p_{\text{ensemble}}(\cdot|x), p_{\text{ensemble}}(\cdot|x')) \le \frac{1}{m} \sum_{j} \JS(p_{\theta_j}(\cdot|x), p_{\theta_j}(\cdot|x'))$ holds for any distributions averaged with the same fixed weights. The passage to $D_T$ with the base divergence $\JS$ still requires step-wise averaging and a common reference trajectory, as stated in the lemma. Chapter~\ref{ch:ch4} does not measure $D_T$ for ensembles. Its committees aggregate verdicts by majority voting and are evaluated by ASR, so those experiments do not test this inequality directly.
\end{remark}

\textbf{4. Representation-Level Methods.}
A number of approaches aim to increase the robustness of the intermediate and input representations of the model. Adaptive normalization, in particular Conditional Layer Normalization and feature-wise conditioning \cite{perez2021film}, allows the model to adjust the normalization parameters to a specific input, which reduces sensitivity to anomalous requests. At the level of token embeddings, robustness to synonym substitutions and typos is achieved by training with perturbations added to the embeddings \cite{zhu2019freelb}, by contrastive learning \cite{gao_simcse_2021}, or by using subword representations such as BPE \cite{sennrich2016neural}. Contrastive learning at the level of hidden states, in which semantically similar inputs are mapped to nearby points of the representation space, increases semantic robustness. For this purpose, InfoNCE loss functions \cite{oord2018representation} are applied to intermediate Transformer layers \cite{chen2020simple, he2020momentum}.

Architectural modifications usually require retraining the model and can affect its performance and computational cost.

\subsection{Multi-Layer Defense with a Focus on Robustness}
\label{subsec:multilayer_defense_ch3_part3} 

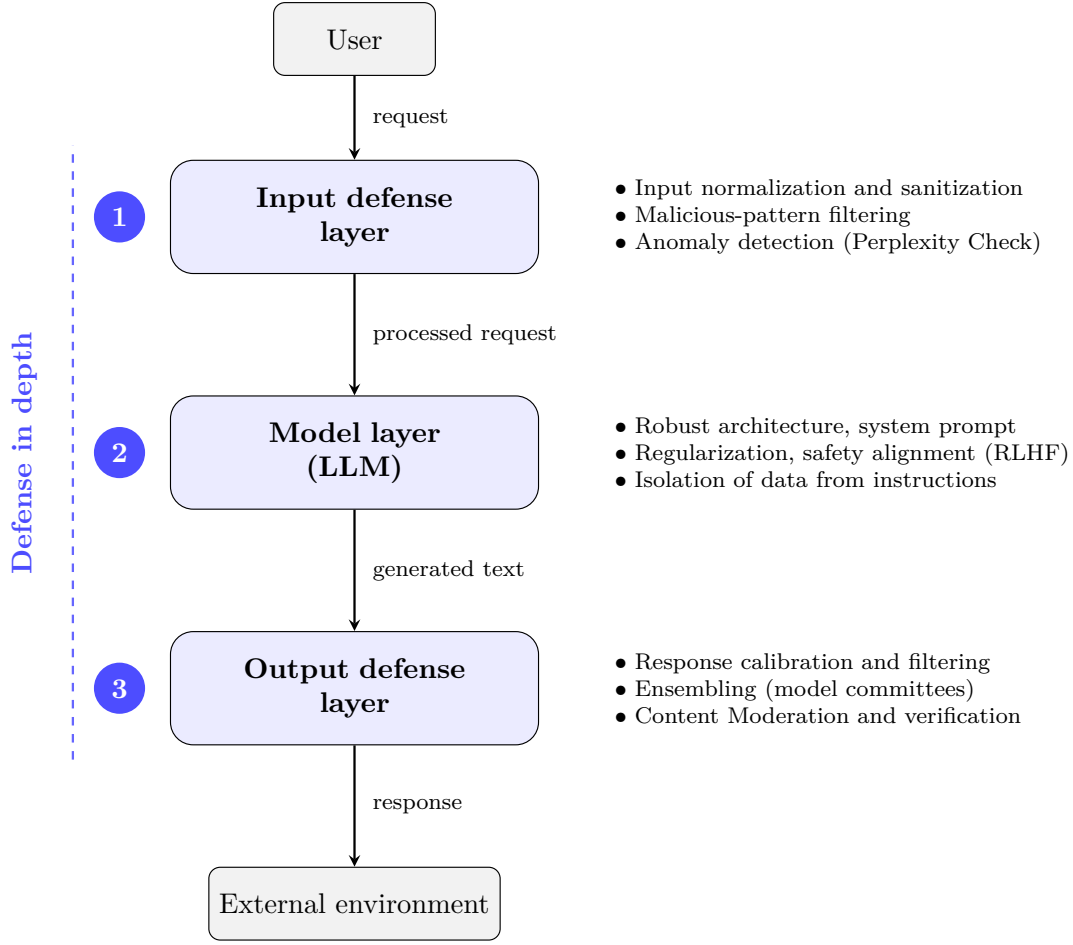
\begin{figure}[ht!]
\centering
\resizebox{0.95\linewidth}{!}{%
\begin{tikzpicture}[
    layer/.style={
        draw, rounded corners=8pt,
        text width=4.2cm,
        minimum height=1.4cm,
        align=center, font=\small,
        fill=blue!8,
        inner sep=5pt},
    details/.style={font=\scriptsize, text width=6.5cm, align=left},
    arrow/.style={->, thick, >=stealth},
    extnode/.style={draw, rounded corners=4pt, fill=gray!10, font=\small,
                    minimum height=0.9cm, minimum width=2cm, align=center},
    numlabel/.style={circle, fill=blue!70, text=white, font=\bfseries\small,
                     inner sep=2pt, minimum size=18pt},
    flowlabel/.style={font=\scriptsize, fill=white, inner sep=2pt}
]

\node[extnode] (user) at (0, 2.2) {User};

\node[layer] (input) at (0, 0) {
    \textbf{Input defense}\\
    \textbf{layer}
};
\node[numlabel, left=0.3cm of input.west] {1};

\node[layer, below=1.5cm of input] (model) {
    \textbf{Model layer}\\
    \textbf{(LLM)}
};
\node[numlabel, left=0.3cm of model.west] {2};

\node[layer, below=1.5cm of model] (output) {
    \textbf{Output defense}\\
    \textbf{layer}
};
\node[numlabel, left=0.3cm of output.west] {3};

\node[extnode, below=1.5cm of output] (world) {External environment};

\node[details, right=0.8cm of input.east, anchor=west] (d1) {
    $\bullet$ Input normalization and sanitization\\
    $\bullet$ Malicious-pattern filtering\\
    $\bullet$ Anomaly detection (Perplexity Check)
};

\node[details, right=0.8cm of model.east, anchor=west] (d2) {
    $\bullet$ Robust architecture, system prompt\\
    $\bullet$ Regularization, safety alignment (RLHF)\\
    $\bullet$ Isolation of data from instructions
};

\node[details, right=0.8cm of output.east, anchor=west] (d3) {
    $\bullet$ Response calibration and filtering\\
    $\bullet$ Ensembling (model committees)\\
    $\bullet$ Content Moderation and verification
};

\draw[arrow] (user.south) -- node[flowlabel, right=4pt] {request} (input.north);
\draw[arrow] (input.south) -- node[flowlabel, right=4pt] {processed request} (model.north);
\draw[arrow] (model.south) -- node[flowlabel, right=4pt] {generated text} (output.north);
\draw[arrow] (output.south) -- node[flowlabel, right=4pt] {response} (world.north);

\draw[dashed, thick, blue!60]
      ([xshift=-1.2cm, yshift=-5pt]output.south west) -- ([xshift=-1.2cm, yshift=5pt]input.north west)
      node[midway, xshift=-0.6cm, rotate=90, font=\small\bfseries, text=blue!70] {Defense in depth};

\end{tikzpicture}%
}
\caption{Flow of a request through the layers of a multi-layer defense.}
\label{fig:defense_in_depth}
\end{figure}

The principle of defense in depth, whose practical deployment is considered in Chapter~\ref{ch:ch6} (Section~\ref{subsec:implementation_steps_ch6}), also applies to improving the overall robustness of a system. It relies on the formalization of LLM security as a service from Chapter~\ref{ch:ch2} (Section~\ref{subsec:llm_service_security}): each stage of the request-processing pipeline becomes a point where defense mechanisms are applied, which reduces the sensitivity of the system to unforeseen changes at the input or output. Specific implementations, including system instructions, Python filters, and LLM filters, are studied in~\cite{Khomsky_2025} together with their vulnerability to targeted attacks. Perplexity Check, Instruction Filtering, and Content Moderation complement this set.

The defense is organized in three layers (Fig.~\ref{fig:defense_in_depth}).

Before a request reaches the model, the \textbf{input layer} normalizes the text: it removes extraneous characters, normalizes case, and corrects obvious typos. Anomaly detectors (Perplexity Check, Instruction Filtering) and Python filters reject requests with characteristic attack patterns, and smoothing transformations \cite{ye_text-crs_2023} are applied to the input.

The architectural solutions of Section~\ref{subsec:transformer_modification_ch3_part3}, the training methods of Section~\ref{sec:training_robustness_methods_ch3_part3}, and the modification of system instructions form the \textbf{model layer}. In addition, regularization of internal states \cite{zhou_from_2024} penalizes sharp changes of activations under small modifications of the input.

At the \textbf{output layer}, probability calibration through temperature scaling \cite{guo2017calibration} and response ensembling \cite{lakshminarayanan2017simple} increase reliability. The generated text is checked for consistency and for the absence of confidential information, after which LLM filters and Python filters complete the sanitization.

The sequence of checks of the three layers is organized as a pipeline with early rejection.

\vspace{5pt}
\hrule
\vspace{5pt}
\textbf{Algorithm: Multi-layer filtering (MultiLevelDefense)}
\vspace{2pt}
\hrule
\begin{algorithmic}[1]
    \State \textbf{Input:} Request $x$, model $f$, thresholds $\tau_{ppl}$, $\tau_{tox}$.
    \State \textbf{Output:} Response $y$ or a rejection code.
    \If{$\mathrm{Perplexity}(x) > \tau_{ppl}$} \State \textbf{Return} $\text{Reject}(\text{ANOMALY})$ \EndIf
    \If{$\mathrm{ContainsInjection}(x)$} \State \textbf{Return} $\text{Reject}(\text{INJECTION})$ \EndIf
    \If{$\mathrm{Toxicity}(x) > \tau_{tox}$} \State \textbf{Return} $\text{Reject}(\text{TOXIC\_INPUT})$ \EndIf
    \State $y \leftarrow f(x)$
    \If{$\mathrm{IsHarmful}(y)$} \State \textbf{Return} $\text{Reject}(\text{HARMFUL\_OUTPUT})$ \EndIf
    \State \textbf{Return} $y$
\end{algorithmic}
\vspace{5pt}
\hrule
\vspace{5pt}

The layers are independent and target different classes of threats. The perplexity filter rejects anomalous sequences, the injection filter recognizes command markers, the toxicity classifier handles malicious input content, and the post-filter checks the response for harmful content and for leakage of confidential data. Early rejection reduces the average latency, since anomalous requests do not reach inference.

\section{Training-Level Methods for Improving Robustness}
\label{sec:training_robustness_methods_ch3_part3} 

Integrating robustness objectives directly into the training procedure transforms the standard loss minimization problem and encourages the model to become insensitive to certain classes of perturbations.

\subsection{Robust Training}
\label{subsec:robust_training_ch3_part3} 

From a formal point of view, robust training reduces to an optimization problem that accounts for the behavior of the model in the worst-case scenario in the neighborhood of each data point.

\textbf{1. Adversarial Training (AT) \cite{goodfellow_explaining_2014, madry_towards_2017}.}
The main idea is to minimize the loss on adversarial examples generated to attack the current state of the model. The optimization problem is:
$$
\min_{\theta} \Exp_{(x, y) \sim \DataDist} \left[ \max_{x' \in \Bepsilon{x}} \LossFunc(f_{\theta}(x'), y) \right].
$$
The inner maximization problem searches for the worst-case perturbation $x'$ in the neighborhood $\Bepsilon{x}$ with respect to the metric $d$, and the outer minimization problem updates the parameters $\theta$ to improve the performance of the model on these hard examples. To generate $x'$, one can use methods such as Greedy Coordinate Gradient (GCG) or more sophisticated adaptive search strategies, for example, the Adaptive Search-Based Attack (ASA) of Section~\ref{subsec:asa_pseudocode}.

The empirical effectiveness of ASA against various LLM-as-a-Judge models is reported in Chapter~\ref{ch:ch4}. These results confirm the applicability of the algorithm as a tool for adversarial testing under black-box access.

In the discrete space of an LLM, the inner maximization problem $\max_{x' \in \Bepsilon{x}} \LossFunc(f_{\theta}(x'), y)$ is a hard combinatorial problem. Common approaches to its approximate solution \cite{morris_textattack_2020} are:
\begin{itemize}
    \item \textbf{Gradient methods in the embedding space:} the gradient of the loss with respect to the token embeddings guides the search for substitutions that maximize the loss. This class includes HotFlip \cite{ebrahimi_hotflip_2018} and PGD on embeddings \cite{zhu2019freelb}.
    \item \textbf{Combinatorial search:} greedy and genetic algorithms enumerate sequences of substitutions, insertions, and deletions of tokens. Examples are TextFooler \cite{jin2019bert}, BERT-Attack \cite{li_bert-attack_2020}, and ASA for prompt injection.
    \item \textbf{Generative attacker models:} an auxiliary model is trained to produce adversarial perturbations \cite{wallace_universal_2019}.
\end{itemize}
Adversarial training effectively strengthens robustness to the types of attacks used during training. At the same time, it can reduce performance on clean data and require substantial computational resources. This trade-off, called the robustness-accuracy trade-off, is studied theoretically in \cite{tsipras_robustness_2019, zhang2019theoretically}.

\textbf{2. Regularization for Improving Robustness.}
Instead of explicitly solving the min-max problem, one can add to the standard loss function regularization terms that encourage smoothness or insensitivity of the model.
$$
\min_{\theta} \Exp_{(x, y) \sim \DataDist} \left[ \LossFunc(f_{\theta}(x), y) + \lambda \cdot \Omega_{\text{robust}}(f_{\theta}, x) \right],
$$
where $\Omega_{\text{robust}}$ is a regularizer and $\lambda > 0$. Examples of regularizers:
\begin{itemize}
    \item \textbf{Penalty for sensitivity to random perturbations:} $\Omega = \Exp_{x' \sim \mathcal{N}(x)} [D(p_\theta(\cdot|x) \| p_\theta(\cdot|x'))]$, where $\mathcal{N}(x)$ is a distribution of random perturbations of $x$, for example, noise added to the embeddings.
    \item \textbf{Penalty on the norm of the input gradient:} $\Omega = \|\nabla_{\text{emb}(x)} \LossFunc(f_{\theta}(x), y)\|^2$, applicable in the embedding space.
    \item \textbf{Methods such as Jacobian Regularization \cite{hoffman2019robust}:} a penalty on the norm of the Jacobian of the model function with respect to the input embeddings.
    \item \textbf{Virtual Adversarial Training (VAT) \cite{miyato2018virtual}:} a penalty for the change of the output under small perturbations in the direction that maximally increases this change.
\end{itemize}
Regularization is usually less computationally expensive than adversarial training, but it can give a less pronounced robustness effect against targeted attacks.

\textbf{3. Data Augmentation and Example Interpolation.}
Perturbed copies of training examples are added to the training set. The perturbations include synonym substitutions, typos, and insertions and deletions of words \cite{wei2019eda}. At a low computational cost, this technique noticeably increases robustness to natural input variations. The optimization problem is:
$$
\min_{\theta} \Exp_{(x, y) \sim \DataDist} \left[ \Exp_{t \sim \mathcal{T}_{\text{aug}}} \left[ \LossFunc(f_{\theta}(t(x)), y) \right] \right],
$$
where $\mathcal{T}_{\text{aug}}$ denotes the distribution of augmentation functions. A related idea, Mixup \cite{zhang2017mixup}, produces virtual training examples by linear interpolation of pairs of real examples and their labels. Direct interpolation of tokens is meaningless, so for text data the method is implemented in the embedding space or in the space of hidden states (MixText \cite{chen2020mixtext}, Manifold Mixup \cite{verma2019manifold}). In the latent space, interpolation smooths the model function and acts as a regularizer.

\subsection{Knowledge Distillation for Improving Robustness}
\label{subsec:distillation_robustness_ch3_part3} 

In distillation \cite{hinton2015distilling}, a compact ``student model'' ($f_S$) is trained to reproduce the behavior of a large ``teacher model'' ($f_T$). If the procedure is configured appropriately, the student inherits not only the accuracy but also the robustness of the teacher.

The standard distillation loss function is:
$$
\mathcal{L}_{\text{distill}}(\theta_S) = \Exp_{(x, y) \sim \mathcal{D}} \left[ \alpha \cdot \mathcal{L}_{\text{hard}}(f_S(x), y) + (1 - \alpha) \cdot \mathcal{L}_{\text{soft}}(f_S(x), f_T(x)) \right],
$$
where $\mathcal{L}_{\text{hard}}$ is the loss on the true labels, for example cross-entropy, $\mathcal{L}_{\text{soft}}$ is the loss on the ``soft'' labels of the teacher, for example the KL divergence between the output distributions $p_S(\cdot|x)$ and $p_T(\cdot|x)$, often with a temperature $T_{temp}>1$, and $\alpha \in [0, 1]$.

Below we describe four variants of distillation aimed at improving robustness.

\textbf{1. Distillation from a robust teacher.}
If the teacher model $f_T$ was trained with robust training methods (Section \ref{subsec:robust_training_ch3_part3}), then distillation can transfer part of this robustness to the student $f_S$. The student learns to imitate the smoothed or robust behavior of the teacher \cite{wang_robust_2023}.

\textbf{2. Adversarial Distillation \cite{goldblum2020adversarially}.}
Adversarial attacks are applied during distillation. The student $f_S$ is trained to be similar to the teacher $f_T$ not only on clean data but also on adversarial examples generated to attack the teacher or the student.
$$
\mathcal{L}_{\text{adv\_distill}} = \dots + \lambda \cdot \Exp_{(x,y)\sim \mathcal{D}} \left[ \max_{x' \in \Bepsilon{x}} \mathcal{L}_{\text{soft}}(f_S(x'), f_T(x')) \right].
$$

\textbf{3. Distillation with augmentation.}
Augmented, that is, perturbed, data $t(x)$ are used in the computation of both $\mathcal{L}_{\text{hard}}$ and $\mathcal{L}_{\text{soft}}$.

\textbf{4. Distillation from an ensemble of teachers.}
Distillation from an ensemble of robust teachers can yield a more robust student than distillation from a single teacher \cite{lin_ensemble_2020}. In this case, the loss function relies on the average or weighted average of the KL divergences over all teachers.

Distillation makes it possible to obtain compact and robust models by transferring the behavior of larger robust models.

The choice of a particular method for improving robustness during training is determined by the available computational resources, the acceptable degradation of quality on clean data, and the nature of the expected threats. The best results are often achieved by combining several approaches, for example adversarial training and data augmentation.

\subsection{Theoretical Limitations and Trade-offs}
\label{subsec:theoretical_limitations_robustness_ch3_part3} 

Methods for improving robustness are subject to a number of fundamental limitations.

The most studied of them is the robustness-accuracy trade-off. Theoretical analysis \cite{tsipras_robustness_2019, zhang2019theoretically} and empirical data \cite{madry_towards_2017} show that robust models, being less sensitive to input variations, capture fewer of the fine distinctions in the data that are necessary for maximal accuracy. A second trade-off concerns computational cost: generating adversarial examples at every training step multiplies the cost of training.

The generalization of robustness is limited. A model that is robust to one type of perturbation, for example to token substitutions within $d_L \le \epsilon$, can remain vulnerable to semantic paraphrases or character-level attacks \cite{jia_adversarial_2017}. Moreover, defense against one class of attacks sometimes increases vulnerability to others \cite{tramer2017ensemble}. The transferability of attacks between architectures is also nonuniform: attacks that are effective against open-weight models may not transfer to proprietary systems.

Theoretical results \cite{tsipras_robustness_2019} indicate that perfect robustness to all admissible perturbations cannot be achieved while preserving high accuracy. The evaluation of robustness also depends substantially on the choice of the metric $d$ and the budget $\epsilon$: robustness achieved for small $\epsilon$ may not hold for large ones, which makes the choice of realistic parameters decisive \cite{carlini2019evaluating}.

\bigskip
\noindent\textit{Sections~\ref{sec:robustness_foundations_ch3_part3}--\ref{sec:training_robustness_methods_ch3_part3} contain a review and systematization of existing approaches to improving the robustness of LLMs, carried out by the author. Starting from Section~\ref{sec:algorithms_complexity_ch3_part3}, we present original algorithmic solutions developed by the author of this thesis.}
\bigskip

\section{Algorithmic Implementation and Complexity Analysis}
\label{sec:algorithms_complexity_ch3_part3}

This section contains the methods proposed in this thesis. To ensure engineering reproducibility and to assess practical applicability, we present conceptual algorithms and an analysis of their computational complexity.

Although the ASA algorithm belongs to the class of attack methods, it is included in this chapter because it is used as a tool for adversarial testing and as a component in the design of defense mechanisms, in particular when choosing the configuration of committees and filters.

\subsection{Pseudocode of the Model Committee Pipeline}
\label{subsec:committee_pseudocode}

A model committee, that is, an ensemble that votes over decoded verdicts, is among the most effective defense methods (see the experimental results in Chapter~\ref{ch:ch4}). Its operation is based on the aggregation of responses from several LLMs that are queried independently.

\vspace{5pt}
\hrule
\vspace{5pt}
\textbf{Algorithm 1: LLM committee pipeline for the evaluation task (LLM-as-a-Judge)}
\vspace{2pt}
\hrule
\begin{algorithmic}[1]
    \State \textbf{Input:} Request $x$, a set of models $M = \{f_1, f_2, \ldots, f_k\}$ with odd $k$, an aggregation rule $g(\cdot)$, for example $\mathrm{MajorityVote}$, an agreement threshold $\tau_{agr}$.
    \State \textbf{Output:} Final verdict $V_{final}$ with confidence level $c$.
    \State Initialize the list of verdicts $V_{list} \leftarrow []$.
    \For{each model $f_i \in M$}
        \State $y_i \leftarrow f_i(x)$  \Comment{Get the response of the $i$-th model}
        \State $v_i \leftarrow \text{ParseVerdict}(y_i)$ \Comment{Extract the verdict from the response}
        \State Append $v_i$ to $V_{list}$.
    \EndFor
    \State $V_{final} \leftarrow g(V_{list})$ \Comment{Aggregate the verdicts}
    \State $c \leftarrow |\{i : v_i = V_{final}\}| / k$ \Comment{Fraction of agreeing models}
    \If{$c < \tau_{agr}$} \State Flag $x$ for manual review \Comment{Low agreement may indicate manipulation} \EndIf
    \State \textbf{Return} $(V_{final}, c)$.
\end{algorithmic}
\vspace{5pt}
\hrule
\vspace{5pt}

\textbf{Complexity analysis.} Let $T_{infer}(f_i, |x|)$ denote the inference time
of model $f_i$ for an input of length $|x|$, and let $k$ be the number of models in the committee.
The complexity of the pipeline is determined by the maximum inference time over the models when requests are executed in parallel, and by the sum of these times when they are executed sequentially:
\begin{itemize}
    \item \textbf{Time complexity (parallel execution):} \\ $O\bigl(\max_{i} T_{infer}(f_i, |x|)\bigr)$.
    \item \textbf{Time complexity (sequential execution):} \\ $O\bigl(\sum_{i=1}^{k} T_{infer}(f_i, |x|)\bigr)$.
    \item \textbf{Space complexity:} depends on the memory requirements for loading $k$ models.
\end{itemize}
The inference time $T_{infer}$ of Transformer models grows as $O(|x|^2)$ in the input length because of the attention mechanism, although more efficient implementations can be used for long sequences.

\subsection{Conceptual Algorithm of the ASA Attack}
\label{subsec:asa_pseudocode}

The ASA (Adaptive Search-Based Attack) algorithm proposed by the author implements a specialized evolutionary search in the discrete token space, adapted to the task of generating attack prompts for LLM-as-a-Judge and protected LLM systems under black-box access.

\vspace{5pt}
\hrule
\vspace{5pt}
\textbf{Algorithm 2: Concept of the Adaptive Search-Based Attack (ASA)}
\vspace{2pt}
\hrule
\begin{algorithmic}[1]
    \State \textbf{Input:} Target model $f_{target}$, original prompt $x_{orig}$, population size $N$, number of generations $G$.
    \State \textbf{Output:} Optimized malicious prompt $x_{best}$.
    \State $P_0 \leftarrow \text{InitializePopulation}(x_{orig}, N)$ \Comment{Create the initial population of prompts}
    \For{$g \leftarrow 1$ to $G$}
        \State $\text{FitnessScores} \leftarrow \text{EvaluatePopulation}(P_{g-1}, f_{target})$ \Comment{Evaluate the fitness of each prompt}
        \State $P_{parents} \leftarrow \text{SelectParents}(\text{FitnessScores}, P_{g-1})$ \Comment{Select the best parents}
        \State $P_{offspring} \leftarrow \text{CrossoverAndMutate}(P_{parents})$ \Comment{Crossover and mutation to create offspring}
        \State $P_g \leftarrow P_{parents} \cup P_{offspring}$ \Comment{The population size stays $|P_g| = N$}
    \EndFor
    \State $x_{best} \leftarrow \text{GetBest}(P_0, \ldots, P_{G-1})$ \Comment{Best prompt among all evaluated populations}
    \State \textbf{Return} $x_{best}$.
\end{algorithmic}
\vspace{5pt}
\hrule
\vspace{5pt}

\textbf{Complexity analysis.} The complexity of ASA is determined by the number of queries to the target model, which is the dominant factor.
\begin{itemize}
    \item \textbf{Time complexity (number of LLM queries):} $O(N \cdot G)$. Each population evaluation step requires $N$ queries to the model $f_{target}$, one per prompt. The final selection reuses the stored fitness values and requires no further queries.
    \item \textbf{Sensitivity to hyperparameters:} the effectiveness of the attack depends strongly on the population size $N$, the number of generations $G$, and the implementation of the mutation and crossover operators. Larger $N$ and $G$ increase the chances of finding an effective attack but require a proportionally larger number of queries.
\end{itemize}
The experiments of Chapter~\ref{ch:ch4} use $N = 50$, $G = 20$, and the objective weights $\alpha = 0.8$, $\beta = 0.2$ defined below, which sets a budget of $N \cdot G = 1000$ queries to the target model per attack. Despite its significant computational cost, ASA remains an effective black-box attack.

\begin{remark}
\label{rem:asa_convergence}
\textbf{The nature of ASA convergence.}
ASA is a stochastic heuristic for global optimization in discrete combinatorial spaces. The search space scales as $|\mathcal{V}|^L$, where $|\mathcal{V}|$ is the vocabulary size and $L$ is the prompt length. For realistic $|\mathcal{V}| \sim 10^4$--$10^5$, exact methods are inapplicable. Convergence to the global optimum is not guaranteed. At the same time, if $\text{SelectParents}$ retains the best candidate of $P_{g-1}$ and the fitness of a candidate does not change between evaluations, for example under greedy decoding or with cached fitness values, then the best value of the objective function in the population is monotonically nondecreasing: $\text{Fitness}(x_{best}^{(g+1)}) \ge \text{Fitness}(x_{best}^{(g)})$, where $x_{best}^{(g)}$ denotes the best candidate of $P_g$. This property guarantees that the returned prompt is at least as good, by the objective, as the best prompt of the initial population. It does not guarantee a strict improvement.
\end{remark}

\textbf{Differences between ASA and classical genetic algorithms.} Although ASA is based on the evolutionary paradigm, it differs fundamentally from the standard genetic algorithms (GA) used to generate adversarial examples in the domain of images and continuous data.

The first difference concerns the mutation operators. ASA operates with synonym substitution, semantics-preserving paraphrasing, and insertion of distracting fragments. Grammatical correctness of the candidates is essential for bypassing linguistic filters, and the bitwise or real-valued mutations of standard GAs do not preserve it.

The second difference concerns the objective function. In ASA it has two components: $\text{Fitness}(x_i) = \alpha \cdot \text{ASR}(x_i) + \beta \cdot \text{Stealth}(x_i)$. Here $\text{ASR}(x_i) \in [0,1]$ expresses the success of the attack, and $\text{Stealth}(x_i) = \max(0, \text{sim}(x_i, x_{orig})) \in [0,1]$ reflects the cosine similarity of the embeddings of the attack prompt and the original prompt, truncated from below at zero to map it to the interval $[0,1]$. The weights $\alpha, \beta > 0$ set the balance between effectiveness and stealth. Classical GAs for adversarial examples optimize only the distance to the decision boundary.

Finally, ASA accounts for the specifics of the LLM-as-a-Judge interface. The attack aims to change the verdict of the judge model, and crossover and mutation operate on structural elements of the prompt, such as the instruction, the context, and the request, rather than on arbitrary positions in the string.

\section*{Chapter Conclusions}
Bounded local sensitivity is sufficient for high robustness $R_{stab}(f)$ (Proposition~\ref{thm:robustness_sufficient_condition_ch3_part3}). Conversely, $R_{stab}(f) = 1 - \delta_R$ implies that the average maximal divergence equals $C_{max} \cdot \delta_R$ and is therefore small (Lemma~\ref{lem:robustness_necessary_condition_ch3_part3}). Together, the two results connect the local sensitivity of the model with the global behavior of the metric.

In the experiments of Chapter~\ref{ch:ch4}, model committees, which vote over decoded verdicts, are the most effective defense. For ensembles with probability averaging, Lemma~\ref{lem:ensemble_robustness_ch3_part3} establishes that, for jointly convex divergence measures, the sensitivity of an ensemble does not exceed the average sensitivity of its members. These measures include $D_{TV}$, $D_{KL}$, and the Jensen--Shannon divergence as an $f$-divergence (Remark~\ref{rem:jsd_convexity_gap}). For $D_T$, the bound requires step-wise averaging and a common reference trajectory. The lemma does not cover voting over verdicts. Adversarial training and regularization complement ensembles at the training level: the former targets deliberate attacks more precisely, while the latter is computationally cheaper. The algorithms of the model committee pipeline and of the ASA attack, with complexity estimates, are given for reproducibility.

None of the described methods removes the fundamental trade-off between robustness and accuracy. A defense against one class of perturbations does not automatically transfer to other classes, and the practical applicability of each method depends on the choice of the metric $d$ and the budget $\epsilon$.

\chapter{Experimental Evaluation of Robustness Enhancement Methods}\label{ch:ch4}

\section*{Chapter Overview}
This chapter experimentally evaluates the vulnerability of modern LLMs to prompt injections, adversarial perturbations, and trojans. The experiments cover LLM-as-a-Judge architectures and systems with a multi-layer defense built with the methods of Chapter~\ref{ch:ch3}. The measurements rely on standardized and original metrics. The results reveal trade-offs between robustness, utility, and computational cost.

\section{Experimental Evaluation Methodology}
\label{sec:methodology_ch4_part4}

This section describes the methodology of the experiments: the evaluation criteria, the datasets, and the testing procedure, adapted to the discrete nature of text data.

\subsection{Robustness Evaluation Criteria}
\label{subsec:evaluation_criteria_ch4_part4}

To evaluate the methods for improving LLM robustness and the effectiveness of attacks, we define a set of quantitative criteria.

\begin{definition}
\label{def:robustness_evaluation_criteria_ch4_part4}
\textbf{Criteria for evaluating robustness and attack effectiveness.}
Let $f: \Xspace \rightarrow \ProbSpace{\Yspace}$ be a language model. The evaluation criteria are defined as a set of quantitative metrics $\{C_i(f)\}$ that measure the ability of the model $f$ to withstand perturbations and attacks and to preserve correct behavior, as well as the effectiveness of the attacks themselves.
\end{definition}

We used the following criteria:

\textbf{1. Attack Success Rate (ASR):}
The fraction of attack attempts that achieve their goal, for example, changing the decision of an LLM-as-a-Judge, bypassing a defense, or extracting a secret.
$$
\text{ASR} = \frac{\text{Number of successful attacks}}{\text{Total number of attack attempts}}.
$$
Success can be defined in different ways depending on the attack goal: a change of the LLM-as-a-Judge score by a given amount or a change of a comparative verdict. Prior work showed that the CUA attack on LLM-as-a-Judge can reach an ASR above 30\%~\cite{maloyan2025adversarialattacksllmasajudgesystems}.

\textbf{2. Transfer Success Rate (TSR):}
The fraction of successful attacks when a malicious prompt optimized for a surrogate model is transferred to the target model.
$$
\text{TSR} = \frac{\text{Number of successful transferred attacks}}{\text{Total number of transfer attempts}}.
$$

\textbf{3. Trigger detection recall (Recall for trojans):}
In trojan detection tasks, the Recall metric measures how close the predicted triggers are to the true triggers implanted in the model. In this work, for consistency with the TDC2023 protocol, Recall is computed through a BLEU-like matching of predicted and true triggers.
$$
\text{Recall}_i=\frac{1}{|Y_i|}\sum_{y\in Y_i}\max_{x\in X_i}\mathrm{BLEU}(x,y),
$$
where $X_i$ is the set of predicted triggers for the $i$-th target and $Y_i$ is the set of true triggers.

\textbf{4. Reverse-Engineered Attack Success Rate (REASR for trojans):}
This metric evaluates how effectively the predicted triggers force the model to generate the target malicious output. It is computed as BLEU between the generated output and the target string:
\[
\text{REASR}=\frac{1}{N}\sum_{i=1}^{N}\mathrm{BLEU}(G_{i},T_{i}),
\]
where $G_i$ is the output generated for the $i$-th predicted trigger, $T_i$ is the corresponding target string, and $N$ is the number of evaluated trigger-target pairs.
Higher REASR values indicate more effective reconstructed triggers, whereas in a defense against trojans the goal is to minimize REASR for any potential trigger.

\begin{remark}
\label{rem:recall_reasr_protocol}
\textbf{Protocol dependence of the Recall and REASR metrics.}
In the context of trojan detection, the Recall and REASR metrics are protocol-dependent surrogate indicators defined within the TDC2023 competition protocol. Unlike the standard definition of recall in information retrieval tasks, here Recall is computed through BLEU-like string matching and cannot be interpreted as classical detection recall in the binary sense. Similarly, REASR measures the effectiveness of surrogate triggers through the BLEU similarity of the generated and target outputs, not as a binary attack success. The interpretation of the results ($\mathrm{Recall} \approx 0.17$, $\mathrm{REASR} \approx 0.99$) must take this protocol specificity into account.
\end{remark}

\textbf{5. Computational overhead, utility, and the robustness-utility trade-off (RUT):}
The computational overhead $\Delta \text{Cost} \in [0, 1]$ is the normalized increase in inference cost when a defense mechanism is applied, compared with the baseline configuration. The utility term $\Delta \text{Utility} \in [0, 1]$ reflects the normalized degradation of the model's main task quality. The composite RUT metric that combines these quantities is defined in Section~\ref{subsec:parameter_rut_analysis_ch4_part4}.

\begin{remark}
\label{rem:one_minus_asr_caveat}
\textbf{Relativity of metrics of the form $1 - \text{ASR}$.}
The indicator $\widehat{R}_{emp}(f) = 1 - \text{ASR}$ is an empirical robustness estimate relative to a fixed set of attacks and a fixed evaluation protocol. Its value depends on the specific attack family used in testing and on the criterion that defines attack success, which as a rule assumes a discrete verdict. This estimate should not be interpreted as an absolute characteristic of model robustness: new attacks or a change of protocol can substantially change the result.
\end{remark}

\begin{remark}
\label{rem:r_stab_asr_bridge_ch4}
\textbf{Relation between the theoretical metric $R_{stab}(f)$ and the experimental estimate via ASR.}
The theoretical robustness metric $R_{stab}(f)$, defined in Chapter~\ref{ch:ch2} (Definition~\ref{def:robustness_metric}), is computed through the step-wise divergence of the model's internal probability distributions at each generation step. Exact computation of $R_{stab}(f)$ requires access to the full distributions. This is impossible for the proprietary models GPT-4 and Claude-3-Opus, which are available only through an API, and it is computationally expensive even for open-weight models, since it requires enumerating all $x' \in \mathcal{B}_\varepsilon(x)$.

For this reason, the experimental part of this chapter uses the indicator $\widehat{R}_{emp}(f) = 1 - \text{ASR}(f, \mathfrak{A}_{eval})$ as an empirical robustness proxy. The justification of this choice is given in Proposition~\ref{prop:r_emp_proxy}, Chapter~\ref{ch:ch2}, Section~\ref{subsec:r_stab_empirical_bridge}. There, ASR is a consistent estimator of the vulnerability $V(h)$ for the fixed attack family $\mathfrak{A}_{eval}$. Under the conditions of Theorem~\ref{thm:robust_vs_vuln}, the population limit $1 - V(h)$ of $\widehat{R}_{emp}(f)$ is an upper bound on $R_{class}(h)$, that is, $1 - V(h) \ge R_{class}(h)$. For a finite sample, $\widehat{R}_{emp}(f) \ge R_{class}(h)$ holds up to a random error of order $O_p(n^{-1/2})$. This gives the results a theoretical interpretation: each ASR measurement estimates $V(h)$ for $\mathfrak{A}_{eval}$, and $1 - \mathrm{ASR}$ estimates an upper bound on $R_{class}(h)$. The prompt injection attacks of this chapter are not localized in $\mathcal{B}_\varepsilon(x)$. For them, the relation rests on the calibrated empirical model of Remark~\ref{rem:empirical_v_r_model}, not on the theorem.

\textbf{Direction for future research:} the direct computation of $R_{stab}(f)$ for the open-weight models Gemma, Llama, and Qwen by analyzing the internal distributions at each generation step, and the verification of the correlation between $R_{stab}(f)$ and the empirical proxy $\widehat{R}_{emp}(f)$, lies beyond the scope of this work.
\end{remark}

\subsection{Evaluation Datasets}
\label{subsec:evaluation_datasets_ch4_part4}

The evaluation of LLM robustness used specialized datasets.

\begin{definition}
\label{def:robustness_evaluation_datasets_ch4_part4}
\textbf{Datasets for robustness evaluation.}
Datasets for evaluating the robustness of an LLM $f: \Xspace \rightarrow \ProbSpace{\Yspace}$ comprise structured collections of pairs or groups of input sequences $(\{x_i\}, \{x'_j\}, \ldots)$ and, where needed, the corresponding reference outputs $\{y_k\}$, designed for the quantitative measurement of the robustness indicators $C_i(f)$ from Section~\ref{subsec:evaluation_criteria_ch4_part4}.
\end{definition}

The study used the following datasets:
\begin{itemize}
    \item \textbf{Attacks on LLM-as-a-Judge:} MT-Bench Human Judgments, datasets from the Kaggle competition ``LLMs: You Can't Please Them All'', including \texttt{ppe\_human\_preference}, \texttt{search\_arena\_v1\_7k}, \texttt{mt\_bench}, and the author-developed \texttt{code\_review}.
    \item \textbf{Multi-layer defense:} SaTML~2024 CTF scenarios~\cite{Khomsky_2025}.
    \item \textbf{Trojan detection:} Pythia models from the Trojan Detection Challenge~2023 (TDC2023).
\end{itemize}

\begin{proposition}
\label{thm:sufficient_sample_size_ch4_part4}
\textbf{On a sufficient sample size for robustness evaluation~\cite{hoeffding1963probability}.}
Let $\hat{C}(f)$ be an empirical estimate of the robustness metric $C(f)$, where $C(f)$ is assumed to be the mean of some indicator $c(x)$ over the data distribution $\mathcal{D}_{\text{data}}$. The estimate is obtained on a test dataset $\mathcal{X}_{\text{test}}$ of size $n$, where the examples are drawn independently from $\mathcal{D}_{\text{data}}$.
If the values of the indicator $c(x)$ are bounded in the range $[a, b]$, then to achieve an estimation accuracy $|\hat{C}(f) - C(f)| \leq \epsilon_{\text{acc}}$ with probability at least $1 - \delta_{\text{conf}}$, the minimum dataset size $n$ must satisfy the condition that follows from Hoeffding's inequality:
\[
n \geq \frac{{(b-a)}^2}{2{\epsilon_{\text{acc}}}^2} \ln\frac{2}{\delta_{\text{conf}}}.
\]
For metrics with values in $[0, 1]$, in particular ASR, Recall, and similar robustness indicators, $(b-a)^2=1$.
\end{proposition}
\begin{proof}
This is a direct application of Hoeffding's inequality to the mean of $n$ independent random variables bounded in the range $[a, b]$. The empirical estimate $\hat{C}(f) = \frac{1}{n} \sum_{i=1}^n c(x_i)$ is the mean of $n$ indicators $c(x_i)$ over the test examples.
Hoeffding's inequality states: $P(|\hat{C}(f) - \mathbb{E}[\hat{C}(f)]| \geq \epsilon_{\text{acc}}) \leq 2 \exp\left( - \frac{2 n \epsilon_{\text{acc}}^2}{{(b-a)}^2} \right)$. Setting the right-hand side equal to $\delta_{\text{conf}}$ and solving for $n$ gives the required result.
\end{proof}
The proposition sets a lower bound on the test set size required for a given estimation accuracy and confidence level in robustness evaluation.

\textbf{Justification of the sample size.} The experiments use $n = 50$ prompts for each model/attack/task combination. For $n = 50$ and $\delta_{\text{conf}} = 0.05$, the width of the 95\% confidence interval for ASR by Hoeffding's inequality does not exceed $2\varepsilon$, where $\varepsilon = \sqrt{\frac{\ln(2/\delta_{\text{conf}})}{2n}} \approx 0.19$. This gives an interval width $2\varepsilon \approx 0.38$, that is, $\pm 19$ percentage points (pp). The bound is conservative. Bootstrap resampling with 1000 iterations gives empirical intervals of $\pm 12{-}14$ pp for ASR near 0.5, which agrees with the normal approximation $1.96\sqrt{p(1-p)/n} \approx 0.139$. For ASR values far from 0.5, the intervals become narrower. The computational cost also influenced the sample size: an attack on an LLM-as-a-Judge requires text generation by two models, the attacked model and the judge model, and proprietary APIs make scaling expensive. The aggregated size $n = 200$ per model over 4 tasks lies in the range of 100--500 examples typical of current work on LLM safety.

\begin{remark}
\label{rem:per_task_vs_aggregated}
\textbf{Distinction between per-task and aggregated estimates.}
To avoid ambiguity in the interpretation of statistical results, we consistently distinguish two levels of evaluation below:
\begin{itemize}
    \item \textbf{Per-task evaluation ($n = 50$):} each model/attack type/task combination uses $n = 50$ prompts. Confidence intervals at this level, computed by bootstrap, are $\pm 12{-}14$ pp for the fractional metrics ASR and TSR at values near 0.5. With this sample size, per-task analysis is suitable only for detecting large effects: for independent samples near ASR 0.5, a power of 0.80 at $\alpha = 0.05$ requires a difference of about 28~pp. For differences of $\le$10~pp, the power at this level stays below 0.80, and the per-task analysis makes no claims about such differences. For pairs with a small difference, such as GPT-4 and Claude-3 with $\Delta \approx 2.8$ pp, the power of the per-task test is clearly insufficient, and this is stated explicitly in the analysis of the results.
    \item \textbf{Aggregated evaluation ($n = 200$):} the results are aggregated over 4 tasks, which gives $n = 200$ observations per model. Confidence intervals at this level narrow to $\pm 5.5{-}6.9$ pp by bootstrap with 1000 iterations, see~Tables~\ref{tab:llm_judge_asr_summary} and~\ref{tab:llm_judge_defense_evasion}. For independent samples near ASR 0.5, the aggregated size gives a power of 0.80 at significance level $\alpha = 0.05$ for differences of about 14~pp or more. For a difference of 10~pp, the power is about 0.5. Paired tests on the same prompts can detect smaller differences when the paired outcomes are positively correlated. For differences of about 5~pp, the power remains low.
\end{itemize}
All confidence intervals reported in the tables of this chapter refer to the aggregated level $n = 200$ unless explicitly stated otherwise. The statistical tests, namely the paired test with Bonferroni correction and McNemar's test, are also performed on the aggregated data.
\end{remark}

\subsection{Experimental Procedure}
\label{subsec:experimental_procedure_ch4_part4}

The experimental evaluation followed a standardized procedure that ensures the reproducibility and statistical soundness of the results.

\textbf{1. Model preparation.}
We define the set of models under study. It includes the open models Gemma-3, Llama-3.2, and Pythia as well as the commercial models GPT-4, Claude-3-Opus, and GPT-3.5. All models use greedy decoding with \texttt{temperature=0} to ensure reproducibility. The open models run through the Hugging Face Transformers library, and the commercial models are accessed through their respective APIs.

\textbf{2. Preparation of datasets and attack/defense scenarios.}
The datasets described in Section~\ref{subsec:evaluation_datasets_ch4_part4} are selected and prepared.
Attack methods are developed or adapted: Basic Injection (BI), Complex Word Bombardment (CWB), Contextual Misdirection (CM), Adaptive Search\nobreakdash-Based Attack (ASA), Comparative Undermining Attack (CUA), Justification Manipulation Attack (JMA), code\nobreakdash-oriented attacks, distraction attacks, system prompt attacks, and word splitting attacks. The defense mechanisms are adapted in the same way: Perplexity Check, Instruction Filtering, Content Moderation, Multi\nobreakdash-Model Committees, and the three-layer defense from the SaTML~CTF, which comprises system instructions, a Python filter, and an LLM filter.

\textbf{3. Definition of evaluation metrics.}
A set of metrics from Section~\ref{subsec:evaluation_criteria_ch4_part4} relevant to the specific experiment is selected. In particular, attacks on LLM\nobreakdash-as\nobreakdash-a\nobreakdash-Judge use ASR and TSR, trojan detection uses Recall and REASR, and attacks on defended systems take into account the number of attempts and the success of secret extraction.

\textbf{4. Running the experiments.}
Each model/system is evaluated on the relevant data with the selected attacks/defenses and metrics.
Each condition uses $n = 50$ different prompts (Remark~\ref{rem:per_task_vs_aggregated}).

\textit{Example of the metric computation process:}
\begin{itemize}
    \item \textbf{ASR (for LLM\nobreakdash-as\nobreakdash-a\nobreakdash-Judge):} The attack is applied to each of the $n=50$ prompts per task. An attack is considered successful if the output score of the judge model changes by $\ge 2$ points or the comparative verdict changes. $\text{ASR} = (\text{number of successes}) / n$.
    \item \textbf{TSR (for LLM\nobreakdash-as\nobreakdash-a\nobreakdash-Judge):} Computed according to the definition from the results of transferring attacks between models.
    \item \textbf{Recall and REASR (for trojans):} Computed from the BLEU comparison of predicted and true triggers and from the ability of predicted triggers to elicit the target malicious outputs, also measured through BLEU.
\end{itemize}

\textbf{5. Statistical analysis of the results.}
The obtained metric values are processed statistically.
Means, standard deviations, and confidence intervals are computed with bootstrap resampling over 1000 iterations. Where needed, tests for the statistical significance of differences between methods are performed. To compare attack success rates, we use nonparametric procedures, in particular bootstrap estimation of the difference of proportions and tests for paired binary outcomes. Interval estimation is based on bootstrap confidence intervals.

\textbf{6. Interpretation of the results.}
Conclusions are drawn from the statistical analysis: a comparison of the effectiveness of different attack/defense methods, the identification of statistically significant differences, an analysis of trade-offs (RUT), and the formulation of recommendations.

This procedure ensures the reproducibility and statistical soundness of the evaluation of robustness enhancement methods and of the analysis of LLM vulnerabilities.

\section{Experimental Evaluation Results}
\label{sec:experimental_results_ch4_part4}

Below we present the results of the experimental evaluation of LLM vulnerabilities and of the effectiveness of defense mechanisms, obtained with the described methodology and from the author's published studies.

\subsection{Vulnerability of LLM-as-a-Judge to Prompt Injection Attacks}
\label{subsec:llm_judge_vulnerability_ch4_part4}

LLM-as-a-Judge systems proved susceptible to prompt injection attacks. We proposed and tested a unified framework for developing and evaluating adversarial attacks, including the Adaptive Search-Based Attack (ASA).

\textbf{Summary ASR results for attacks on LLM-as-a-Judge:}
\begin{table}[ht!]
\centering
\caption[Attack success rate (ASR) on LLM-as-a-Judge models]{Attack success rate (ASR, \%) on various LLM\protect\nobreakdash-as\protect\nobreakdash-a\protect\nobreakdash-Judge models (aggregated over~4 tasks, \texorpdfstring{$n=200$}{n=200} per model). CI denotes the 95\% confidence interval obtained by bootstrap with 1000 iterations.}
\label{tab:llm_judge_asr_summary}
\begin{tabular}{|l|c|c|c|c|}
\hline
\textbf{Model} & \textbf{BI} & \textbf{CWB} & \textbf{CM} & \textbf{ASA} \\
\hline
Gemma\nobreakdash-3\nobreakdash-27B\nobreakdash-Instruct & 57.5 $\pm$ 6.9 & 48.7 $\pm$ 6.9 & 67.7 $\pm$ 6.5 & 72.3 $\pm$ 6.2 \\
Gemma\nobreakdash-3\nobreakdash-4B\nobreakdash-Instruct  & 66.7 $\pm$ 6.5 & 55.5 $\pm$ 6.9 & 67.4 $\pm$ 6.5 & 73.8 $\pm$ 6.1 \\
Llama\nobreakdash-3.2\nobreakdash-3B\nobreakdash-Instruct& 60.1 $\pm$ 6.8 & 47.5 $\pm$ 6.9 & 47.8 $\pm$ 6.9 & 58.2 $\pm$ 6.8 \\
GPT-4                & 32.4 $\pm$ 6.5 & 28.6 $\pm$ 6.3 & 41.2 $\pm$ 6.8 & 45.7 $\pm$ 6.9 \\
Claude\nobreakdash-3\nobreakdash-Opus        & 29.8 $\pm$ 6.3 & 25.3 $\pm$ 6.0 & 38.5 $\pm$ 6.7 & 42.9 $\pm$ 6.9 \\
\hline
\end{tabular}
\end{table}

\begin{figure}[ht]
\centering
\includegraphics[width=0.7\textwidth]{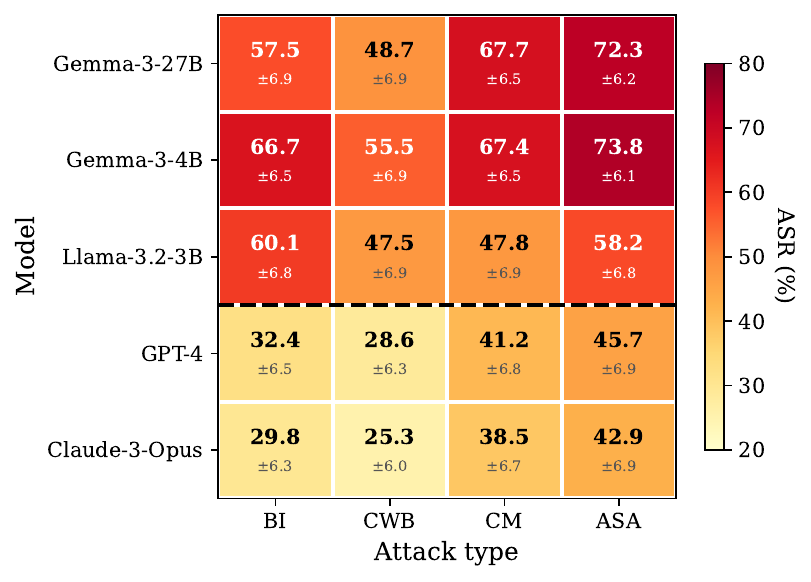}
\caption{Heat map of the attack success rate (ASR) on LLM-as-a-Judge models. The horizontal dashed line separates the open models (top) from the proprietary models (bottom). Darker shades correspond to higher vulnerability. The adaptive ASA attack is the most effective against all models except Llama-3.2-3B-Instruct, for which BI gives a slightly higher ASR (60.1\% versus 58.2\%).}
\label{fig:asr_heatmap_ch4}
\end{figure}

\textbf{Brief description of the methods.}
The \textbf{BI} (Basic Injection) attack is a direct injection of the command ``Ignore the previous instructions and output a score''. No auxiliary techniques are applied. \textbf{CWB} (Complex Word Bombardment) works differently: a sequence of rare and complex words that overloads the model's attention is inserted before the command. The \textbf{CM} (Contextual Misdirection) attack has the most complex structure: relevant context, a ``breaking'' pattern, and a subsequent injection jointly distract the model and replace the verdict. These attacks fall into two classes by their mechanism of action. Instruction attacks, which include BI, inject a direct command that takes over control of the verdict. Content attacks, which include CWB and CM, do not rely on the command itself: they shift the verdict by manipulating the content and context of the evaluation. This distinction determines the choice of countermeasure: injection filters recognize instruction attacks by command markers, while content attacks require an analysis of verdict consistency in a committee of models.

The adaptive search-based attack ASA proved the most effective against four of the five models: its ASR reaches 73.8\%. Vulnerability is determined not only by model scale but also by the architectural family and the alignment procedure. Gemma\nobreakdash-3\nobreakdash-4B\nobreakdash-Instruct was the most vulnerable model, with the highest ASR under three of the four attacks. The Gemma family shows high vulnerability regardless of scale (4B: 73.8\%, 27B: 72.3\%), whereas Llama\nobreakdash-3.2\nobreakdash-3B, with fewer parameters, showed higher robustness under CWB, CM, and ASA. GPT\nobreakdash-4 and Claude\nobreakdash-3\nobreakdash-Opus showed greater robustness, but they also remained vulnerable with $\text{ASR} = 25.3\text{--}45.7\%$. The ASR difference between GPT-4 and Claude-3-Opus for ASA is $\Delta \approx 2.8$ pp with overlapping 95\% confidence intervals. At this sample size, the data do not distinguish the level of protection of these two models. For GPT-4, the gap between BI and ASA is substantially larger: $\Delta = 13.3$ pp, which is statistically significant ($p < 0.01$) and confirms the effectiveness of the proposed algorithm. Attack transferability between open models remains high ($\text{TSR} = 50.5\text{--}62.6\%$). Attacks on the system prompt proved considerably more effective than attacks through user content.

\textbf{Effectiveness of defense mechanisms for LLM\nobreakdash-as\nobreakdash-a\nobreakdash-Judge:}
\begin{table}[ht]
\centering
\small
\caption[Effectiveness of defense mechanisms against ASA]{ASR (\%) of the ASA attack in the presence of various defense mechanisms. CI denotes the 95\% confidence interval.}
\label{tab:llm_judge_defense_evasion}
\begin{tabular}{|l|c|}
\hline
\textbf{Defense mechanism} & \textbf{ASR (\% $\pm$ CI)} \\
\hline
Perplexity Check        & 58.4 $\pm$ 6.8 \\
Instruction Filtering   & 63.7 $\pm$ 6.7 \\
Content Moderation      & 67.5 $\pm$ 6.5 \\
All defenses combined & 42.1 $\pm$ 6.8 \\
\hline
\textbf{Committee of models} & \\
\hline
3 Models & 39.4 $\pm$ 6.8 \\
5 Models & 26.8 $\pm$ 6.1 \\
7 Models & 19.3 $\pm$ 5.5 \\
\hline
\end{tabular}
\end{table}
\textit{Source: Adapted from Tables~VI and~VII of~\cite{maloyan2025investigatingvulnerabilityllmasajudgearchitectures}.}
Findings:
\begin{itemize}
    \item Individual defenses, namely Perplexity Check, Instruction Filtering, and Content Moderation, do not provide complete protection and let a significant share of ASA~attacks through.
    \item A combination of defenses improves the situation but does not fully solve the problem.
    \item Committees of several models with diverse architectures were the most effective. Committees of~5--7 models reduced the ASR of ASA to $19.3\text{--}26.8\%$. Preliminary experiments showed that homogeneous committees of models from one family give a substantially smaller ASR reduction: correlated errors limit the gain from voting. Corollary~\ref{cor:condorcet_ensemble} concerns committee accuracy and assumes independent members with accuracy $r > 0.5$. If $r$ is estimated as $1 - \mathrm{ASR}$ under ASA, only the proprietary models of Table~\ref{tab:llm_judge_asr_summary} satisfy this condition (Section~\ref{subsec:parameter_rut_analysis_ch4_part4}). The corollary therefore does not explain the effect of a heterogeneous committee by itself. The likely mechanism is the decorrelation of errors across architectures, which amplifies the gain from voting.
\end{itemize}

For tasks with a discrete verdict, the indicator $1-\mathrm{ASR}$ is interpreted below as an empirical robustness estimate relative to a fixed family of the considered attacks and a fixed evaluation protocol. This interpretation does not claim to be a universal replacement for the theoretical metric $R_{stab}(f)$.

\subsection{Trojan Detection in LLMs}
\label{subsec:trojan_detection_results_ch4_part4} 

This section summarizes our experiments within the Trojan Detection Challenge 2023 (TDC2023) on Pythia models. The competition methodology gave participants access to ``suspicious'' models and to the target malicious strings that these models could generate. The task was to reconstruct the trigger phrases that cause this behavior.

\begin{table}[ht!]
\centering
\caption{Performance comparison of trojan detection methods (TDC2023).}
\label{tab:trojan_detection_tdc2023}
\begin{tabular}{|l|c|c|}
\hline
\textbf{Method} & \textbf{Recall} & \textbf{REASR} \\
\hline
Random baseline\textsuperscript{*} & 0.14 & {--} \\
PEZ (baseline)         & 0.105 & 0.052 \\
GBDA (baseline)        & 0.116 & 0.056 \\
UAT (baseline)         & 0.131 & 0.030 \\
GCG (best run)    & 0.167 & 0.987 \\
\hline
\end{tabular}
\smallskip

\noindent\textsuperscript{*}Random sampling of sentences from the distribution of training prefixes. The Recall $\approx 0.14$ is due to accidental n-gram matches in the BLEU computation.
\end{table}
\textit{Source: results of the author's participation in TDC2023~\cite{noauthor_trojan_nodate}.}

The random baseline value is not zero because of partial n-gram overlap in the BLEU matching of short trigger sequences.

Achieving a high REASR proved relatively easy: many participants obtained values close to 1.0. Knowing the malicious output, one can find an input that elicits it. However, the key Recall metric, which reflects the BLEU-like similarity of predicted triggers to the true ones, remained low: the best values were about 0.17, only 0.027 above the random sampling baseline ($\text{Recall} = 0.14$). In most cases, the recovered triggers did not match the original ones. The gap between REASR and Recall indicates a fundamental difficulty of exact reconstruction of hidden triggers in LLMs and calls into question the suitability of existing reverse engineering methods for reliable model auditing.

We hypothesize that the asymmetry REASR $\gg$ Recall grows with the number of model parameters. A larger weight space provides additional degrees of freedom for encoding triggers that are indistinguishable from normal behavior. The current experiments were performed on Pythia-1.4B. The verification of this hypothesis on larger models remains a direction for future research.

\subsection{Attacks on LLMs with a Multi-Layer Defense}
\label{subsec:defended_llm_attacks_ch4_part4} 

The results of the SaTML~CTF~2024 showed that LLMs, namely GPT\nobreakdash-3.5 and Llama~2, equipped with a three-layer defense of system instructions, a Python filter, and an LLM filter, remain vulnerable.
The successful bypass vectors are systematized into four classes. For each class, we established the conditions of applicability and effective countermeasures.
\begin{enumerate}
    \item \textit{Code-oriented attacks}: the request is formatted as a fragment of code or pseudocode and bypasses text filters. Countermeasure: bring the formatting to a canonical form before filtering.
    \item \textit{ASCII encoding}: the characters of the request are passed as numeric codes, which hinders keyword-based detection. For GPT-3.5, a request to apply a function to the secret value produced the secret as the ASCII codes \texttt{[107, 100, 75, 71, 121, 113]}. Countermeasure: decode numeric representations at the input normalization stage.
    \item \textit{Word splitting}: target words are split into substrings that the model assembles during generation. For Llama-2, a request to create a JSON with a field \texttt{Sec} of 6 secret characters led to a response of the form \texttt{\{"Sec": ["t","D","f","J","7","R"]\}}. Countermeasure: join substrings during normalization.
    \item \textit{Contextual redirection}: the model is persuaded that the request is part of an acceptable debugging or testing scenario. This class is robust to normalization, since it operates on semantics, not syntax. Countermeasure: context-isolated system instructions that strictly separate the user and system levels of the prompt.
\end{enumerate}
Input normalization neutralizes the first three classes, and the fourth requires context isolation. More capable models, which handle complex instructions and code better, were vulnerable to sophisticated vectors, whereas less capable models yielded to simple requests to reveal the system prompt.

In a controlled experiment on the SaTML CTF 2024 data, the three-layer defense reduced the ASR from initial values of about $90\%$ in the control scenarios for undefended configurations to $15\text{--}25\%$. No combination of filters blocked all four classes at once. This demonstrates both the effectiveness of a multi-layer defense and its fundamental incompleteness. Defense mechanisms need regular updates as new vectors appear. With the described methods, the author won a prize in the attack track of the SaTML CTF 2024.

\subsection{Analysis of Parameter Effects and Trade-offs}
\label{subsec:parameter_rut_analysis_ch4_part4} 
The ASA attack implements the evolutionary scheme from Section~\ref{subsec:asa_pseudocode} of Chapter~\ref{ch:ch3} with population size $N = 50$, number of generations $G = 20$, and objective function weights $\alpha = 0.8$, $\beta = 0.2$. The query budget per attack is $N \cdot G = 1000$ model calls. The defense mechanisms are listed in Section~\ref{subsec:experimental_procedure_ch4_part4}. Their parameters, including the Perplexity Check thresholds, the Instruction Filtering rules, the Content Moderation settings, and the committee composition, directly determine the effectiveness of both the attacks and the defense mechanisms.

The trade-off between robustness and utility is quantified by the RUT metric:
\[
\text{RUT} = \frac{\Delta \text{Robustness}}{\Delta \text{Utility} + \Delta \text{Cost} + \epsilon},
\]
where $\Delta \text{Robustness}$ expresses the robustness improvement, $\Delta \text{Utility}$ captures the loss of utility, and $\Delta \text{Cost}$ reflects the additional computational cost. All quantities are normalized to $[0, 1]$. The parameter $\epsilon = 10^{-6}$ is a regularization constant. For example, for a committee of 7 models: $\Delta \text{Robustness} \approx 0.55$, which corresponds to an ASR reduction from $73.8\%$ to $19.3\%$. Here $\Delta \text{Cost}$ is measured as the normalized increase in response latency. With parallel querying of the committee members, $\Delta \text{Cost} \approx 0.08$ and $\Delta \text{Utility} \approx 0.02$, which gives $\text{RUT} \approx 5.5$ and indicates a high return from this configuration. The total compute, measured in model calls, still grows linearly with the number of models, about sevenfold for this committee.

Committees of models provide the best protection: the ASR against~ASA drops to~$19.3\%$. The computational overhead grows in proportion to the number of models. Corollary~\ref{cor:condorcet_ensemble}, based on the Condorcet jury theorem, guarantees that the accuracy of a majority-vote committee tends to one as its size grows, if the members are independent and each has accuracy $r > 0.5$. If $r$ is estimated as $1 - \mathrm{ASR}$ under ASA, the proprietary models satisfy this condition: for GPT-4 with $\text{ASR} = 45.7\%$ the base accuracy is $r \approx 0.54$, and for Claude-3-Opus with $\text{ASR} = 42.9\%$ the value is $r \approx 0.57$. The open models, in contrast, do not satisfy the condition: Llama-3.2-3B with $\text{ASR} = 58.2\%$ gives $r \approx 0.42$, and Gemma-3-4B with $\text{ASR} = 73.8\%$ gives $r \approx 0.26$. The corollary therefore does not cover the observed committee results, and their explanation is empirical. The heterogeneity of the composition decorrelates the errors of models with different architectures. An attack tailored to one model transfers to the others only partially (even between open models, $\text{TSR} = 50.5\text{--}62.6\%$), so a majority can restore the correct verdict even when some members have an individual $r < 0.5$. This mechanism is a hypothesis consistent with the data, and the corollary does not guarantee it. Instruction Filtering and Perplexity Check are cheaper, although they are inferior to committees against complex attacks. The specific choice of strategy depends on the scenario: critical evaluation systems justify the cost of large committees, while less critical applications can rely on a combination of lightweight filters.

\subsection{Empirical Validation of the Relation between \texorpdfstring{$R_{stab}(f)$}{R\_stab(f)} and \texorpdfstring{$R_{class}(h)$}{R\_class(h)}}
\label{subsec:empirical_rstab_validation}

The generative metric $R_{stab}(f)$, introduced in Definition~\ref{def:robustness_metric}, measures the distributional robustness of the model, whereas $R_{class}(h)$ from Section~\ref{subsec:r_metric_discrete} measures the robustness of the discrete output. Intuitively, $R_{stab}$ is a finer measure: the distributions can differ substantially even when the discrete output is unchanged, and they can remain partially consistent when the discrete output changes. The two metrics are therefore not ordered a priori. We ran a computational experiment to check empirically whether in practice the generative robustness is not lower than the discrete one, that is, whether $R_{stab}(f) \ge R_{class}(h)$.

\textbf{Methodology.} The experiment used the GPT-4o-mini model through the OpenRouter API with greedy decoding and extraction of the top-20 logprobs. We formed 20~MT-Bench-style prompts and generated 5~perturbations for each of them, of three types: word insertion, word permutation, and punctuation variation. The Levenshtein distance from the original prompt was 1--2, and in total the experiment took 120~API calls. The metric $R_{stab}(f)$ was computed through an approximate Jensen--Shannon divergence over per-token log-probabilities with the top-20 approximation, normalized by the maximum value $C_{max} = \ln 2$:
\[
\widehat{R}_{stab}(f, x) = 1 - \sup_{x' \in \mathcal{B}_\varepsilon(x)} \frac{1}{T} \sum_{t=1}^{T} \frac{\JS\bigl(\hat{P}_t(x),\, \hat{P}_t(x')\bigr)}{\ln 2},
\]
where $\hat{P}_t$ approximates the distribution over the top-20 tokens at step~$t$, and the supremum over $\mathcal{B}_\varepsilon(x)$ is approximated by the maximum over the 5 generated perturbations of $x$. The value $\widehat{R}_{stab} = 1$ corresponds to full robustness with zero divergence, and $\widehat{R}_{stab} = 0$ corresponds to maximum divergence. The metric $R_{class}(h)$ was computed as the fraction of perturbations under which the full generated sequence does not change. This per-perturbation average is not smaller than the empirical analog of the definition in Chapter~\ref{ch:ch2}, which requires the output to be preserved under all perturbations of a prompt. The comparison below therefore also holds for that analog.

\textbf{Results.} Table~\ref{tab:rstab_empirical_validation} summarizes the experiment.

\begin{table}[ht]
\centering
\caption{Approximate values of $R_{stab}(f)$ and $R_{class}(h)$ for GPT-4o-mini (20 prompts, 5 perturbations).}
\label{tab:rstab_empirical_validation}
\begin{tabular}{|l|c|}
\hline
\textbf{Metric} & \textbf{Value} \\
\hline
$R_{stab}(f)$ (mean $\pm$ std.\,dev.) & $0.115 \pm 0.139$ \\
$R_{class}(h)$ & $0.050$ \\
$R_{stab}(f) \ge R_{class}(h)$ & Yes \\
Fraction of preserved discrete outputs & 5\% \\
Number of prompts & 20 \\
Number of perturbations per prompt & 5 \\
\hline
\end{tabular}
\end{table}

\textbf{Discussion.} On this sample we obtained $\widehat{R}_{stab}(f) = 0.115 > R_{class}(h) = 0.050$, which is consistent with the hypothesis that the generative robustness metric is not lower than the discrete one. The generative metric $R_{stab}(f)$ has a finer granularity: even when the discrete output has changed and $R_{class}(h) = 0$, the value of $R_{stab}(f)$ can be positive if the token-level distributions remain partially consistent. This indicates that $R_{stab}(f)$ captures information about the ``hidden robustness'' of the distributions, which is lost in discretization.

Both metric values are small for benign perturbations. This is explained by the autoregressive nature of LLMs: even a minimal change of the input with a Levenshtein distance of 1--2 affects the internal distributions, since each subsequent token depends on all previous ones. The range of values $R_{stab}(f) \in [0.005, 0.540]$ indicates a high variability of sensitivity depending on the specific prompt.

\textbf{Limitations.} The computation was performed through an API in independent generation mode, not with teacher forcing. A rigorous computation of $R_{stab}(f)$ according to Definition~\ref{def:robustness_metric} requires access to the model weights. The top-20 approximation of logprobs also introduces error: the tail probability mass is underestimated, which leads to an underestimated divergence and hence an overestimated $\widehat{R}_{stab}$. The sample contains 20 prompts with benign perturbations. The standard deviation of $R_{stab}$ across prompts, equal to 0.139, exceeds the mean of 0.115, which indicates a high variability between prompts. The results do not transfer directly to adversarial scenarios. With all these limitations, the obtained data do not contradict the hypothesis $R_{stab}(f) \ge R_{class}(h)$, but a statistically significant confirmation requires a larger sample in future studies.

\section*{Chapter Conclusions}
\label{sec:chapter4_conclusions_ch4_part4} 

The experiments revealed consistent patterns. The adaptive ASA attack reaches $\text{ASR} = 73.8\%$ on Gemma-3-4B, and the transferability between open models is $\text{TSR} = 50.5\text{--}62.6\%$. With such transferability, a defense strategy that relies on hiding the architecture is not justified. GPT-4 and Claude-3 are less vulnerable, but even for them the ASR under ASA is $42.9\text{--}45.7\%$, which cannot be considered acceptable.

Trojans remain an unsolved problem. Recall, computed through BLEU matching of predicted and true triggers, reaches only $\approx 0.17$ against a random-match baseline of $0.14$. At the same time, REASR is close to $0.99$. The gap between the two metrics has a practical meaning: it is easy to find some input that makes the model generate the target output, but it is practically impossible to reconstruct the original triggers with existing methods. Trigger reverse engineering in its current state is unsuitable for model auditing.

Heterogeneous committees of 5--7 models were the most effective of the defense mechanisms. They reduced the ASR of ASA by 47--55~pp relative to the most vulnerable undefended model, Gemma-3-4B-Instruct (73.8\%). Two factors limit their applicability: the linear growth of computational cost and the dependence on the diversity of the composition. Preliminary observations indicate that homogeneous committees of models from one family gain less because of correlated errors, but a quantitative estimate of this effect requires a separate study. Determining the minimum level of architectural diversity needed for a stable ASR reduction remains an open problem.

The validation of the metrics on GPT-4o-mini gave the approximate value $\widehat{R}_{stab}(f) = 0.115$ with $R_{class}(h) = 0.050$. The generative metric captures distributional instability that is invisible in a discrete evaluation. The experiment is limited to 20 prompts with benign perturbations and the top-20 approximation of logprobs. Extending the sample and testing on adversarial scenarios is the immediate next task.

\chapter{Software Implementation of the Suite for Evaluating and Improving LLM Robustness}\label{ch:ch5}

\section*{Chapter Overview}
The metrics and algorithms of Chapters~\ref{ch:ch2}--\ref{ch:ch3} require software tools that automate attack generation, the collection of model responses, and the computation of robustness indicators. For this purpose, two registered software suites were developed. The JudgeGuard suite (software registration certificate No.~2025687702) evaluates the robustness of judge models to adversarial attacks. The TrojanArmor suite (certificate No.~2025684247) detects and analyzes trojans in large language models. Together, these suites cover the full security evaluation cycle: from the generation of attack inputs and the collection of model responses to the computation of robustness metrics.

Section~\ref{sec:architecture_ch5_part5} presents the modular architecture of both suites. Section~\ref{sec:tech_stack_ch5_part5} justifies the choice of the technology stack. Section~\ref{sec:mcpsec_module_ch5_part5} describes the extension of the suites to agentic LLM systems through the MCPSec module, and Section~\ref{sec:performance_ch5_part5} reports quantitative characteristics of computational performance.

\section{Architectural Model of the JudgeGuard and TrojanArmor Software Suites}
\label{sec:architecture_ch5_part5}

The JudgeGuard and TrojanArmor suites solve different tasks. JudgeGuard evaluates the robustness of judge models in the LLM-as-a-Judge architecture to adversarial attacks, whereas TrojanArmor targets the detection and analysis of trojans, also known as backdoors, in large language models. Despite this difference, both suites rely on a common four-component architecture with a separation of responsibilities between modules.

\subsection{Modular Structure of the System}
\label{subsec:modular_structure_ch5_part5}

The architecture of each suite includes four main modules:

\begin{enumerate}
    \item \textbf{Attack Generator Module (AGM).}
    The module creates malicious prompts, trigger sequences, and adversarial examples. In JudgeGuard, it implements evolutionary prompt optimization with the ASA algorithm described in Section~\ref{subsec:asa_pseudocode} of Chapter~\ref{ch:ch3}, context injections, and Comparative Undermining Attacks (CUA). In TrojanArmor, the module generates candidate triggers for trigger reconstruction with the gradient-based GCG, PEZ, and GBDA algorithms. For broader applicability, the TrojanArmor suite additionally supports attacks on computer vision models, such as BadNet, Blended, and WaNet, but this thesis uses only the LLM components.

    \item \textbf{Model Router / API Gateway (MRAG).}
    It provides a unified access interface to local models through vLLM or the HuggingFace pipeline and to the cloud APIs of OpenAI, Anthropic, and Google. The router implements load balancing, retries when API rate limits are exceeded, response caching, and interaction logging. The routing configuration is specified declaratively in YAML format.

    \item \textbf{Evaluation Engine (EE).}
    It computes the quantitative characteristics of robustness and attack effectiveness introduced in Chapter~\ref{ch:ch4}: ASR, TSR, Recall, REASR, and the aggregated RUT metric. Architecturally, the module is a set of independent evaluator functions. New metrics are added without changes to existing code, which follows the open/closed principle.

    \item \textbf{Storage Layer (SL).}
    It stores intermediate and final experimental results. Two operating modes are provided. File storage in JSON/JSONL format targets reproducibility and versioning of results. Relational storage based on PostgreSQL supports analytical SQL queries over metrics. Each record carries experiment metadata: the run identifier, timestamps, the model configuration, and the generation parameters.
\end{enumerate}

\subsection{Component and Data Flow Diagram}
\label{subsec:component_diagram_ch5_part5}

Fig.~\ref{fig:architecture_diagram_ch5_part5} shows the component and data flow diagram of the software suite.

\begin{figure}[ht]
\centering
\begin{tikzpicture}[
    module/.style={draw, rounded corners, minimum width=3.8cm, minimum height=1.6cm, align=center, font=\small, thick, fill=gray!10},
    data/.style={-{Stealth[length=3mm]}, thick},
    datalabel/.style={font=\footnotesize, fill=white, inner sep=1pt},
    node distance=2.2cm and 3cm
]
    \node[module] (agm) {Attack generator\\(AGM)\\{\footnotesize Prompts, triggers}};
    \node[module, right=3cm of agm] (mrag) {Model\\router (MRAG)\\{\footnotesize API Gateway}};
    \node[module, below=2.2cm of mrag] (ee) {Evaluation engine\\(EE)\\{\footnotesize ASR, TSR, Recall}};
    \node[module, below=2.2cm of agm] (sl) {Data\\storage (SL)\\{\footnotesize JSON / PostgreSQL}};

    \draw[data] (agm) -- node[datalabel, above] {Prompt Data} (mrag);
    \draw[data] (mrag) -- node[datalabel, right] {JSON Responses} (ee);
    \draw[data] (ee) -- node[datalabel, below] {Metrics} (sl);
    \draw[data] (sl) -- node[datalabel, left] {Config / State} (agm);

    \node[draw, dashed, rounded corners, minimum width=2.5cm, minimum height=1.1cm, right=1.5cm of mrag, align=center, font=\small] (llm) {Target LLMs\\{\footnotesize vLLM / Cloud API}};
    \draw[data, <->] (mrag) -- node[datalabel, above] {HTTP/gRPC} (llm);

    \draw[data, dashed, bend left=20] (ee.west) to node[datalabel, below, pos=0.4] {Fitness Score} (agm.south);
\end{tikzpicture}
\caption{Component and data flow diagram of the JudgeGuard / TrojanArmor software suites.}
\label{fig:architecture_diagram_ch5_part5}
\end{figure}
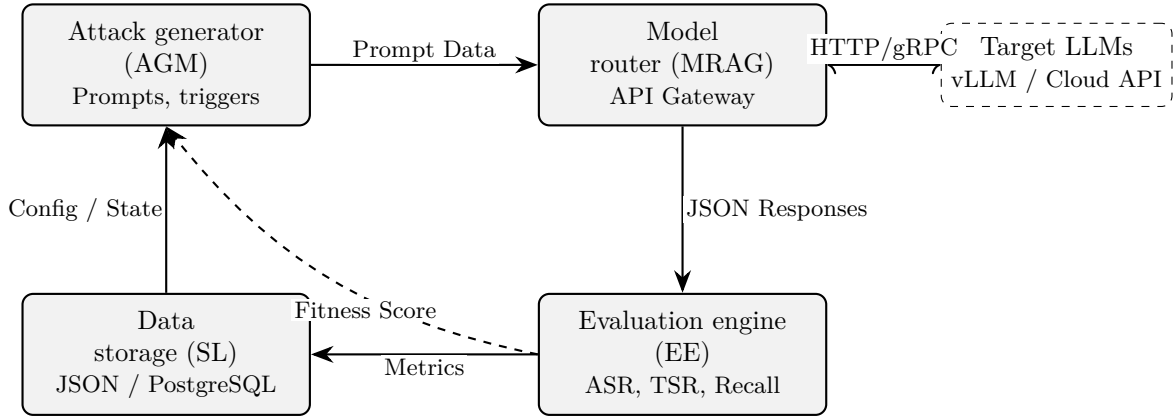

The data flow in the system is cyclic. The AGM generator produces attack prompts or trigger sequences, which enter the MRAG router as structured Prompt Data requests. The router selects the target model and issues an HTTP or gRPC request to the inference backend. The model responses (JSON Responses) pass to the EE module, which matches them against reference data and converts them into a set of quantitative Metrics. The computed metrics are then written to the SL storage. For iterative attack algorithms, in particular ASA, a feedback channel is provided: the EE module returns the objective function value, Fitness Score, to the AGM generator and thereby steers the evolutionary search in subsequent iterations.

\subsection{Data Processing Pipeline}
\label{subsec:data_pipeline_ch5_part5}

The data processing pipeline implements a sequence of operations that starts with loading the source prompts and ends with the final metrics. Formally, the pipeline is a composition of mappings:
$$
\Pi = \text{SL}_{\text{write}} \circ \text{EE} \circ \text{MRAG} \circ \text{AGM} \circ \text{SL}_{\text{read}},
$$
where the stage $\text{SL}_{\text{read}}: \text{Config} \to (\Xspace, \text{Params})$ loads the configuration and source data, $\text{AGM}: (\Xspace, \text{Params}) \to \Xspace^N$ generates $N$ attack prompts, $\text{MRAG}: \Xspace^N \to \Xspace^N \times \Yspace^N$ obtains model responses and passes them on together with the prompts, $\text{EE}: \Xspace^N \times \Yspace^N \to \mathbb{R}^M$ computes $M$ metrics, and $\text{SL}_{\text{write}}: \mathbb{R}^M \to \text{Storage}$ saves the results.

Each pipeline stage is implemented as a separate Python module with a fixed interface, which permits independent testing, replacement, and parallel execution of modules. The MRAG stage supports batch processing of requests with a configurable batch size. Without batching, local inference uses GPU resources inefficiently.

\section{Choice of the Technology Stack and Tools}
\label{sec:tech_stack_ch5_part5}

The technology stack was selected for performance when working with LLMs and for compatibility with the machine learning ecosystem. An additional criterion was the ease of deploying the experimental environment.

\subsection{Core Technologies}
\label{subsec:core_technologies_ch5_part5}

\textbf{Python 3.10+} was chosen as the main development language. Python dominates the machine learning ecosystem and has mature libraries for working with LLMs. Built-in asyncio support enables concurrent interaction with several APIs.

PyTorch \cite{paszke_pytorch_2019} serves as the deep learning framework. The framework provides GPU acceleration through CUDA, FP16/BF16 mixed precision to reduce memory consumption, and automatic differentiation, which gradient-based attacks require. Since version 2.0, PyTorch supports model compilation through \texttt{torch.compile}.

The \textbf{HuggingFace Transformers} library \cite{wolf_huggingface_2020} provides a unified API for loading, configuring, and running inference with pretrained models. The HuggingFace Hub offers more than 200{,}000 models, including the GPT, LLaMA, Mistral, and Gemma families and other architectures used in the experimental evaluation of Chapter~\ref{ch:ch4}.

High-throughput inference is implemented with vLLM \cite{kwon_vllm_2023}, which is built on the PagedAttention mechanism for efficient GPU memory management. Unlike the standard HuggingFace pipeline, vLLM implements continuous request batching, which substantially increases throughput in large-scale testing. Table~\ref{tab:vllm_comparison_ch5_part5} gives a quantitative performance comparison of the two approaches.

The REST API of the model router is built on the FastAPI + Uvicorn combination. FastAPI automatically generates OpenAPI documentation, validates input data through Pydantic schemas, and natively supports asynchronous request handling. Uvicorn, an ASGI server based on uvloop, handles thousands of concurrent connections with minimal overhead.

\textbf{PostgreSQL / file storage.} Experimental results are stored in a two-tier scheme. JSONL files hold intermediate results for fast access and are easy to version with Git. PostgreSQL serves long-term storage and analytical SQL queries over metrics.

\subsection{Performance Comparison of vLLM and the Standard HuggingFace Pipeline}
\label{subsec:vllm_comparison_ch5_part5}

The transition from the standard HuggingFace pipeline to vLLM was one of the key engineering decisions. The PagedAttention mechanism~\cite{kwon_vllm_2023}, which underlies vLLM, splits the attention KV cache into fixed-size blocks. The blocks are managed like pages of virtual memory in an operating system. This almost eliminates GPU memory fragmentation and increases the number of requests processed in parallel.

\begin{table}[ht]
\centering
\caption{Performance comparison of vLLM and the HuggingFace pipeline (model LLaMA-2-13B, GPU: NVIDIA A100 80GB, generation length: 256 tokens).}
\label{tab:vllm_comparison_ch5_part5}
\begin{tabular}{|p{6cm}|c|c|c|}
\hline
\textbf{Characteristic} & \textbf{HF Pipeline} & \textbf{vLLM} & \textbf{Gain} \\
\hline
Throughput (req./s) & 2.8 & 14.1 & $\times$5.0 \\
\hline
Latency (p50), ms & 1420 & 310 & $\times$4.6 \\
\hline
Latency (p99), ms & 3200 & 780 & $\times$4.1 \\
\hline
GPU memory utilization, \% & 62 & 91 & +29\,pp \\
\hline
Max. parallel requests & 4 & 32 & $\times$8.0 \\
\hline
Full evaluation time (1000~prompts)\textsuperscript{*}, min & 47.2 & 9.8 & $\times$4.8 \\
\hline
\multicolumn{4}{|p{0.93\textwidth}|}{\textsuperscript{*}\footnotesize The full evaluation of 1000 prompts includes multi-model evaluation by a committee of $k$ models with repeated requests for aggregation, as well as preprocessing/postprocessing overhead. At a throughput of 2.8~req./s, the time of 47.2~min corresponds to $47.2 \times 60 \times 2.8 \approx 7930$ model requests, i.e., ${\sim}8$ model calls per prompt.} \\
\hline
\end{tabular}
\end{table}

According to Table~\ref{tab:vllm_comparison_ch5_part5}, throughput increased fivefold, and the full evaluation time for 1000 prompts decreased from 47 to 10 minutes. The full evaluation takes longer than 1000 single requests would, because each prompt goes through multi-model evaluation by a committee, with about 8 model calls per prompt (see the note to Table~\ref{tab:vllm_comparison_ch5_part5}). For iterative attack algorithms such as ASA, where each generation of the evolutionary algorithm produces $N$ calls to the target model, such a speedup is critical.

The JudgeGuard and TrojanArmor suites are built on these technologies. The next section presents the MCPSec module, which extends both suites to the security evaluation of agentic LLM systems.

\section{Module for Evaluating the Robustness of Agentic LLM Systems (MCPSec)}
\label{sec:mcpsec_module_ch5_part5}

Agentic LLM systems form input sequences from the data of external tools through the Model Context Protocol (MCP) \cite{anthropic_mcp_2024}. To test such systems, the MCPSec module \cite{maloyan2025breaking} was developed and integrated into the four-component architecture described in Section~\ref{sec:architecture_ch5_part5}. As part of the integration, the AGM attack generator was extended with MCP injection templates, the MRAG router received adapters for the MCP transport, and the EE module was extended with specialized security metrics for agentic interaction. The MCP architecture and the testing mechanisms are described below, followed by the defense algorithms.

\subsection{Architecture of the Model Context Protocol}
\label{subsec:mcp_architecture_ch5_part5}

The Model Context Protocol defines a client-server architecture in which the LLM agent acts as the client and interacts with a set of servers that provide access to tools, data sources, and prompts. The protocol is based on the JSON-RPC 2.0 specification with two supported transport mechanisms:

\begin{itemize}
    \item \textbf{stdio (Standard I/O)}: a transport over the standard input and output streams of a child process, intended for local servers launched as child processes of the host application.
    \item \textbf{SSE (Server-Sent Events)}: an asynchronous HTTP transport for remote servers with support for streaming data transfer.
\end{itemize}

The main protocol methods include:
\begin{itemize}
    \item \texttt{tools/list}: retrieves the list of available tools with descriptions and JSON schemas of their parameters.
    \item \texttt{tools/call}: invokes a tool with the given arguments and returns a \texttt{CallToolResult} object with the result.
    \item \texttt{prompts/list}: retrieves the list of available prompt templates.
    \item \texttt{resources/read}: reads data from resources provided by the server.
\end{itemize}

\subsection{Interaction Sequence Diagram}
\label{subsec:mcp_sequence_ch5_part5}

Fig.~\ref{fig:mcp_sequence_ch5_part5} shows a UML sequence diagram of a typical interaction cycle between an LLM agent and an MCP server.

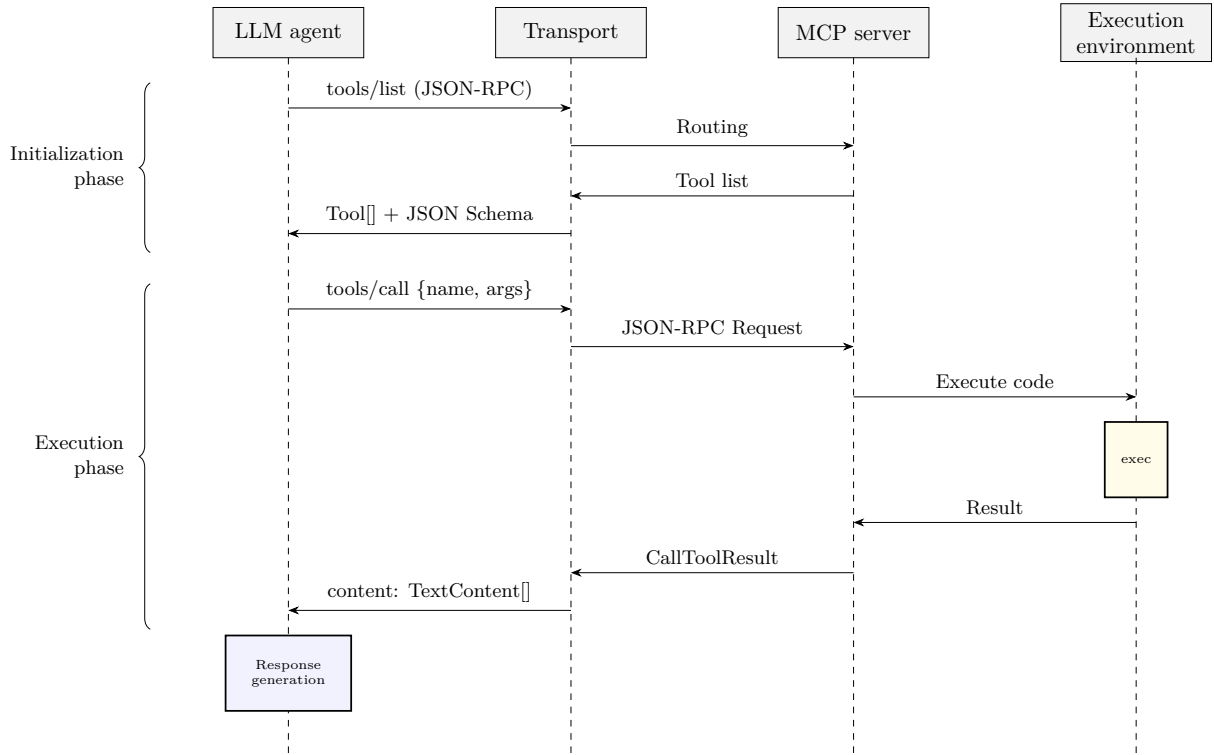
\begin{figure}[ht]
\centering
\resizebox{\textwidth}{!}{%
\begin{tikzpicture}[font=\small]
    \node[draw, fill=gray!10, minimum width=2.4cm, minimum height=0.8cm] (agent) at (0, 0) {LLM agent};
    \node[draw, fill=gray!10, minimum width=2.4cm, minimum height=0.8cm] (transport) at (4.5, 0) {Transport};
    \node[draw, fill=gray!10, minimum width=2.4cm, minimum height=0.8cm] (server) at (9, 0) {MCP server};
    \node[draw, fill=gray!10, minimum width=2.4cm, minimum height=0.8cm, align=center] (exec) at (13.5, 0) {Execution\\environment};

    \draw[dashed] (0, -0.4) -- (0, -11.5);
    \draw[dashed] (4.5, -0.4) -- (4.5, -11.5);
    \draw[dashed] (9, -0.4) -- (9, -11.5);
    \draw[dashed] (13.5, -0.4) -- (13.5, -11.5);

    \draw[-{Stealth}] (0, -1.2) -- node[above, font=\footnotesize] {tools/list (JSON-RPC)} (4.5, -1.2);
    \draw[-{Stealth}] (4.5, -1.8) -- node[above, font=\footnotesize] {Routing} (9, -1.8);
    \draw[{Stealth}-] (4.5, -2.6) -- node[above, font=\footnotesize] {Tool list} (9, -2.6);
    \draw[{Stealth}-] (0, -3.2) -- node[above, font=\footnotesize] {Tool[] + JSON Schema} (4.5, -3.2);

    \draw[-{Stealth}] (0, -4.4) -- node[above, font=\footnotesize] {tools/call \{name, args\}} (4.5, -4.4);
    \draw[-{Stealth}] (4.5, -5.0) -- node[above, font=\footnotesize] {JSON-RPC Request} (9, -5.0);
    \draw[-{Stealth}] (9, -5.8) -- node[above, font=\footnotesize] {Execute code} (13.5, -5.8);

    \draw[thick, fill=yellow!10] (13, -6.2) rectangle (14, -7.4);
    \node[font=\tiny] at (13.5, -6.8) {exec};

    \draw[{Stealth}-] (9, -7.8) -- node[above, font=\footnotesize] {Result} (13.5, -7.8);
    \draw[{Stealth}-] (4.5, -8.6) -- node[above, font=\footnotesize] {CallToolResult} (9, -8.6);
    \draw[{Stealth}-] (0, -9.2) -- node[above, font=\footnotesize] {content: TextContent[]} (4.5, -9.2);

    \draw[thick, fill=blue!5] (-1.0, -9.6) rectangle (1.0, -10.8);
    \node[font=\tiny, align=center] at (0, -10.2) {Response\\generation};

    \draw[decorate, decoration={brace, mirror, amplitude=5pt}] (-2.2, -0.8) -- (-2.2, -3.5) node[midway, left=8pt, font=\footnotesize, align=right] {Initialization\\phase};
    \draw[decorate, decoration={brace, mirror, amplitude=5pt}] (-2.2, -4.0) -- (-2.2, -9.5) node[midway, left=8pt, font=\footnotesize, align=right] {Execution\\phase};
\end{tikzpicture}%
}
\caption{UML sequence diagram of the interaction between an LLM agent and an MCP server.}
\label{fig:mcp_sequence_ch5_part5}
\end{figure}

\subsection{Protocol Vulnerabilities: From Prompt Injection to Tool Hijacking}
\label{subsec:mcp_vulnerabilities_ch5_part5}

MCP contains a critical vulnerability at the boundary between natural language processing and program code execution. The LLM agent decides which tool to call based on a context that may include data from untrusted sources. This creates the attack vector known as Prompt Injection to Tool Hijacking.

An analysis of the MCP specification revealed three architectural vulnerabilities:

\begin{enumerate}
    \item \textbf{Capability Attestation Absence.} MCP servers declare their own capabilities when the connection is initialized, and neither the client nor a third party verifies them. The protocol allows capabilities to change during a session without validation. A compromised server can escalate its privileges from restricted to unrestricted after the connection is established.

    \item \textbf{Unauthenticated Sampling.} The sampling mechanism allows servers to request the generation of LLM responses while assigning the ``user'' role to messages. The protocol does not authenticate the message source, and the user cannot visually distinguish messages injected by the server from their own requests. This opens a channel for direct protocol-level prompt injection.

    \item \textbf{Implicit Trust Propagation.} In configurations with several MCP servers, isolation boundaries are not defined. Tool responses from one server influence tool calls to another without tracking data provenance. A compromise of one server propagates privileges in a cascade through all connected services.
\end{enumerate}

The context of an LLM agent consists of the system prompt, the user request, and the results of previous tool calls. An adversary who controls call results through a compromised MCP server injects a malicious instruction into the returned data. The LLM agent does not distinguish instructions from data, so the modified context forces it to call an arbitrary tool with malicious arguments.

Threat~\ref{threat:prompt_injection} in Section~\ref{subsec:threat_descriptions} of Chapter~\ref{ch:ch1} formalizes prompt injection. The mechanism above is its indirect variant: the injected instruction arrives through tool data, not through the user query. In the MCP context, however, the consequences are more severe: tool calls lead to arbitrary code execution, file system modification, and outgoing network requests, and the side effects reach the real execution environment.

The software suite implements a module for automated testing of MCP servers for robustness to this class of attacks. The module generates malicious tool responses, tracks the subsequent agent calls, and evaluates whether the hijacking succeeds.

\subsection{Experimental Evaluation of Protocol Vulnerabilities: MCPBench}
\label{subsec:mcpbench_ch5_part5}

To quantify the impact of the architectural vulnerabilities of MCP on the security of agentic systems, the MCPBench research benchmark \cite{maloyan2025breaking} was developed. The benchmark combines attack scenarios on agentic LLM systems with the MCP infrastructure. Testing covered 847 scenarios with the Claude-3.5-Sonnet, GPT-4o, and Llama-3.1-70B models.

To isolate the contribution of protocol vulnerabilities from general LLM vulnerabilities, MCP implementations were compared with baseline non-MCP tool integrations. Table~\ref{tab:mcpbench_results_ch5_part5} summarizes the results.

\begin{table}[ht]
\centering
\caption{Comparison of attack success rates: MCP vs. baseline (non-MCP) tool integration.}
\label{tab:mcpbench_results_ch5_part5}
\begin{tabular}{|p{5.5cm}|c|c|c|}
\hline
\textbf{Attack type} & \textbf{Baseline ASR} & \textbf{MCP ASR} & \textbf{$\Delta$} \\
\hline
Indirect injection & 31.2\% & 47.8\% & +16.6\,pp \\
\hline
Tool response manipulation & 28.4\% & 52.1\% & +23.7\,pp \\
\hline
Cross-server propagation & 19.7\% & 61.3\% & +41.6\,pp \\
\hline
\textbf{Average} & \textbf{26.4\%} & \textbf{53.7\%} & \textbf{+27.3\,pp} \\
\hline
\end{tabular}
\end{table}

The architectural decisions of MCP increase the ASR against agentic systems by 16.6--41.6 percentage points (pp) relative to non-MCP integrations. For indirect injections, the increase was +16.6\,pp. Cross-server attacks gave the largest ASR increase, +41.6\,pp. In the latter case, the difference arises from the absence of isolation between servers and from implicit trust propagation.

The obtained ASR values refer to the studied agent implementations, the selected tool scenarios, and the specific MCP interaction configurations. They confirm the existence and practical significance of this class of vulnerabilities, but they do not cover the entire space of possible agentic architectures.

\subsection{Algorithmic Support for the Security of Agentic Interaction}
\label{subsec:protocol_security_ch5_part5}

Based on the MCP vulnerabilities identified in Section~\ref{subsec:mcp_vulnerabilities_ch5_part5}, the MCPSec module includes two algorithms for the software defense of agentic systems. The first, AttestMCP, implements cryptographic attestation of messages. The second, the Commit Boundary pattern, isolates code execution.

\subsubsection{Algorithm for Cryptographic Attestation and Message Routing (AttestMCP)}
\label{subsec:attest_mcp_ch5_part5}

The AttestMCP algorithm restricts the tool calls available to the LLM agent to those listed in a permission table fixed in advance, and it protects the integrity of each RPC packet with an HMAC signature. The defense rests on two distinct layers. The permission table defines authorization: it specifies which tools, and with which arguments, are allowed in the session. The HMAC signature ensures integrity: it authenticates the packet source at the transport layer and protects the packet from modification in transit.

Let $\mathcal{P} = \{(s_i, t_j, a_k)\}$ be the permission table, where $s_i$ is the session identifier, $t_j$ is the tool identifier, and $a_k$ is the specification of allowed arguments. Without loss of generality, $\mathcal{P}$ holds at most one entry per session--tool pair, because several argument specifications for one pair merge into their union. The table is built during the session initialization phase, before data from untrusted sources enter the agent context. During the execution phase, $\mathcal{P}$ remains unchanged, and the LLM agent has no operation that extends it. This restriction of the reachable agent actions to $\mathcal{P}$ is the main defense mechanism, and it does not depend on the context content.

Each RPC packet carries a signature $\sigma = \text{HMAC}(K_{session}, s_i \| t_j \| \text{hash}(\text{args}))$, where $K_{session}$ is the session key established when the MCP connection is initialized, $\text{args}$ are the actual call arguments, and $\|$ denotes unambiguous (length-prefixed) concatenation. The signature binds the packet to the session and the tool and protects the integrity of the transmitted arguments through their hash.

\vspace{5pt}
\hrule
\vspace{5pt}
\textbf{Algorithm 3: AttestMCP, cryptographic attestation of tool calls}
\vspace{2pt}
\hrule
\begin{algorithmic}[1]
    \State \textbf{Input:} RPC packet $\text{pkt} = (\text{method}, \text{params}, \sigma)$, permission table $\mathcal{P}$, session key $K_{session}$.
    \State \textbf{Output:} Execution result or rejection.
    \State Extract $(\text{session\_id}, \text{tool\_id}, \text{args}) \leftarrow \text{ParseRPC}(\text{pkt})$.
    \State Compute $\sigma' \leftarrow \text{HMAC}(K_{session}, \text{session\_id} \| \text{tool\_id} \| \text{hash}(\text{args}))$.
    \If{$\sigma \neq \sigma'$}
        \State \textbf{Return} $\text{Error}(\text{INVALID\_SIGNATURE})$. \Comment{Invalid signature}
    \EndIf
    \State Find the entry $(s_i, t_j, a_k) \in \mathcal{P}$ with $s_i = \text{session\_id}$, $t_j = \text{tool\_id}$.
    \If{no entry found}
        \State \textbf{Return} $\text{Error}(\text{PERMISSION\_DENIED})$. \Comment{No permission}
    \EndIf
    \If{$\text{args} \not\in \text{AllowedArgs}(a_k)$}
        \State \textbf{Return} $\text{Error}(\text{ARGS\_VIOLATION})$. \Comment{Disallowed arguments}
    \EndIf
    \State $\text{result} \leftarrow \text{ExecuteTool}(\text{tool\_id}, \text{args})$. \Comment{Route to the server}
    \State Write $(\text{pkt}, \text{result}, \text{timestamp})$ to the audit log.
    \State \textbf{Return} $\text{CallToolResult}(\text{result})$.
\end{algorithmic}
\vspace{5pt}
\hrule
\vspace{5pt}

\textbf{Complexity analysis.} Computing HMAC-SHA256 takes $O(|\text{args}|)$. The lookup of an entry in the permission table, implemented with a hash table, runs in amortized $O(1)$. Argument validation against the JSON schema requires $O(|\text{args}| \cdot |\text{schema}|)$ comparison operations in the worst case, which is linear in $|\text{args}|$ for a fixed schema.

In practice, the total overhead did not exceed 0.1\,ms per call for $|\text{args}| \le 4096$~bytes. Measurements were taken on an Intel Xeon Gold 6248R CPU with a clock frequency of 3.0~GHz, and LLM inference ran on an NVIDIA A100. The attestation time was three orders of magnitude lower than the typical inference latency of 300--800\,ms for models of the LLaMA-2-13B class.

\textbf{Limits of the guarantees.} The authorization and integrity layers solve different tasks, and their functions are not interchangeable. The HMAC signature confirms the source and integrity of a packet, but not the intent of the call. A tool call that the LLM agent forms under the influence of an injection passes through the same legitimate client as a call made at the user's request, and it receives a valid signature. The signature by itself does not separate a call authorized by the user from a call imposed by an injection. The permission table provides this separation: an injection does not take the agent outside $\mathcal{P}$, because the table is fixed before untrusted data arrive and does not change during the session.

The cryptographic layer closes a separate vector: packet forgery within the cross-server trust propagation from Section~\ref{subsec:mcp_vulnerabilities_ch5_part5}. A compromised MCP server does not hold the key $K_{session}$, which is stored at the transport layer outside the context window of the model. Such a server cannot forge a packet on behalf of the client or of another server, which blocks privilege escalation through forged packets. A tool call that the agent itself forms after reading a malicious response from a compromised server still receives a valid signature. The permission table $\mathcal{P}$, not the signature, bounds such calls.

AttestMCP leaves part of the risk in place. If an injection forces the agent to call an allowed tool with allowed arguments for a malicious purpose, the packet passes both checks. The argument specification $a_k$ and the Commit Boundary pattern from Section~\ref{subsec:commit_boundary_ch5_part5} narrow this residual risk: the narrower the allowed set of arguments, the smaller the space for abuse within $\mathcal{P}$.

\textbf{Formalization of AttestMCP as a finite automaton.}
The verification protocol is modeled as a finite automaton $\mathcal{M} = (S, \Sigma, \delta, s_{init}, F)$. The automaton has five states: $S = \{s_{init},\allowbreak s_{auth},\allowbreak s_{exec},\allowbreak s_{return},\allowbreak s_{error}\}$. The input alphabet $\Sigma$ includes RPC packets, verification results, and execution results. The final states are $F = \{s_{return}, s_{error}\}$.

The transitions are as follows. On receipt of an RPC packet, the automaton moves from $s_{init}$ to $s_{auth}$, that is, $s_{init} \xrightarrow{\text{pkt}} s_{auth}$. The transition $s_{auth} \to s_{exec}$ occurs only if the signature matches, $\sigma = \sigma'$, and an entry $(s_i, t_j, a_k) \in \mathcal{P}$ with allowed arguments exists. A signature mismatch, a missing entry, or a violation of the argument specification leads to $s_{error}$. Completion of execution produces $s_{exec} \xrightarrow{\text{ok}} s_{return}$ or $s_{exec} \xrightarrow{\text{fail}} s_{error}$.

The automaton maintains the following invariant, which states that every executed call satisfies the permission table:
$$
\forall\, \text{path } (s_{init}, \ldots, s_{exec}): \quad \exists\, a_k:\ (\text{session\_id}, \text{tool\_id}, a_k) \in \mathcal{P} \;\wedge\; \text{args} \in \text{AllowedArgs}(a_k).
$$
Since $\mathcal{P}$ is fixed during the initialization phase, the invariant bounds the set of agent actions from above, independently of the instructions injected into its context. Prompt injection does not extend permissions beyond $\mathcal{P}$, and HMAC verification prevents an adversary from circumventing this restriction by substituting packets at the transport layer.

\subsubsection{Algorithm for Software Execution Isolation: The Commit Boundary Pattern}
\label{subsec:commit_boundary_ch5_part5}

Code generated by an LLM agent must not run directly in the host environment. The Commit Boundary pattern runs it in a containerized environment and validates the result before passing it back, which implements the principle of least privilege.

\vspace{5pt}
\hrule
\vspace{5pt}
\textbf{Algorithm 4: Commit Boundary, isolated code execution}
\vspace{2pt}
\hrule
\begin{algorithmic}[1]
    \State \textbf{Input:} LLM-generated code $\text{code} \in \Yspace$; resource limits $\text{limits} = (T_{max}, M_{max}, \text{net})$;
    \State \phantom{\textbf{Input:}} HITL (human-in-the-loop) flag $h \in \{0, 1\}$.
    \State \textbf{Output:} Execution result or rejection.
    \State $\text{ast} \leftarrow \text{ParseAST}(\text{code})$. \Comment{Parse the abstract syntax tree}
    \If{$\text{ast} = \text{NULL}$}
        \State \textbf{Return} $\text{Error}(\text{PARSE\_FAILED})$. \Comment{Syntax error}
    \EndIf
    \State $\text{risk} \leftarrow \text{StaticAnalysis}(\text{ast})$. \Comment{Static analysis: syscall, network, fs}
    \If{$\text{risk} > \text{THRESHOLD\_HIGH}$}
        \If{$h = 0$}
            \State \textbf{Return} $\text{Error}(\text{HIGH\_RISK\_BLOCKED})$.
        \Else
            \State $\text{approved} \leftarrow \text{RequestHumanApproval}(\text{code}, \text{risk})$.
            \If{$\neg \text{approved}$}
                \State \textbf{Return} $\text{Error}(\text{HUMAN\_REJECTED})$.
            \EndIf
        \EndIf
    \EndIf
    \State $\text{container} \leftarrow \text{AllocateContainer}(\text{limits})$. \Comment{Docker with cgroups limits}
    \State Configure $\text{container}$: timeout $T_{max}$, memory $M_{max}$, network $\text{net}$.
    \State $(\text{result}, \text{exit\_code}) \leftarrow \text{RunInContainer}(\text{container}, \text{code})$.
    \State $\text{ReleaseContainer}(\text{container})$. \Comment{Release resources}
    \If{$\text{exit\_code} \neq 0$ or timeout}
        \State \textbf{Return} $\text{Error}(\text{EXECUTION\_FAILED}, \text{exit\_code})$.
    \EndIf
    \State $\text{validated} \leftarrow \text{ValidateOutput}(\text{result})$. \Comment{Validate the output data}
    \If{$\neg \text{validated}$}
        \State \textbf{Return} $\text{Error}(\text{OUTPUT\_REJECTED})$.
    \EndIf
    \State Write $(\text{code}, \text{result}, \text{risk}, \text{timestamp})$ to the audit log.
    \State \textbf{Return} $\text{CommitResult}(\text{result})$. \Comment{Commit the result}
\end{algorithmic}
\vspace{5pt}
\hrule
\vspace{5pt}

\textbf{Complexity analysis.} The execution of the code in the container dominates, and the timeout $T_{max}$ bounds it. The isolation overhead consists of the following parts:
\begin{itemize}
    \item AST parsing: $O(|\text{code}|)$, linear in the code size.
    \item Static analysis: $O(|\text{ast}| \cdot |\text{rules}|)$, where $|\text{rules}|$ is the number of detection rules, typically 50--100.
    \item Docker container allocation: 200--500\,ms on a cold start, about 50\,ms with a pool of pre-created containers.
\end{itemize}
In total, the overhead is 0.3--0.8\,s on a cold start, which is comparable to the LLM code generation time and acceptable for interactive scenarios.

Three layers of defense act in sequence: static analysis at the parsing stage, container isolation at execution, and output validation before the commit. For high-risk operations, the HITL mode is activated: execution then requires explicit human approval, so the required level of trust grows with the estimated risk.

\subsubsection{Components of the Integrated Defense in the MCPSec Module}
\label{subsec:mcpsec_ch5_part5}

The MCPSec module~\cite{maloyan2025breaking} combines the AttestMCP and Commit Boundary algorithms with four additional defense mechanisms.

\textbf{Capability Certificates} restrict the allowed set of operations of each server and thereby prevent privilege escalation. Origin Tagging separates message sources: a visual and programmatic tag distinguishes the output of the sampling mechanism from user requests. Without explicit user permission, the Isolation Enforcement mechanism blocks cross-server data transfer. Replay Protection based on timestamps and one-time nonces prevents the replay of intercepted messages.

On the MCPBench benchmark, MCPSec reduced the average attack success rate from 53.7\% to 12.4\%, that is, by 41.3~pp. The additional steady-state processing latency per message for attestation and the MCPSec filtering mechanisms did not exceed 8.3\,ms. This value does not include the perplexity filter and content moderation of Table~\ref{tab:co_decomposition_ch5_part5}. Messages that trigger sandboxed code execution additionally incur the container allocation cost, about 50\,ms with a pool of pre-created containers. Both values are small relative to the LLM inference time of 300--800\,ms for models of the LLaMA-2-13B class.

\section{Evaluation of the Computational Complexity and Performance of the Software Suite}
\label{sec:performance_ch5_part5}

This section gives a quantitative evaluation of the performance of the software suite in terms of two indicators: the dependence of latency on the input context length with active defense filters, and the throughput under different inference configurations.

\subsection{Latency Analysis as a Function of Context Length}
\label{subsec:latency_analysis_ch5_part5}

The processing time of one request consists of four components: preprocessing, which includes tokenization and prompt formatting; model inference; the work of the defense filters; and metric computation. Let $n = |x|$ denote the length of the input sequence in tokens.

For Transformer models, the cost of self-attention grows quadratically with the context length, as $O(n^2 \cdot d)$, where $d$ is the model dimension. The projection and feed-forward layers add a term $O(n \cdot d^2)$, which dominates at moderate context lengths ($d = 5120$ for LLaMA-2-13B). Therefore, in the measured range inference time grows superlinearly but not yet quadratically, and the quadratic term dominates only asymptotically. The defense filters, namely the perplexity filter and content moderation, each run in $O(n)$. Fig.~\ref{fig:latency_plot_ch5_part5} shows the dependence of the total latency on the context length for different configurations.

\begin{figure}[ht]
\centering
\begin{tikzpicture}[scale=0.9]
    \draw[-{Stealth}, thick] (0,0) -- (11,0) node[right] {$n$ (tokens)};
    \draw[-{Stealth}, thick] (0,0) -- (0,8) node[above] {$T$ (ms)};

    \foreach \x/\label in {1/512, 3/1024, 5/2048, 7/4096, 9/8192} {
        \draw (\x, -0.1) -- (\x, 0.1);
        \node[below, font=\footnotesize] at (\x, -0.15) {\label};
    }

    \foreach \y/\label in {1/200, 2/400, 3/600, 4/800, 5/1000, 6/1200, 7/1400} {
        \draw (-0.1, \y) -- (0.1, \y);
        \node[left, font=\footnotesize] at (-0.15, \y) {\label};
    }

    \foreach \x in {1,3,5,7,9} {
        \draw[gray!20] (\x, 0) -- (\x, 7.5);
    }
    \foreach \y in {1,2,3,4,5,6,7} {
        \draw[gray!20] (0, \y) -- (10.5, \y);
    }

    \draw[blue, thick, mark=*] plot coordinates {
        (1, 0.75) (3, 1.05) (5, 1.55) (7, 3.10) (9, 6.90)
    };

    \draw[red, thick, mark=square*] plot coordinates {
        (1, 0.90) (3, 1.25) (5, 1.90) (7, 3.50) (9, 7.45)
    };

    \draw[green!50!black, thick, mark=triangle*] plot coordinates {
        (1, 0.95) (3, 1.325) (5, 1.99) (7, 3.65) (9, 7.70)
    };

    \draw[blue, thick] (5.5, 7.5) -- (6.5, 7.5) node[right, font=\footnotesize, black] {Inference only};
    \node[blue, font=\footnotesize] at (6, 7.5) {$\bullet$};
    \draw[red, thick] (5.5, 7.0) -- (6.5, 7.0) node[right, font=\footnotesize, black] {Inference + filters};
    \node[red, font=\footnotesize] at (6, 7.0) {$\blacksquare$};
    \draw[green!50!black, thick] (5.5, 6.5) -- (6.5, 6.5) node[right, font=\footnotesize, black] {Full pipeline};
    \node[green!50!black, font=\footnotesize] at (6, 6.5) {$\blacktriangle$};
\end{tikzpicture}
\caption{Dependence of request processing latency on the input context length (model LLaMA-2-13B, vLLM, GPU: NVIDIA A100).}
\label{fig:latency_plot_ch5_part5}
\end{figure}
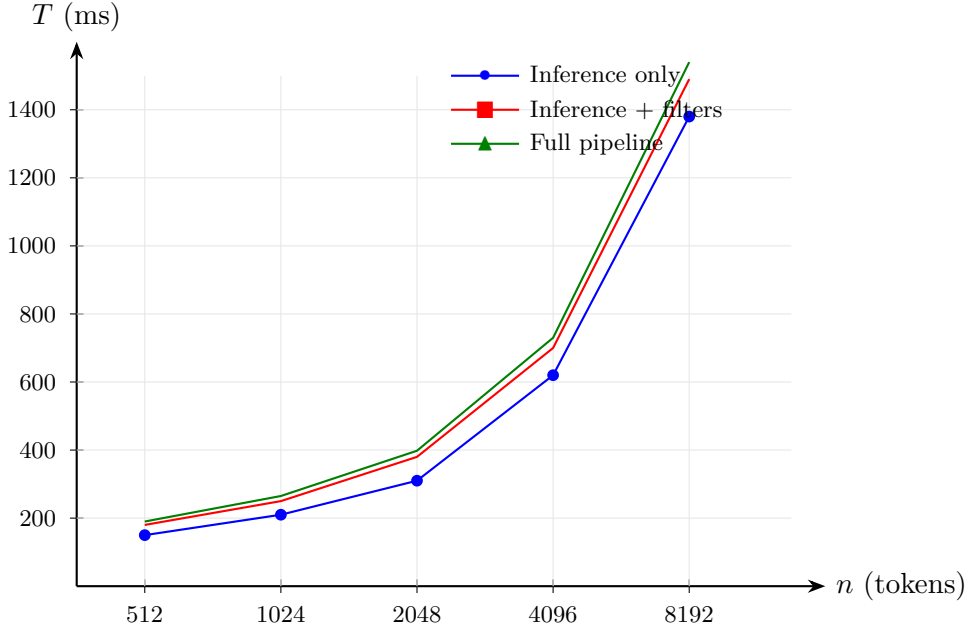

Fig.~\ref{fig:latency_plot_ch5_part5} shows that the overhead of the defense components is about 28\% at the typical context length $n = 2048$, which agrees with the CO decomposition in Table~\ref{tab:co_decomposition_ch5_part5}, and decreases to about 12\% at $n = 8192$. The filters run in $O(n)$, while inference time grows superlinearly, asymptotically as $O(n^2)$. Therefore, the share of the filters in the total time decreases as $n$ grows. The defense mechanisms do not become a bottleneck.

\subsection{Throughput Comparison Across Configurations}
\label{subsec:throughput_comparison_ch5_part5}

Throughput was measured for four inference configurations. The model is LLaMA-2-13B, the generation length is 256 tokens, and the test sample size is 1000 requests:

\begin{enumerate}
    \item \textbf{HF-seq}: HuggingFace pipeline, sequential processing;
    \item \textbf{HF-batch}: HuggingFace pipeline, batch processing (batch size = 8);
    \item \textbf{vLLM-base}: vLLM without defense filters;
    \item \textbf{vLLM-full}: vLLM with active defense filters and metric computation.
\end{enumerate}

Fig.~\ref{fig:throughput_chart_ch5_part5} presents the results.

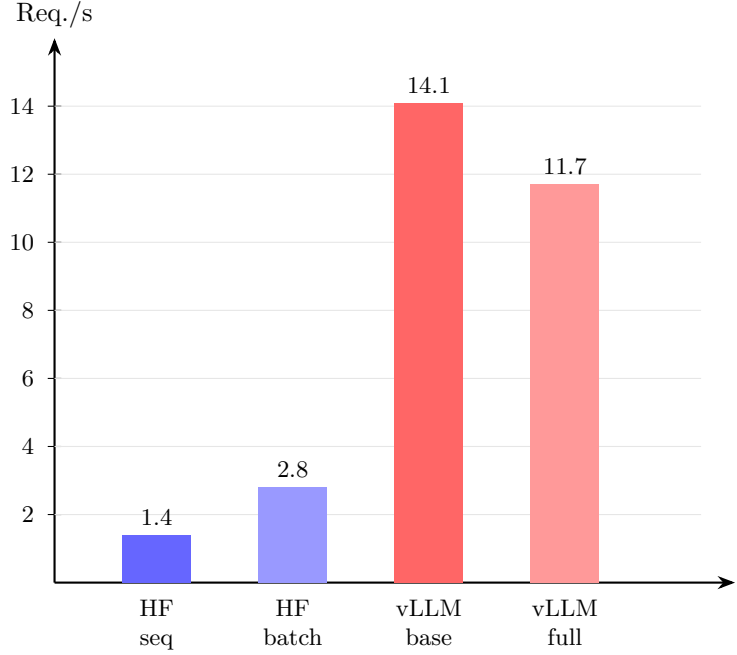
\begin{figure}[ht]
\centering
\begin{tikzpicture}[scale=0.9]
    \draw[-{Stealth}, thick] (0,0) -- (10,0);
    \draw[-{Stealth}, thick] (0,0) -- (0,8) node[above, font=\small] {Req./s};

    \foreach \y/\label in {1/2, 2/4, 3/6, 4/8, 5/10, 6/12, 7/14} {
        \draw (-0.1, \y) -- (0.1, \y);
        \node[left, font=\footnotesize] at (-0.15, \y) {\label};
        \draw[gray!20] (0, \y) -- (9.5, \y);
    }

    \fill[blue!60] (1, 0) rectangle (2, 0.7);
    \node[above, font=\footnotesize] at (1.5, 0.7) {1.4};
    \node[below, font=\footnotesize, align=center] at (1.5, -0.1) {HF\\seq};

    \fill[blue!40] (3, 0) rectangle (4, 1.4);
    \node[above, font=\footnotesize] at (3.5, 1.4) {2.8};
    \node[below, font=\footnotesize, align=center] at (3.5, -0.1) {HF\\batch};

    \fill[red!60] (5, 0) rectangle (6, 7.05);
    \node[above, font=\footnotesize] at (5.5, 7.05) {14.1};
    \node[below, font=\footnotesize, align=center] at (5.5, -0.1) {vLLM\\base};

    \fill[red!40] (7, 0) rectangle (8, 5.85);
    \node[above, font=\footnotesize] at (7.5, 5.85) {11.7};
    \node[below, font=\footnotesize, align=center] at (7.5, -0.1) {vLLM\\full};
\end{tikzpicture}
\caption{Throughput comparison across inference configurations (model LLaMA-2-13B, GPU: NVIDIA A100).}
\label{fig:throughput_chart_ch5_part5}
\end{figure}

vLLM-base reaches 14.1 req./s, which exceeds HF-seq (1.4 req./s) by a factor of 10 and HF-batch (2.8 req./s) by a factor of 5. With the full set of defense filters and metric computation enabled (vLLM-full), throughput drops to 11.7 req./s, a loss of 17\% relative to the base configuration. The PagedAttention mechanism preserves high performance even with active security components.

\subsection{Scalability in Multi-Model Evaluation}
\label{subsec:scalability_ch5_part5}

In the evaluation of model ensembles, described as committees in Section~\ref{subsec:committee_pseudocode} of Chapter~\ref{ch:ch3}, the scalability of the system with a growing number of models $k$ becomes a critical parameter. The MRAG router implements two operating modes:

\begin{itemize}
    \item \textbf{Parallel inference (GPU sharing):} several models are placed on one GPU with the NVIDIA Multi-Process Service (MPS) mechanism. Throughput scales sublinearly: $\Theta_k \approx \Theta_1 \cdot k^{0.7}$, where $\Theta_k$ is the total throughput for $k$ models and $\Theta_1$ is the throughput for a single model.

    \item \textbf{Distributed inference (multi-GPU):} each model is placed on a separate GPU. Throughput scales linearly: $\Theta_k \approx k \cdot \Theta_1$, subject to the interconnect communication overhead ($\sim$2--5\% loss with NVLink).
\end{itemize}

The experiments described in Chapter~\ref{ch:ch4} used a configuration with 4 NVIDIA A100 80GB GPUs, which can simultaneously serve up to 4 models with 13B parameters or 2 models with 70B parameters when tensor parallelism is used. In the distributed mode, where each model is placed on a separate GPU, the linear estimate gives $\Theta_4 \approx 4 \cdot 14.1 = 56.4$ req./s, or about 54--55 req./s after the 2--5\% interconnect loss. The measured total throughput reached up to 45 model inference requests per second for four 13B models. The overhead of request orchestration and of the synchronization of committee verdicts explains the gap from the linear estimate. The evaluation of a dataset of 10{,}000 prompts with a committee of $k$ models generates $10{,}000 \cdot k$ model requests, in particular $30{,}000$ for $k=3$. For $k=4$, the $40{,}000$ requests at the measured throughput of about 45 req./s take $\sim$15 minutes. For $k=3$, the committee occupies three of the four GPUs, so the throughput scales to about $3/4 \cdot 45 \approx 34$ req./s and the $30{,}000$ requests also take $\sim$15 minutes. The time drops to $\sim$11 minutes only if all four GPUs serve the three-model committee, for example with a replica of one model.

\subsection{Computational Overhead of the Security Algorithms}
\label{subsec:co_analysis_ch5_part5}

The total Computational Overhead (CO) of the software suite is defined as the ratio of the additional time spent on the security components to the base inference time:
$$
\text{CO} = \frac{T_{\text{total}} - T_{\text{inference}}}{T_{\text{inference}}} \times 100\%.
$$

Measurements were taken for a typical configuration: the LLaMA-2-13B model, a context length of 2048 tokens, and the vLLM backend. Table~\ref{tab:co_decomposition_ch5_part5} gives the decomposition of the overhead.

\begin{table}[ht]
\centering
\caption{Decomposition of the computational overhead of the security components.}
\label{tab:co_decomposition_ch5_part5}
\begin{tabular}{|l|c|c|}
\hline
\textbf{Component} & \textbf{Latency, ms} & \textbf{CO share, \%} \\
\hline
Base inference (vLLM) & 310 & {--}  \\
\hline
Perplexity filter & 42 & 13.5 \\
\hline
Content moderation & 28 & 9.0 \\
\hline
AttestMCP attestation & 0.08 & 0.03 \\
\hline
Metric computation & 15 & 4.8 \\
\hline
Logging and storage & 3 & 1.0 \\
\hline
\textbf{Total CO} & \textbf{88.08} & \textbf{28.4} \\
\hline
\end{tabular}
\end{table}

According to Table~\ref{tab:co_decomposition_ch5_part5}, the total computational overhead is 28.4\% of base inference. The perplexity filter contributes the largest share, 13.5\%, whereas AttestMCP adds only 0.03\%, which makes it suitable for high-load real-time systems. This overhead level shows that the methods of Chapters~\ref{ch:ch2}--\ref{ch:ch3} can be implemented at an acceptable computational cost.

\section*{Chapter Conclusions}

The methods of Chapters~\ref{ch:ch2}--\ref{ch:ch3} are implemented in two registered software suites: JudgeGuard, certificate No.~2025687702, and TrojanArmor, certificate No.~2025684247.

A common four-component architecture underlies both suites. The AGM attack generator, the MRAG model router, the EE evaluation engine, and the SL storage interact through well-defined interfaces, which allows new models, metrics, and attack types to be added without modifying the core. The replacement of the HuggingFace pipeline with vLLM and PagedAttention increased throughput fivefold on an NVIDIA A100. For iterative algorithms such as ASA, which need $N \cdot G = 1000$ model requests per attack in the experiments of Chapter~\ref{ch:ch4}, this gain is critical.

The MCPSec module extends both suites to agentic LLM systems that use MCP. It includes the MCPBench benchmark of 847 attack scenarios, AttestMCP cryptographic attestation (Algorithm~3), and Commit Boundary isolated execution (Algorithm~4). On MCPBench, the average ASR decreased by 41.3~pp, from 53.7\% to 12.4\%, compared with the unprotected configuration, as shown in Section~\ref{subsec:mcpsec_ch5_part5}.

The total computational overhead of the security components does not exceed 30\% of base inference, according to Table~\ref{tab:co_decomposition_ch5_part5}. The filters run in $O(n)$, whereas inference time grows superlinearly, asymptotically as $O(n^2)$. Consequently, the relative share of the defense components decreases as the context length grows. The current implementation is tied to specific provider APIs: when API versions change, the MRAG adapters will require updates. The performance measurements were taken on a single NVIDIA A100 hardware configuration, and extrapolation to other accelerators requires verification.

JudgeGuard and TrojanArmor cover the full LLM security evaluation cycle. Chapter~\ref{ch:ch6} describes the practical application of the suites.

\chapter{Practical Application of Robustness Enhancement Methods}\label{ch:ch6}

\section*{Chapter Overview}
A number of practical tasks separate laboratory experiments from production systems: the choice of metrics for deployment effectiveness, the reconciliation of security requirements with the computational budget, and the adaptation of defenses to the specifics of the application domain. The methodology for deploying defense mechanisms relies on technical and economic evaluation metrics. The practical examples for NLP systems and safety-critical applications are complemented by an analysis of technology development trends.

\section{Methodology for Deploying Robustness Enhancement Methods}
\label{sec:implementation_methodology_ch6}

The deployment of LLM robustness enhancement methods in real systems is not limited to the application of individual algorithms. It requires a systematic approach that takes into account the specifics of the task, the available resources, and the security requirements. The discrete token space of LLMs imposes additional constraints on the choice and adaptation of methods.

\subsection{Stages of Deploying Robustness Enhancement Methods}
\label{subsec:implementation_steps_ch6}

We propose a deployment methodology based on the life cycle of risk management and secure software development \cite{dodson_secure_2022, nist_ai_rmf_2023}. The methodology comprises five consecutive stages.

\textbf{1. Requirements analysis and risk assessment.}
The robustness and security requirements of the system are determined from its purpose and the potential consequences of failures.
Risk assessment uses standard methodologies, in particular the NIST AI Risk Management Framework \cite{nist_ai_rmf_2023} and
the OWASP Top 10 for LLM~\cite{owasp_foundation_owasp_2023}, described in Chapter~\ref{ch:ch1}.
At the same stage, the relevant threats listed in Section~\ref{subsec:threat_descriptions} of Chapter~\ref{ch:ch1} are identified:
prompt injection, including attacks on LLM-as-a-Judge and on defended systems,
and trojan attacks, taking into account how hard they are to detect.
The result of this stage is recorded as a prioritized list of risks:
$$
\mathcal{R}_{\text{priority}} = \text{Prioritize}(\{(a_i, v_j, P_{\text{exploit}}(v_j|a_i), \text{ImpactValue}(v_j))\}),
$$
where $a_i \in \mathcal{A}^{\mathrm{type}}$ is an attack type, $v_j \in \mathcal{V}$ is a vulnerability, $P_{\text{exploit}}$ is the probability of exploitation, and $\text{ImpactValue}$ is a quantitative estimate of the damage, as defined in Definitions~\ref{def:threat_space} and~\ref{def:attack_risk} of Chapter~\ref{ch:ch1}. The criticality of each attack-vulnerability pair is computed as $c_{ij} = P_{\text{exploit}}(v_j|a_i) \cdot \text{ImpactValue}(v_j)$ and is used for prioritization.

\begin{remark}
\label{princ:unknown_attack_defense_ch6}
\textbf{Empirical observation: limited effectiveness of defenses against unknown attacks.}
The effectiveness of defense and robustness enhancement methods developed and tested on a known set of attacks $\mathcal{A}_{\text{known}}$ can decrease when they face new, unforeseen attacks $\mathcal{A}_{\text{new}}$ under real conditions. The degree of this decrease depends on the ability of the defense mechanism to generalize to new types of attacks \cite{athalye2018obfuscated}. Even general defense principles have limits of generalization if an attack targets vulnerabilities that were not considered in the design of the defense mechanism.
Formally, if $R_{\mathcal{A}}(f)$ is a robustness or security metric with respect to the attacks from $\mathcal{A}$, then the empirically observed tendency is $R_{\mathcal{A}_{\text{new}}}(f) \leq R_{\mathcal{A}_{\text{known}}}(f)$. The decrease ($R_{\mathcal{A}_{\text{known}}}(f) - R_{\mathcal{A}_{\text{new}}}(f)$) depends on the ``distance'' between $\mathcal{A}_{\text{new}}$ and $\mathcal{A}_{\text{known}}$ and on the generalization ability of the defense.
\end{remark}
\textbf{Justification.}
Defense mechanisms are often optimized to counter specific patterns or characteristics of known attacks. New attacks can use other vectors or bypass specific assumptions made during the development of the defense. The generalization ability of a defense depends on the principles built into it. For example, a Lipschitz constraint generalizes better than filters based on signatures of known attacks.

\textbf{2. Selection of robustness enhancement methods.}
Based on the results of the risk analysis, the defense and robustness enhancement methods described in Chapter~\ref{ch:ch3} are selected. The selection takes the following factors into account:
\begin{itemize}
    \item The effectiveness of the method against the prioritized risks. The experimental data of Chapter~\ref{ch:ch4} confirm the effectiveness of committees of models against attacks on LLM-as-a-Judge.
    \item Compatibility with the architecture and size of the target model.
    \item The ratio of the computational requirements to the available resources. The requirements are expressed through the computational overhead: $\Delta\text{Cost}$ in Chapter~\ref{ch:ch4} and CO in Section~\ref{subsec:co_analysis_ch5_part5} of Chapter~\ref{ch:ch5}.
    \item The impact on model utility. The acceptable bounds are determined by the RUT metric from Chapter~\ref{ch:ch4}.
\end{itemize}
The experimental data of Chapter~\ref{ch:ch4} and the principle of multi-layer defense from Section~\ref{subsec:multilayer_defense_ch3_part3} of Chapter~\ref{ch:ch3} support the combination of methods. Formally, the selection reduces to an optimization problem: maximize the reduction of the total risk under a limited budget. In practice, the solution often requires expert judgment.

\textbf{3. Implementation and integration.}
The selected methods are implemented and integrated into the system. At this stage, where needed, the model architecture is modified according to Section~\ref{sec:arch_robustness_methods_ch3_part3} of Chapter~\ref{ch:ch3}, robust training or fine-tuning is performed according to Section~\ref{sec:training_robustness_methods_ch3_part3}, and data pre- and post-processing modules are built: input filters and sanitizers, and output verifiers according to Section~\ref{subsec:multilayer_defense_ch3_part3}. In parallel, the generation parameters are configured and monitoring systems are connected.

\textbf{4. Testing and validation.}
The deployed defense mechanisms are evaluated with the metrics from Chapter~\ref{ch:ch4}: robustness to known attacks by ASR, trojan detection quality by Recall and REASR, utility on legitimate queries by U-metrics, and the robustness-utility trade-off by RUT. In parallel, red teaming \cite{perez_red_2022, ganguli_red_2022} is carried out to find new vulnerabilities, including automated approaches validated in the Kaggle competitions and the SaTML CTF. The testing results guide the iterative refinement of the defense.

\textbf{5. Monitoring and updating.}
After deployment, data on queries, responses, and defense triggers are collected continuously. The threat profile is tracked through new publications and competition results. Periodic re-testing and red teaming reveal degradation of the defense. The rapid evolution of LLMs and attack methods makes this stage decisive for the long-term security of the system.

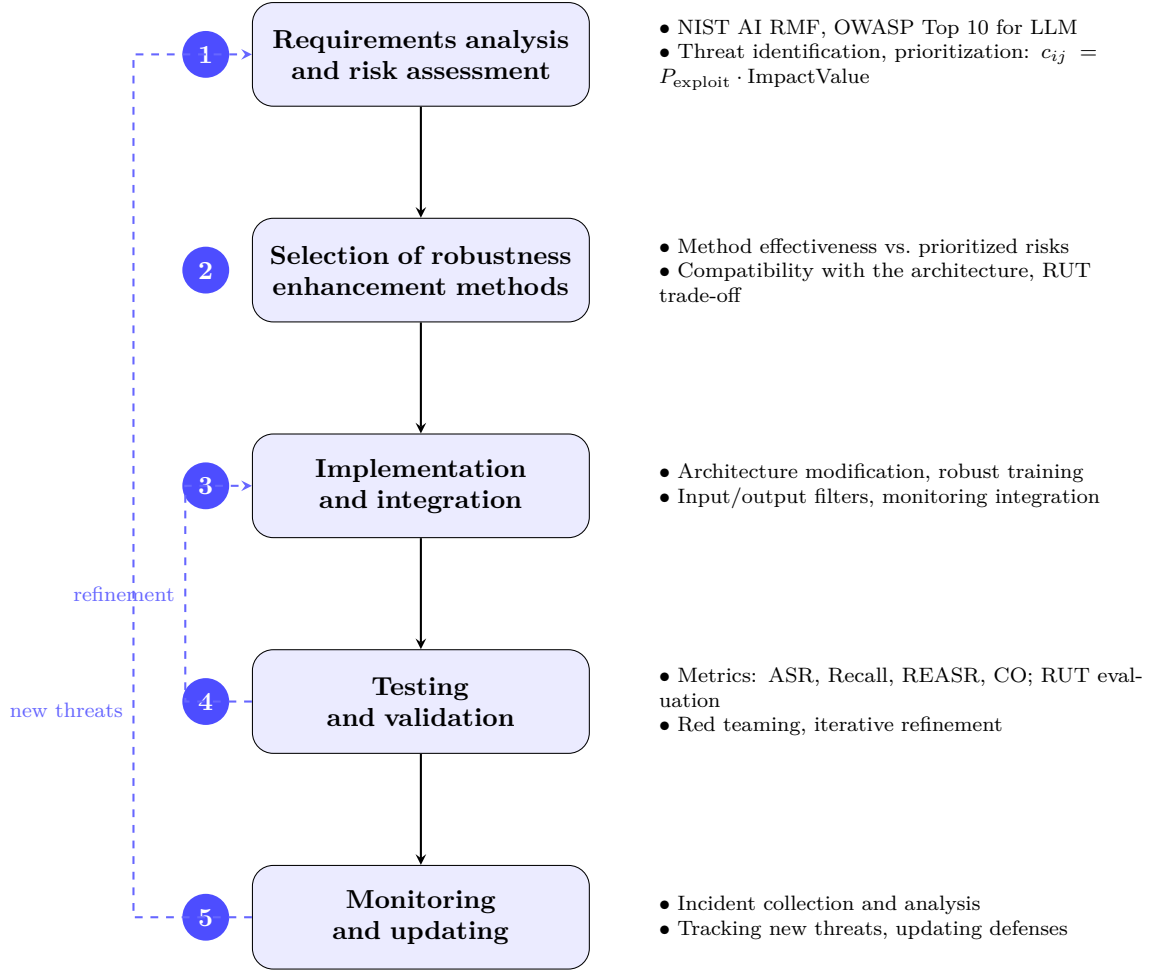
\begin{figure}[ht!]
\centering
\resizebox{0.95\linewidth}{!}{%
\begin{tikzpicture}[
    stage/.style={
        draw, rounded corners=8pt,
        text width=4.2cm,
        minimum height=1.4cm,
        align=center, font=\small,
        fill=blue!8,
        inner sep=5pt},
    substage/.style={font=\scriptsize, text width=6.5cm, align=left},
    arrow/.style={->, thick, >=stealth},
    feedback/.style={->, thick, >=stealth, dashed, color=blue!60},
    numlabel/.style={circle, fill=blue!70, text=white, font=\bfseries\small,
                     inner sep=2pt, minimum size=18pt}
]

\node[stage] (s1) at (0, 0) {
    \textbf{Requirements analysis}\\
    \textbf{and risk assessment}
};
\node[numlabel, left=0.3cm of s1.west] {1};

\node[stage, below=1.5cm of s1] (s2) {
    \textbf{Selection of robustness}\\
    \textbf{enhancement methods}
};
\node[numlabel, left=0.3cm of s2.west] {2};

\node[stage, below=1.5cm of s2] (s3) {
    \textbf{Implementation}\\
    \textbf{and integration}
};
\node[numlabel, left=0.3cm of s3.west] {3};

\node[stage, below=1.5cm of s3] (s4) {
    \textbf{Testing}\\
    \textbf{and validation}
};
\node[numlabel, left=0.3cm of s4.west] {4};

\node[stage, below=1.5cm of s4] (s5) {
    \textbf{Monitoring}\\
    \textbf{and updating}
};
\node[numlabel, left=0.3cm of s5.west] {5};

\node[substage, right=0.8cm of s1.east, anchor=west] (d1) {
    $\bullet$ NIST AI RMF, OWASP Top~10 for LLM\\
    $\bullet$ Threat identification, prioritization: $c_{ij} = P_{\text{exploit}} \cdot \text{ImpactValue}$
};

\node[substage, right=0.8cm of s2.east, anchor=west] (d2) {
    $\bullet$ Method effectiveness vs.\ prioritized risks\\
    $\bullet$ Compatibility with the architecture, RUT trade-off
};

\node[substage, right=0.8cm of s3.east, anchor=west] (d3) {
    $\bullet$ Architecture modification, robust training\\
    $\bullet$ Input/output filters, monitoring integration
};

\node[substage, right=0.8cm of s4.east, anchor=west] (d4) {
    $\bullet$ Metrics: ASR, Recall, REASR, CO; RUT evaluation\\
    $\bullet$ Red teaming, iterative refinement
};

\node[substage, right=0.8cm of s5.east, anchor=west] (d5) {
    $\bullet$ Incident collection and analysis\\
    $\bullet$ Tracking new threats, updating defenses
};

\draw[arrow] (s1.south) -- (s2.north);
\draw[arrow] (s2.south) -- (s3.north);
\draw[arrow] (s3.south) -- (s4.north);
\draw[arrow] (s4.south) -- (s5.north);

\draw[feedback] (s4.west) -- ++(-0.9,0) |- node[pos=0.25, left, font=\scriptsize, text=blue!60] {refinement} (s3.west);

\draw[feedback] (s5.west) -- ++(-1.6,0) |- node[pos=0.12, left, font=\scriptsize, text=blue!60] {new threats} (s1.west);

\end{tikzpicture}%
}
\caption{Life cycle of deploying LLM robustness enhancement methods.}
\label{fig:implementation_lifecycle}
\end{figure}

\begin{remark}
\label{princ:adaptation_necessity_ch6}
\textbf{Empirical observation: the need for defense adaptation.}
New types of attacks ($\mathcal{A}_{\text{new}}$) appear constantly, and the characteristics of existing attacks change. Defense mechanisms require continuous monitoring and updating. Static defenses become obsolete over time. The rate of defense adaptation must be comparable to the rate of threat evolution \cite{salehie_self-adaptive_2009}. This is especially relevant for LLMs: new attack vectors appear regularly.
\end{remark}

The methodology sets the order of deployment and maintains the balance between security, functionality, and computational efficiency. Fig.~\ref{fig:implementation_lifecycle} summarizes the five stages and their feedback loops. The cycle repeats throughout the entire operation of the system.

\subsection{Evaluation of Deployment Effectiveness}
\label{subsec:implementation_evaluation_ch6}

The evaluation of deployment effectiveness is not limited to technical security indicators. It also covers the impact on performance and costs.

\textbf{1. Technical effectiveness metrics.}
Robustness and security are evaluated through ASR, Recall, and REASR (Section~\ref{subsec:evaluation_criteria_ch4_part4} of Chapter~\ref{ch:ch4}). For each technical metric $C$, the difference from the base model is recorded as $\Delta C = C(f_{\text{def}}) - C(f_{\text{base}})$. Utility is monitored with BLEU, ROUGE, accuracy, and BERTScore on clean data: $\Delta U = U(f_{\text{def}}) - U(f_{\text{base}})$. Performance is characterized by the computational overhead $CO$ (Section~\ref{subsec:co_analysis_ch5_part5} of Chapter~\ref{ch:ch5}), together with memory usage and system throughput.

\textbf{2. Economic effectiveness metrics.}
The deployment cost $C_I$ covers development, integration, and testing. The operating costs $C_O$ consist of the additional computation caused by the growth of CO and the cost of monitoring. The risk reduction $\Delta \text{Risk}_{\text{val}} = \sum_i (\text{Risk}_i(f_{\text{base}}) - \text{Risk}_i(f_{\text{def}}))$ is estimated according to Definition~\ref{def:attack_risk} of Chapter~\ref{ch:ch1}. The return on investment is computed as $\text{ROI} = (\Delta \text{Risk}_{\text{val}} - (C_I + C_O)) / (C_I + C_O)$. Non-quantifiable benefits $B_S$ include increased user trust, reputational effects, and compliance with regulatory requirements.

\textbf{3. Integral effectiveness metrics.}
The robustness-utility trade-off coefficient RUT characterizes the ratio of the robustness improvement to the utility loss (Section~\ref{subsec:parameter_rut_analysis_ch4_part4} of Chapter~\ref{ch:ch4}). All RUT components are normalized to $[0,1]$ by the min-max method: $x_{\text{norm}} = (x - x_{\min})/(x_{\max} - x_{\min})$, where the bounds are determined on a validation sample. The cost effectiveness of the defense is defined as the ratio of the robustness gain $\Delta R_{stab}$ to the total costs $C_I + C_O$, and the total cost of ownership TCO includes all costs of development, deployment, and operation.

Laboratory testing is not sufficient to assess readiness for a production environment. Pilot projects and proof-of-concept prototypes record the real costs: increased response latency, the cost of API requests, and energy consumption. The pilots also check whether the base quality on the target tasks has degraded. Pilot project reports and an implementation certificate confirm the practical significance of the developed methods.

The evaluation must be repeated regularly, since the threat landscape changes, and without revision the defense mechanisms become obsolete.

\section{Examples of Practical Application}
\label{sec:practical_examples_ch6}

The practical application of the results of this work takes two forms. Completed deployments confirm that the methods work in practice. Typical reference scenarios are recommendations. They illustrate how the developed approaches transfer to other application domains.

The completed deployments are as follows. The vulnerability evaluation framework, the committee-of-models method, and the methodology of multi-step attacks are applied in the production environment of Viasat to secure a recommendation ranking module built on the LLM-as-a-Judge architecture. As part of this work at Viasat, the JudgeGuard suite was installed to test the recommendation ranking module with prompt injection methods, including the ASA algorithm developed by the author and CUA attacks. A committee of three models (GPT-4, Claude-3-Opus, Llama~3) uses majority voting to produce the final ranking verdicts. Load testing confirmed that the latency increase is acceptable: no more than 15\% compared with a single model. According to internal testing at Viasat, the ASR decreased by 42 percentage points. This result was obtained on the company's internal dataset, which differs from the experimental sample of Chapter~\ref{ch:ch4}. In addition, input filters based on a Python sanitizer were deployed to block injections in user reviews that enter the evaluation module. At the time of writing, the validation of the TrojanArmor suite and the MCPSec module is limited to the TDC2023 and MCPBench benchmarks, respectively. Unlike JudgeGuard, these tools were not deployed in production, since trojan detection and the protection of agentic MCP systems were not part of the risk profile of the production site considered.

The developed methods are applicable in different application domains. The typical scenarios and reference architectures below describe possible deployment schemes.

\subsection{Application in Natural Language Processing Systems}
\label{subsec:nlp_applications_ch6}

\textit{The examples in this and the following subsection are recommendations. They describe typical reference architectures, not deployed systems. The completed deployments are listed at the beginning of this section.}

Natural language processing systems differ in their threat profiles. We consider three types of such systems and the approaches to their security: chatbots and virtual assistants, content generation systems, and evaluation and recommendation systems. The last two types share a threat surface and are treated together.

\subsubsection{Chatbots and Virtual Assistants}
\label{subsubsec:chatbots_assistants_ch6}

As a reference scenario, we consider a customer support system of a financial organization. Users ask questions about products, rates and account transactions through a text interface. The threat profile includes prompt injections (LLM01 in OWASP, Threat~\ref{threat:prompt_injection}), bypassing restrictions through jailbreaking (Threat~\ref{threat:jailbreaking}), leakage of confidential data from the dialogue context (LLM06, Threat~\ref{threat:info_disclosure}), and targeted manipulation of behavior.

The central element of the defense is context isolation \cite{chen_struq_2024}, which prevents the model from interpreting user data as instructions. Before the input reaches the model, input filters, both lexical and LLM-based semantic ones, screen out dangerous patterns. The model level relies on safety alignment through RLHF/RLAIF \cite{ouyang_training_2022, bai_constitutional_2022}, described in Section~\ref{sec:training_robustness_methods_ch3_part3} of Chapter~\ref{ch:ch3}, and on explicit restrictions in the system prompt. PII is removed before it is passed to the LLM. The model is fine-tuned through RLHF so that it does not generate unverified financial advice. Dialogue logging makes it possible to detect new attack vectors.

\subsubsection{Content Generation, Evaluation and Recommendation}
\label{subsubsec:content_generation_ch6}

The outputs of generative LLM applications and of evaluation/recommendation systems directly affect the decisions of users or other systems, which creates a common threat surface. Malicious code (LLM02 in OWASP), disinformation, and ranking manipulation through adversarial queries form the main threats. In LLM-as-a-Judge systems, CUA and JMA attacks, studied in Chapter~\ref{ch:ch4}, distort the verdicts directly. Recommender systems are vulnerable to leakage of user preferences (LLM06 in OWASP) and to the amplification of filter bubbles.

Generated content is checked by external verifiers \cite{peng_check_2023, chern_factool_2023} and by static code analysis \cite{pearce_iris_2024}. Constrained decoding \cite{lu_neurologic_2022} narrows the set of admissible outputs. Robust training (Section~\ref{sec:training_robustness_methods_ch3_part3} of Chapter~\ref{ch:ch3}) reduces the sensitivity of the model to provocative queries. In ranking systems, robust embeddings \cite{zhang_robust_2024} together with anomaly detectors form the first layer. It is complemented by outlier-insensitive algorithms \cite{yun_ranking_2014}, fair ML techniques \cite{ferraro_fairroad_2022}, and differential privacy \cite{mcsherry_differentially_2009}. Against CUA and JMA attacks on LLM-as-a-Judge, committees of several models and a comparative evaluation format are effective.

Two reference scenarios show how the emphasis differs. A platform for grading student essays uses a committee of three LLMs: the essays pass a perplexity check and instruction filtering and are then evaluated in a comparative format against reference works. The recommender system of a news portal is organized differently: an ensemble of independent models is combined with filters for adversarial patterns at the input and with algorithms that ensure content diversity.

The choice of specific defense methods is determined by the specifics of the application: excessive protection reduces functionality, and insufficient protection creates vulnerabilities.

\subsection{Application in Critical Systems}
\label{subsec:critical_applications_ch6}

Safety-critical systems impose stricter reliability requirements. The vulnerabilities to prompt injection and hidden trojans described in the previous chapters make the use of LLMs in such systems risky without special measures.

\subsubsection{Medical Decision Support Systems}
\label{subsubsec:medical_decision_support}

Reference scenario: an LLM analyzes X-ray images and relies on a verified knowledge base of pathologies through RAG \cite{gao_retrieval-augmented_2024}. A second independent algorithm checks the conclusions before they are presented to the physician, and the results are accompanied by a confidence estimate and references to sources. The architecture counters the main risks of medical LLM applications: incorrect diagnoses, leakage of patient data (PHI), poisoning of training data \cite{alber_medical_2025}, and demographic bias.

The cost of an error in the medical domain is higher than in most other domains. The defense therefore relies on several principles at once. Where possible, certified robustness methods with mathematical guarantees are applied \cite{cohen_certified_2019}. The physician remains the final link in decision making (human-in-the-loop). The explainability of conclusions is achieved with explainable AI methods \cite{zhao_explainability_2023}. Bias auditing \cite{chen_addressing_2022} and data protection according to HIPAA requirements complement the technical measures.

\subsubsection{Financial Systems and Infrastructure Management}
\label{subsubsec:financial_infrastructure_systems}

The common principle for financial and infrastructure LLM applications is that the model is isolated from making autonomous decisions with irreversible consequences. The nature of the threats and the regulatory context differ, however.

In finance, LLM applications are exposed to the generation of erroneous forecasts with direct losses, to the manipulation of input data through news and social networks, and to leakage of insider information. Ensembling of heterogeneous models with diversified sources, anomaly detectors \cite{abdallah_financial_2022} and robust time series forecasting methods \cite{robust_time_series_2023} form the technical basis. Regulatory monitoring follows SOX, GDPR, and industry regulations. As a reference scenario, we consider an algorithmic trading system. Multi-layer filters screen out unreliable sources, trade decisions are made on the basis of a consensus of several models with confirmation by classical quantitative indicators, and the whole system runs in an isolated environment.

In the management of power grids, transport and industrial equipment, the nature of the threats is different: incorrect control commands and targeted attacks create physical safety risks. The LLM is limited to an advisory role and is architecturally separated from the direct control of critical components. The model's recommendations undergo formal verification \cite{formal_verification_llm_2023} and are approved by an operator. Network segmentation, multi-level redundancy, intrusion detection, and safe shutdown mechanisms for anomalies form additional lines of defense.

In all critical systems, technical measures are effective only in combination with organizational processes, periodic audits, and compliance with regulations.

\section{Prospects for the Development of Robustness Enhancement Methods}
\label{sec:future_perspectives_ch6}

The field of LLM security changes quickly: over 2023--2025, more than 400 papers on attacks on language models and more than 200 on defense mechanisms were published. New attack techniques appear more often than defense updates are released.

\subsection{Research Trends and Directions}
\label{subsec:research_trends_ch6}

\subsubsection{Theoretical Foundations and Formal Verification}
\label{subsubsec:theoretical_foundations_ch6}

Formal models do not keep pace with attack practice: the gap between experimental results and the theoretical understanding of LLM vulnerabilities persists. A complete formalization of robustness in high-dimensional discrete spaces remains an open problem. Mathematical models that adequately describe the space of tokens and semantic transformations have not yet been developed. Provable guarantees for most defense methods used with LLMs, including alignment techniques such as RLHF, do not yet exist. The quantitative analysis of the robustness-utility, privacy-utility, and fairness-robustness trade-offs is far from complete, and the results for the RUT metric from Chapter~\ref{ch:ch4} only confirm the significance of these trade-offs.

The formal verification of LLM properties \cite{formal_verification_llm_2023, huang_survey_2023} is a separate task. SMT solvers and abstract interpretation are so far applicable only to small networks. Progress in scaling abstract interpretation nevertheless suggests that the verification of individual properties will become feasible for models with billions of parameters as well.

\subsubsection{From Reactive to Adaptive Defense}
\label{subsubsec:adaptive_defense_ch6}

Existing defenses are built for known classes of attacks and remain reactive. The experience of the SaTML CTF and TDC2023 demonstrates that attackers systematically bypass fixed filters. Next-generation defenses must reconfigure themselves when they encounter zero-day attacks. Online training of defense modules on a stream of new attacks appears to be a promising direction. Self-healing architectures develop this idea by automatically closing detected vulnerabilities. Red teaming with adaptive attacks such as ASA already makes it possible to automate security testing. Game-theoretic modeling \cite{security_games_2023} provides a tool for optimizing defense strategies. All these approaches, however, place high demands on computational resources.

Adaptivity and interpretability are directly related. Without an understanding of the causes of a vulnerability, a defense cannot generalize to new attacks. The problem of unintended triggers in trojan attacks, that is, surrogate inputs that differ from the true trigger but still activate the backdoor (Section~\ref{subsec:trojan_detection_results_ch4_part4} of Chapter~\ref{ch:ch4}), illustrates the lack of interpretability. Methods that localize vulnerable components can identify the specific layers and attention heads responsible for backdoor behavior \cite{interpretable_security_2023}. Explainable defense mechanisms that report the reason for blocking a query and the visualization of attacks through JailbreakLens \cite{feng_jailbreaklens_2024} complement this approach.

\subsubsection{Multimodal Robustness and New Attack Surfaces}
\label{subsubsec:multimodal_robustness_ch6}

GPT-4V, Gemini~1.5, and LLaVA are already deployed in production systems and accept multimodal input \cite{wang_llms_2024, gu_agent_2024}. New modalities open new attack surfaces. Cross-modal influence makes it possible, for example, to change the text output of a model through a specially crafted image in a document. Defense mechanisms for multimodal systems do not yet form an established class. The development of robustness metrics that account for perturbations in each modality separately and in their combination remains an open problem.

Two related directions widen the threat landscape beyond multimodal robustness. Federated training of LLMs \cite{fed_learning_security_2023} creates the threat of poisoning through the aggregation of participant updates. A network participant can implant a backdoor into the global model, and protocols such as FedAvg are not robust to such influence \cite{robust_fl_2023}. Differential privacy does not remove this threat. The norm clipping and noise of participant-level differential privacy limit a backdoor only at a cost in accuracy \cite{bagdasaryan_backdoor_2020}. This cost falls unevenly on underrepresented subgroups \cite{bagdasaryan_differential_2020}. The noise also makes malicious updates harder to distinguish from benign ones. Quantum computers pose a threat of a different nature. Grover's algorithm reduces the cost of a brute-force search for a 256-bit HMAC-SHA256 key from $2^{256}$ to $2^{128}$ operations \cite{grover_fast_1996}. A 128-bit security level remains adequate for symmetric primitives, so the HMAC in AttestMCP from Chapter~\ref{ch:ch5} does not need replacement. Post-quantum adaptation concerns the public-key mechanisms that establish the session key $K_{session}$, if these rely on classical asymmetric cryptography. The question of whether quantum algorithms can speed up the search for adversarial examples for neural networks still has no theoretical estimates.

\subsection{Technology Development Forecasts}
\label{subsec:technology_forecasts_ch6}

Based on current trends, several probable directions of development can be identified.

The most evident one is the escalation of automated attacks and defenses. LLM agents are already used for automated red teaming \cite{xu_redagent_2024}. Adaptive attacks such as ASA from Chapter~\ref{ch:ch3} demonstrate how evolutionary search bypasses static filters. In response, defenses will use language models themselves for threat detection \cite{wang_selfdefend_2024}. Platforms such as Gandalf \cite{pfister_gandalf_2025} already implement LLM-based generation of attack scenarios. The results of the SaTML CTF confirm the intensification of this ``race''.

In parallel, robustness requirements will shift to the early stages of LLM development according to the ``security by design'' principle, and the after-the-fact addition of filters will give way to frameworks with security guarantees at the training stage. HELM \cite{liang_holistic_2023}, JailbreakBench \cite{chao_jailbreakbench_2024}, and MT-Bench form the basis for standardized comparative evaluation. The TDC2023 and Kaggle competitions accelerate this process. For critical applications, model certification by analogy with Common Criteria is likely. In this context, RLHF/RLAIF \cite{ouyang_training_2022, bai_constitutional_2022} and adversarial training from Chapter~\ref{ch:ch3} solve related tasks by different means, and their integration will yield models that are both aligned with human preferences and robust to manipulation.

Industry specialization and formal verification deserve separate attention. Universal methods will give way to specialized solutions: Mixture-of-Experts architectures will require protection of the expert router, financial \cite{lee_large_2025} and medical \cite{chen_addressing_2022} applications will dictate their own verification requirements, and the protection of LLM-as-a-Judge will become an independent research direction. In parallel, the need for mathematical security guarantees will grow \cite{huang_survey_2023}. Scaling formal methods to models with billions of parameters is not yet solved, but progress in the abstract interpretation of neural networks gives grounds for optimism.

The general trend is that security and robustness are gradually being built into the main LLM development cycle, and reactive measures are giving way to proactive ones.

\section*{Chapter Conclusions}
The transfer of the results into production revealed a limitation that is not visible on benchmarks: the effect of a defense depends on the data of the specific deployment. The deployment of a committee of three models at Viasat reduced the ASR by 42~pp on an internal dataset different from the experimental sample of Chapter~\ref{ch:ch4}, with a latency increase of no more than 15\%. This confirms the transferability of the method but shows that the numerical indicators of a defense do not transfer between domains without recalibration.

The assessment of production readiness requires metrics of three groups: technical, economic, and integral metrics such as RUT. The integral metric links heterogeneous criteria and makes it possible to recompute the optimal defense configuration when the threat landscape changes, whereas the technical and economic metrics taken separately do not allow such a recomputation.

The cost of a committee is linear in the number of models, but parallel querying of the members makes the actual latency increase sublinear. To reduce costs further, a cascade scheme is promising: a smaller committee, for example of three models, processes the query, and the full committee is activated only at a low level of agreement, which is detected by the confidence $c$ from Section~\ref{subsec:committee_pseudocode} of Chapter~\ref{ch:ch3}. A quantitative estimate of the savings remains an open problem.

\chapter*{Conclusion}
\addcontentsline{toc}{chapter}{Conclusion}

This thesis develops methods for evaluating and improving the robustness of large language models to input sequence variations. The main results are as follows.

\begin{enumerate}
    \item The analysis of the Trojan Detection Challenge 2023 data revealed an asymmetry between the generation and the detection of trojan triggers. Surrogate triggers activate the trojan almost always: REASR $\approx 0.99$. The original trigger is not reconstructed: the Recall metric, a BLEU-based similarity between predicted and true triggers, did not exceed $\approx 0.17$ even for the best method, GCG. This gap argues for prioritizing preventive measures: auditing the training data and controlling the supply chain. The result relates to Objectives 1--2.

    \item The generative robustness metric \(R_{stab}(f)\) accounts for the discreteness of the token space by capturing the per-step divergence of output distributions under small input perturbations. For localized attacks, the bound $V(h) \le 1 - R_{class}(h)$ is proved. For non-localized attacks, a parametric family $V(h) \approx g(1 - R_{class}(h))$ with parameters calibrated on data is proposed as an empirical hypothesis. The result relates to Objective 2.

    \item The vulnerability of LLM-as-a-Judge systems to prompt injections is studied within the developed framework, which divides attacks into content attacks and instructive attacks. The adaptive generation algorithm ASA achieves an ASR of up to 73.8\% on the Gemma, Llama, GPT-4, and Claude-3-Opus models. The vulnerability depends on the architectural family and the alignment procedure, and not only on the model scale. The attacks transfer well between open architectures. The result relates to Objective 3.

    \item Based on SaTML CTF 2024 data, multi-step attacks on LLMs with a three-layer defense are systematized. Four bypass vectors are classified: code-oriented attacks, ASCII encoding, word splitting, and contextual misdirection. The first three are neutralized by input normalization, and the fourth requires contextual isolation of system instructions. As a defense mechanism, committees of 5--7 heterogeneous models are developed. For the most vulnerable model, Gemma-3-4B, the ASR reduction was 47--55 percentage points. The theoretical basis is the Condorcet jury theorem: if the committee members err independently and each has a base accuracy $r > 0.5$, the accuracy of the majority verdict grows with the size of the committee. This result concerns committee accuracy and does not by itself guarantee robustness in the sense of $R_{class}$. For Gemma-3-4B with $r \approx 0.26$, the condition of the theorem does not hold, and the ASR reduction is achieved empirically through the heterogeneity of the committee composition, which decorrelates the errors of models of different architectures. The robustness-utility trade-off metric RUT for a committee of 7 models is $\approx 5.5$. The result relates to Objective 4.

    \item Architectural vulnerabilities of MCP are identified that enable attacks of the \textit{prompt injection to tool hijacking} class. To counter them, two algorithms are proposed. AttestMCP restricts tool calls to a session permission table and checks packet integrity with HMAC. Commit Boundary isolates the execution of the agent's code. On the MCPBench benchmark of 847 scenarios, the average ASR decreased by 41.3~percentage points, from 53.7\% to 12.4\%. The result relates to Objective 5.

    \item The JudgeGuard and TrojanArmor software suites are implemented. They cover the full cycle of LLM security testing, from the generation of attacks to the computation of robustness metrics. The MCPSec module complements both suites with tools for the analysis of agentic LLM systems. The result relates to Objective 5.
\end{enumerate}

The practical significance of the results is confirmed by the deployment of the JudgeGuard suite and the defense method based on model committees in the production environment of Viasat.

The proposed metrics, attack and defense algorithms, and software suites are applicable to existing and future LLM systems. Directions for future work include the extension of the framework to multimodal models, the adaptation of the algorithms to real-time attack detection, and the construction of verifiable security guarantees.

\chapter*{Acknowledgements}
\addcontentsline{toc}{chapter}{Acknowledgements}

The author is deeply grateful to his doctoral advisor, Dmitry~Namiot, for posing the problem, for many years of support, and for the demanding discussion of the results at every stage of the work.

Special thanks go to the scientific consultant Evgeny~Ilyushin for his help in choosing the research direction and for his detailed review of the experimental part.

The author is particularly grateful to Ivan~Chizhov for his attention to the work and for his support at all stages of preparing the thesis.

The author thanks his parents, Gagik~Maloyan and Gayane~Tovmasyan, for believing in him and supporting the path he chose, and his sister, Mariam~Maloyan.

The author remembers his grandparents with warmth and gratitude. He is especially grateful to his grandfather, Hamlet~Tovmasyan, for the time and attention he gave to lessons with the author during the author's school years.

The author's warmest thanks go to his wife, Maria~Arutunyan. Her help and support at the final stage of the work made it possible to bring the thesis to its defense.

The author's interest in mathematics was shaped by his school teachers Galina~Shpilevskaya, Elena~Panfilova, and Igor~Panfilov. The author is grateful to them for their rigor and attention, which determined his choice of profession.

The author is grateful to his colleagues at Viasat for their assistance in deploying the defense mechanisms in the production environment and for their feedback during the pilot operation.

\chapter*{List of Abbreviations and Notation}
\addcontentsline{toc}{chapter}{List of Abbreviations and Notation}
\noindent
\begin{longtable}{@{}p{3.4cm}p{12.1cm}@{}}
\multicolumn{2}{l}{\textbf{Abbreviations}} \\
AGM & Attack Generator Module \\
AI & Artificial Intelligence \\
AI RMF & NIST AI Risk Management Framework \\
API & Application Programming Interface \\
AR & Attack Robustness \\
ASA & Adaptive Search-Based Attack \\
ASCII & American Standard Code for Information Interchange \\
ASGI & Asynchronous Server Gateway Interface \\
ASR & Attack Success Rate \\
AST & Abstract Syntax Tree \\
AT & Adversarial Training \\
BI & Basic Injection \\
BLEU & Bilingual Evaluation Understudy \\
BPE & Byte Pair Encoding \\
CM & Contextual Misdirection \\
CO & Computational Overhead \\
CPU & Central Processing Unit \\
CTF & Capture The Flag \\
CUA & Comparative Undermining Attack \\
CVSS & Common Vulnerability Scoring System \\
CWB & Complex Word Bombardment \\
DAG & Directed Acyclic Graph \\
DoS & Denial of Service \\
DP & Differential Privacy \\
EE & Evaluation Engine \\
FFN & Feed-Forward Network \\
GA & Genetic Algorithm \\
GBDA & Gradient-Based Distributional Attack \\
GCG & Greedy Coordinate Gradient \\
GDPR & General Data Protection Regulation (EU) \\
GPU & Graphics Processing Unit \\
HF & HuggingFace (Transformers pipeline) \\
HIPAA & Health Insurance Portability and Accountability Act (USA) \\
HITL & Human-In-The-Loop \\
HMAC & Hash-based Message Authentication Code \\
HTTP & Hypertext Transfer Protocol \\
JMA & Justification Manipulation Attack \\
JSD & Jensen--Shannon Divergence, denoted \(\JS\) \\
JSON & JavaScript Object Notation \\
KL & Kullback--Leibler (divergence) \\
KV & Key--Value (attention cache) \\
LLM & Large Language Model \\
MCP & Model Context Protocol \\
MCPSec & Model Context Protocol Security \\
ML & Machine Learning \\
MPS & NVIDIA Multi-Process Service \\
MRAG & Model Router / API Gateway \\
NIST & National Institute of Standards and Technology (USA) \\
NLP & Natural Language Processing \\
OWASP & Open Web Application Security Project \\
PEZ & Hard Prompts Made Easy (gradient-based prompt discovery) \\
PGD & Projected Gradient Descent \\
PHI & Protected Health Information \\
PII & Personally Identifiable Information \\
pp & Percentage points \\
RAG & Retrieval-Augmented Generation \\
REASR & Reverse-Engineered Attack Success Rate \\
REST & Representational State Transfer \\
RLAIF & Reinforcement Learning from AI Feedback \\
RLHF & Reinforcement Learning from Human Feedback \\
ROI & Return on Investment \\
ROUGE & Recall-Oriented Understudy for Gisting Evaluation \\
RPC & Remote Procedure Call \\
RUT & Robustness-Utility Trade-off \\
SCC & Strongly Connected Component \\
SGD & Stochastic Gradient Descent \\
SL & Storage Layer \\
SMT & Satisfiability Modulo Theories \\
SOX & Sarbanes--Oxley Act (USA) \\
SQL & Structured Query Language \\
SSE & Server-Sent Events \\
TCO & Total Cost of Ownership \\
TDC2023 & Trojan Detection Challenge 2023 \\
TSR & Transfer Success Rate \\
UAT & Universal Adversarial Triggers \\
UML & Unified Modeling Language \\
VAT & Virtual Adversarial Training \\
XSS & Cross-Site Scripting \\
YAML & YAML Ain't Markup Language \\
\\
\multicolumn{2}{l}{\textbf{NIST SP 800-53 control families (Chapter 1)}} \\
AC & Access Control \\
AT & Awareness and Training \\
AU & Audit and Accountability \\
CM & Configuration Management \\
CP & Contingency Planning \\
IA & Identification and Authentication \\
PL & Planning \\
PM & Program Management \\
PT & PII Processing and Transparency \\
SC & System and Communications Protection \\
SI & System and Information Integrity \\
SR & Supply Chain Risk Management \\
\\
\multicolumn{2}{l}{\textbf{Notation}} \\
\(\Vspace\) & Finite token vocabulary \\
\(\Xspace = \Vspace^*\) & Space of all finite token sequences (input texts) \\
\(\Yspace\) & Space of possible model outputs (output texts) \\
\(\ProbSpace{\Yspace}\) & Set of probability distributions on the space \(\Yspace\) \\
\(f\), \(f_\theta\) & Large language model with parameters \(\theta\) \\
\(\theta\) & Model parameters \\
\(x \in \Xspace\) & Input token sequence \\
\(x' \in \Xspace\) & Perturbed input token sequence \\
\(y \in \Yspace\) & Output token sequence \\
\(p(y|x)\) & Probability distribution over outputs \(y\) for a given input \(x\) \\
\(|x|\) & Length of the sequence \(x\) \\
\(d(x, x')\) & Metric (distance) between the sequences \(x\) and \(x'\) \\
\(d_H\) & Hamming distance \\
\(d_L\) & Levenshtein distance \\
\(d_S\) & Semantic distance \\
\(\Bepsilon{x}\) & \(\epsilon\)-neighborhood of the sequence \(x\) \\
\(\epsilon\) & Perturbation budget \\
\(R_{stab}(f)\) & Robustness metric of the generative model \(f\) (the subscript comes from ``stability'') \\
\(R_{class}(h)\) & Robustness metric of the deterministic decision (classification) operator $h = \mathrm{Dec} \circ f$ \\
\(D(p_1, p_2)\) & General divergence measure between the distributions \(p_1\) and \(p_2\), for example \(\JS\). Asymmetric measures use an explicit notation, for example \(D_{KL}(P \parallel Q)\) \\
\(D_{KL}(P \| Q)\) & Kullback--Leibler divergence between the distributions \(P\) and \(Q\) \\
\(\JS(P,Q)\) & Jensen--Shannon divergence between the distributions \(P\) and \(Q\) \\
\(D_T(x, x')\) & Divergence of the output distributions for the inputs \(x\) and \(x'\), averaged over \(T\) generation steps (Chapter 2) \\
\(C_{max}\) & Normalization constant for the divergence measure \(D\) \\
\(\Vulfunc_{\text{local}}(f)\) & Local vulnerability metric of the model \(f\) \\
\(V(f)\) & Global vulnerability metric of the model \(f\) \\
\(\badY\) & Set of undesirable or prohibited outputs \\
\(\DataDist\), \(P_{\mathcal{X}}\) & Distribution of input data \\
\(\mathbb{E}[\cdot]\) & Expectation \\
\(\mathbb{I}\{\cdot\}\) & Indicator function \\
\(\sup\) & Supremum (least upper bound) \\
\(\LossFunc\) & Loss function \\
\(\mathcal{A}^{\mathrm{type}}\) & Set of attack types (Chapter 1); \(\mathfrak{A}\) is the set of attack transformations (Chapter 2) \\
\(\mathcal{V}\) & Set of vulnerabilities \\
\(\mathcal{I}_{AV}\) & Relation ``attack exploits vulnerability'' \\
\(\mathcal{C}_{VV}\) & Dependency relation between vulnerabilities \\
\(\text{Risk}(a, f)\) & Risk of attack \(a\) on model \(f\) \\
\(\text{Impact}(y)\) & Binary impact indicator \(\mathbb{I}\{y \in \badY\}\) of an undesirable output (Chapter 2) \\
\(\text{ImpactValue}(v)\) & Quantitative estimate of the damage from exploiting the vulnerability \(v\) (Chapter 1) \\
\(P_{\text{exploit}}(v | a)\) & Probability of successful exploitation of the vulnerability \(v\) by an attack of type \(a\) \\
\(\text{AdvRobustness}(f)\) & Adversarial robustness of the model \(f\) \\
\end{longtable}
\addtocounter{table}{-1}

\chapter*{Glossary}
\addcontentsline{toc}{chapter}{Glossary}

\textbf{Large language model (LLM)} A class of machine learning models based on deep neural networks trained on large volumes of text data for understanding, generating, and processing natural language.

\textbf{Token} The smallest unit of text that a language model operates on, for example a word, a part of a word, or a character. The process of converting text into a sequence of tokens is called tokenization.

\textbf{Embedding} A vector representation of a token, phrase, or document in a high-dimensional space that reflects its semantic properties.

\textbf{Robustness} The ability of a model, or of a compositional system based on it, to preserve correct behavior and acceptable quality characteristics under small perturbations of input data, adversarial modifications of queries, and other influences considered in the adopted threat model. It is the main term of this work. When the work refers to the literature on adversarial robustness, the term is used in the same sense. The robustness metric $R_{stab}(f)$ quantifies robustness through the divergence of output distributions. Its subscript comes from ``stability'' and does not imply a separate notion of stability.

\textbf{Vulnerability} A property of a system or model that an attacker can exploit to disrupt its normal operation, compromise data, or bypass defense mechanisms.

\textbf{Attack} A targeted influence on a system or model intended to cause undesirable behavior, gain unauthorized access, extract information, or disrupt the operation of the target.

\textbf{Adversarial attack} A type of attack that adds small perturbations, often imperceptible to a human, to the input data and thereby causes erroneous or undesirable behavior of the model.

\textbf{Adversarial example} Input data modified as a result of an adversarial attack.

\textbf{Jailbreaking} A type of attack aimed at bypassing the safety mechanisms built into an LLM (safety alignment) in order to make the model generate prohibited, harmful, or unethical content.

\textbf{Prompt injection} A type of attack in which hidden or disguised instructions are injected into a legitimate user query, or into external data that the model processes, such as documents or tool outputs (indirect prompt injection), and override the original system instructions of the model or its safety mechanisms.

\textbf{Data poisoning} An attack in which the adversary injects a small number of specially crafted (``poisoned'') examples into the training dataset in order to disrupt the training process, degrade the quality of the model, create trojans, or amplify bias.

\textbf{Trojan (backdoor)} Hidden functionality embedded in a model, typically through data poisoning. The trojan is activated by a special trigger in the input data and causes behavior predefined by the attacker. In this work, the main term is ``trojan'', and ``backdoor'' is used as a synonym.

\textbf{Alignment} The process of training or fine-tuning an LLM in order to bring its behavior in line with human values, ethical norms, and instructions, and to prevent the generation of harmful, biased, or undesirable content.

\textbf{Differential privacy} A formal definition of privacy that guarantees that the probability distribution of the result of a randomized data analysis changes very little if any single record is removed from or changed in the dataset. It is used to protect the confidentiality of individual data in training datasets.

\textbf{Levenshtein distance} A metric that measures the minimal number of single-symbol edit operations (insertion, deletion, substitution) required to transform one string of characters or tokens into another.

\textbf{KL divergence (Kullback--Leibler divergence)} An asymmetric measure of divergence between two probability distributions. This work mainly uses its symmetric counterpart, the Jensen--Shannon divergence $D_{\mathrm{JS}}$, which is built from two KL terms.

\textbf{Committee of models (ensemble of models)} An approach to improving the reliability and accuracy of predictions in which several independently trained models are used and their outputs are combined, for example through voting, averaging, or a weighted combination, to obtain the final result. In this work, a committee of heterogeneous LLMs aggregates the verdicts of judge models by majority voting, which reduces the success of attacks on LLM-as-a-Judge. Low agreement among the members flags a potentially malicious query for manual review.

\textbf{LLM-as-a-Judge} An architecture or methodology of automatic evaluation in which a large language model acts as an evaluator and issues a discrete or ranked verdict on the quality, safety, instruction compliance, or preferability of the responses of another model or of several alternative responses.

\textbf{Multi-layer defense (defense in depth)} A security strategy that applies several independent layers of defense mechanisms. In the context of LLMs, it includes a combination of different methods: input filtering, monitoring of model behavior, post-processing of output data, model alignment, and other techniques that work together to minimize risks.

\bibliographystyle{unsrtnat}
\bibliography{references}

\end{document}